\documentclass[journal,  onecolumn]{IEEEtran}
\usepackage{url}
\usepackage{hyperref}
\hypersetup{
  colorlinks   = true, 
  linkcolor    = black!50!brown, 
  citecolor   = black!50!brown, 
}
\usepackage{balance}
\usepackage{verbatim}
\usepackage{bm}
\usepackage{color,graphicx,xcolor}
\usepackage{mdwtab}
\usepackage{mathtools,tikz}
\usepackage{pgfplots}
\pgfplotsset{soldot/.style={color=myblue,only marks,mark=*}} \pgfplotsset{holdot/.style={color=myblue,fill=white,only marks,mark=*}}
\usepackage{hhline}
\usepackage{multirow}
\usepackage{pdfpages}
\usepackage{enumitem}  
\usepackage{caption}
\usepackage{subcaption}
\usepackage{cite}
\usepackage{amsmath,amssymb,amsfonts,amsthm,xspace,bm}
\usepackage[bb=dsserif]{mathalpha}
\usepackage{algorithmic}
\usepackage{graphicx}
\usepackage{textcomp}
\usepackage{todonotes}
\usepackage{etoolbox}
\usetikzlibrary{arrows,shapes.geometric,calc}
\usepackage[rightcaption]{sidecap}
\sidecaptionvpos{figure}{m}
\usepackage{mfirstuc}
\allowdisplaybreaks
\definecolor{bluishgreen}{RGB}{0,158,115}
\definecolor{vermillion}{RGB}{213,94,0}
\definecolor{myblue}{RGB}{0,114,178}
\definecolor{myorange}{RGB}{230,159,0}

\theoremstyle{definition}
\newtheorem{example}{Example}
\newtheorem{remark}{Remark}
\newtheorem{theorem}{Theorem}
\newtheorem{definition}{Definition}

\newtheorem{claim}[theorem]{Claim}

\newtheorem*{definition*}{Definition}
\newtheorem*{lemma*}{Lemma}

\newtheorem{lemma}[theorem]{Lemma}

\newtheoremstyle{thmnum}{\topsep}{\topsep}{\itshape}{0pt}{\bfseries}{.}{ }{\thmname{#1}\thmnote{ \bfseries #3}}
\theoremstyle{thmnum}

\usepackage{pgffor}
\foreach \x in {a,...,z}{%
\expandafter\xdef\csname vec\x \endcsname{\noexpand\ensuremath{\noexpand\bm{\x}}}
}

\foreach \x in {A,...,Z}{%
\expandafter\xdef\csname vec\x \endcsname{\noexpand\ensuremath{\noexpand\bm{\x}}}
}

\foreach \x in {A,...,Z}{%
\expandafter\xdef\csname c\x \endcsname{\noexpand\ensuremath{\noexpand\mathcal{\x}}}
}

\foreach \x in {A,...,Z}{%
\expandafter\xdef\csname bb\x \endcsname{\noexpand\ensuremath{\noexpand\mathbb{\x}}}
}

\foreach \x in {A,...,Z}{%
\expandafter\xdef\csname s\x \endcsname{\noexpand\ensuremath{\noexpand\sf{\x}}}
}

\foreach \x in {J,M,X,Y,Z}{%
\expandafter\xdef\csname h\x \endcsname{\noexpand\ensuremath{\noexpand\widehat{\x}}}
}

\foreach \x in {m,x,y,z}{%
\expandafter\xdef\csname h\x \endcsname{\noexpand\ensuremath{\noexpand\hat{\x}}}
}

\foreach \x in {m,x,y,z}{%
\expandafter\xdef\csname t\x \endcsname{\noexpand\ensuremath{\noexpand\tilde{\x}}}
}
\xdef\tzero{\noexpand\ensuremath{\noexpand\tilde{0}}}
\xdef\tone{\noexpand\ensuremath{\noexpand\tilde{1}}}

\foreach \x in {M,P,X,Y,Z,x,y,z}{%
\expandafter\xdef\csname b\x \endcsname{\noexpand\ensuremath{\noexpand\bar{\x}}}
}

\newcommand{\inb}[1]{\left\{#1\right\}}
\newcommand{\inp}[1]{\left(#1\right)}
\newcommand{\insq}[1]{\left[#1\right]}
\newcommand{\red}[1]{{\textcolor{red}{#1}}}

\newcommand{\pr}{\ensuremath{\textup{P}}}

\newcommand{\decC}{\ensuremath{g_{\sC}}}
\newcommand{\decB}{\ensuremath{g_{\sB}}}

\newcommand{\enc}{\ensuremath{f}}
\newcommand{\ch}{\ensuremath{W}}
\newcommand{\tP}{\ensuremath{\widetilde{P}}}
\newcommand{\pP}{\ensuremath{\bar{P}}}

\newcommand{\err}{\ensuremath{P_\textup{e}}}

\newcommand{\erravg}{\ensuremath{P_\textup{avg}}}
\newcommand{\capstr}{\ensuremath{C_\textup{Byz}}}

\newcommand{\capcom}{\ensuremath{C_\textup{com-msg}}}
\newcommand{\capptop}{\ensuremath{C_\textup{p-to-p}}}

\newcommand{\rad}{\delta}

\newcommand{\indep}{\raisebox{0.05em}{\rotatebox[origin=c]{90}{$\models$}}\xspace}
\def\BibTeX{{\rm B\kern-.05em{\sc i\kern-.025em b}\kern-.08em
    T\kern-.1667em\lower.7ex\hbox{E}\kern-.125emX}}

\newcommand{\representationchannel}{representation channel}
\newcommand{\RepresentationChannel}{\xcapitalisewords{\representationchannel}}

\newtoggle{long}
\toggletrue{long}

\IEEEoverridecommandlockouts

\begin{document}
\title{Consensus Capacity of Noisy Broadcast Channels}
\author{Neha Sangwan,~\IEEEmembership{Member,~IEEE,}
Varun Narayanan,
Vinod M. Prabhakaran,~\IEEEmembership{Member,~IEEE}%
\thanks{This work was presented in part at the 2022 IEEE International Symposium on Information Theory (ISIT).\\N. Sangwan's work was supported in part by the TCS Foundation through the TCS Research Scholar Program. V. Narayanan's work was supported by ERC Project NTSC (742754) and ISF Grants 1709/14 and 2774/20. V. Prabhakaran's work was supported in part by SERB through project MTR/2020/000308. N. Sangwan and V. Prabhakaran acknowledge support of the DAE under project no. RTI4001.\\
N. Sangwan is with the University of California, San Diego. She was with the School of Technology and Computer Science, Tata Institute of Fundamental Research, Mumbai 400005. Varun Narayanan is with the University of California, Los Angeles. He was with the Technion, Israel. Vinod Prabhakaran is with the School of Technology and Computer Science, Tata Institute of Fundamental Research, Mumbai 400005.}}
\maketitle
\begin{abstract}
We study communication with consensus over a broadcast channel -  the receivers reliably decode the sender's message when the sender is honest, and their decoder outputs agree even if the sender acts maliciously. We characterize the broadcast channels which permit this byzantine consensus and determine their capacity. We show that communication with consensus is possible only when the broadcast channel has embedded in it a natural ``common channel'' whose output both receivers can unambiguously determine from their own channel outputs. Interestingly, in general, the consensus capacity may be larger than the point-to-point capacity of the common channel, i.e., while decoding, the receivers may make use of parts of their output signals on which they may not have consensus provided there are some parts (namely, the common channel output) on which they can agree.
\end{abstract}
\section{Introduction}

The question of how communication can be carried out when the communicating agents do not trust each other has received considerable attention in the distributed computation and cryptography literatures~\cite{LamportSP82,Dolev82,Lynch1996}. Lamport, Shostak and Pease, in their seminal work, formulated the so-called byzantine generals problem~\cite[pg.~384]{LamportSP82}\cite{Dolev82}, where a commanding general (sender node) wants to communicate a message to a set of lieutenant generals (other nodes) such that, 
\begin{enumerate} 
\item[(i)] if the commander is honest, all honest lieutenants agree on the commander's message, and
\item[(ii)] all honest lieutenants agree on the same message even if the commander is malicious.
\end{enumerate} 
They showed that with three nodes (a commander and two lieutenants), this is impossible to achieve when the nodes communicate over private pairwise communication links~\cite{LamportSP82,Dolev82,FischerLM86}. For the general case, the impossibility holds when at least one-third of the nodes may collude and act maliciously. There has been a renewed interest in this problem because of applications in blockchains~\cite{CheMic19,Shi2020}.

In this work, we consider communication with consensus over 
the broadcast channel\footnote{We use the term broadcast channel in the sense it is used in network information theory~\cite{Cover72,ElGamalK2011}, where it refers to a potentially noisy channel with a single sender and multiple receivers. In cryptography, the term generally refers to the noiseless special case.}~\cite{Cover72,ElGamalK2011}. We require the following:
\begin{enumerate}
\item[(i)] When the sender is honest, the receivers must reliably decode the sender's message.
\item[(ii)] Even if the sender acts maliciously, the receivers' decoder outputs must agree.
\end{enumerate}
In the latter case, no correct decoding is demanded; indeed a malicious sender need not have a message in mind while crafting its attack. The problem may be thought of as a common message transmission problem~\cite{KornerM77} over broadcast channels with the additional stipulation of consensus among receivers even when the sender deviates.
%
We address the following questions:
\begin{center}
\parbox{0.9\columnwidth}{Which broadcast channels allow byzantine consensus?\\ And when consensus is possible, what is the capacity of communication with consensus?}
\end{center}

There is an extensive literature in information theory on communication in the presence of external adversaries, both passive~\cite{Wyner75,CsiszarK78} and active~\cite{BlackwellBT59,BlackwellBT60,Wolfowitz59,CsiszarN88} (also see surveys~\cite{BlochB11,PoorS17,LapidothN98}). More closely related to the present work are those on communication when the users are byzantine~\cite{JaggiLKHKM07,HeY13,KosutTT14,KosutK16,KosutK18,SangwanBDP19a,SangwanBDP24a,SangwanBDP24b}. Our setup can also be thought of as one in a line of works in cryptography which use stochastic resources (channels and sources) not controlled by the users to realize, with information theoretic security, cryptographic tasks such as privacy amplification~\cite{BennettBR88,Maurer91,AhlswedeC93,Maurer93,BennettBCM95,CsiszarN00,Maurer03a,Maurer03b,CsiszarN04}, oblivious transfer (and secure computation, in general)~\cite{CrepeauK88,CrepeauMW05,Wullschleger09,IshaiKOPSW11}, and commitment~\cite{DamgardKS99,WinterNI03,RanellucciTWW11}. The works which are closest to the present work are~\cite{FitziWW04,NarayananPSW23} which study communication with byzantine consensus when the users have access to a distributed source and they are connected by private noiseless pairwise communication links. The model here differs in two respects -- instead of a distributed source, we consider a noisy broadcast channel, and there are no private links between the users. The only means of communication is via the broadcast channel. Thus, ours is a non-interactive one-way setup like~\cite{GargIKOS15,AgrawalIKNPPR21}. This precludes the type of protocols used in~\cite{FitziWW04,NarayananPSW23}.



We show that communication with consensus is possible only when the broadcast channel has embedded in it a natural ``common channel'' whose output both receivers can unambiguously determine from their own channel outputs.
Interestingly, in general, the consensus capacity may be larger than the point-to-point capacity of the common channel, i.e., while decoding, the receivers may make use of parts of their output signals on which they may not have consensus provided there are some parts (namely, the common channel output) on which they can agree. 
A natural upper bound to consensus capacity is the (non-byzantine) common message capacity~\cite{KornerM77}.
This turns out to be loose in general, see~Figure~\ref{fig:bsc}.

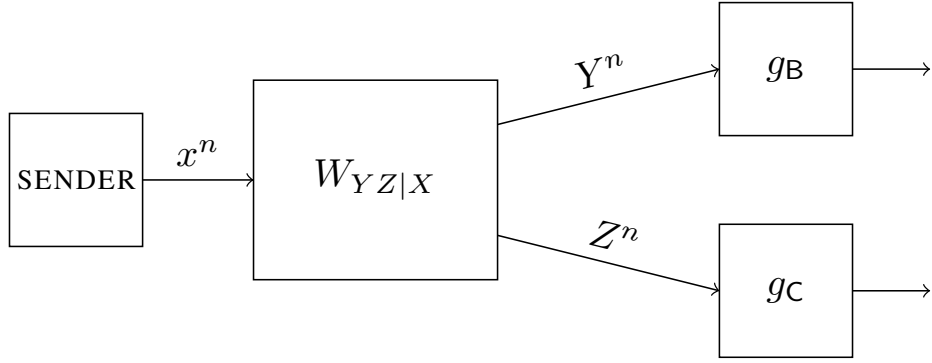
\begin{figure}[tb]
\centering
\resizebox{0.7\columnwidth}{!}{
\begin{tikzpicture}
    \draw (8.5,2.5) rectangle ++(2.2,1.8) node[pos=.5]{$W_{YZ|X}$};
    \draw[->] (7.5,3.4)  -- node[above]{$x^n$} ++ (1,0);
    \draw (7.5-1.2,3.4-0.6) rectangle ++(1.2,+1.2) ;

    \node at (7.5-0.6,3.4) {\small\sc{sender}} ;
    \draw[->] (10.7,3.9) -- node[midway,above,sloped] {$Y^n$} ++ (2,0.8-0.3);
    \draw[->] (10.7,2.9) -- node[midway,above,sloped] {$Z^n$} ++ (2,-0.8+0.3);
    \draw (12.7,4.1-0.3) rectangle ++(1.2,1.2) node[pos=.5]{$\decB$};
    \draw (12.7,2.7+0.3) rectangle ++(1.2,-1.2) node[pos=.5]{$\decC$};
    \draw[->] (12.7+1.2,4.1+0.6-0.3) -- ++ (0.7, 0);
    \draw[->] (12.7+1.2,2.7-0.6+0.3) -- ++ (0.7, 0);

\end{tikzpicture}
}
\caption{For any input $x^n$, the outputs of the decoders $\decB$ and $\decC$ must agree; furthermore, if $x^n$ is the codeword for some message $m$, then the decoders must output $m$ (both conditions need to hold with high probability).}
\label{fig:blockdiagram}
\end{figure}

The paper is organized as follows. In Section~\ref{sec:setup} we formally set up the problem and define the notion of {\em common channel} of a broadcast channel which will play an important role in the rest of the paper. We consider the special case of a ``two-step'' binary erasure channel in Section~\ref{sec:short} to illustrate the key ideas behind the proof of our consensus capacity theorem; the theorem itself is presented in Section~\ref{sec:general} and its proof in Section~\ref{sec:proof}. The paper concludes with a discussion of generalizations and open problems.

\nottoggle{long}{\vspace{-.225cm}}{}
\section{Notation}\label{sec:notation} See~\cite{CoverT2006} for definitions of information theoretic quantities such as mutual information, entropy and KL divergence. These quantities are defined in logarithm base 2. We mostly follow the notation from ~\cite{CoverT2006}. We employ the method of types in some of our proofs for which we adopt the notation from~\cite{Csiszar98}. 
Random variables are denoted by capital letters like $X, X', Y,$ etc. The corresponding alphabets are denoted by calligraphic letters in the same format, for example, the random variables $X$ and $X'$ have alphabet $\cX$. Its $n-$product set is denoted by $\cX^n$. $x^n$, $y^n$ denote vectors in $\cX^n$ and $\cY^n$ respectively. 
For an alphabet $\cX$, let $\mathcal{P}^n\inp{\cX}$ denote the set of all empirical distributions (types) of $n$ length strings from $\cX^n$. For a random variable $X$, we denote its distribution by $P_X$ and use the notation $X\sim P_{X}$ to indicate this. If $P_X \in \mathcal{P}^n\inp{\cX}$, we use $\cT^n_X$ to denote the set of all sequences with empirical distribution specified by $P_X$.  If $x^n\in \cX^n$ has empirical distribution $P_X$, we say $x^n$ is of type $P_X$ and write $x^n\in \cT^n_X$.  When $P_X$ is not already defined, note that we write $x^n\in \cT^n_{X}$ to implicitly define the type $P_X$ {associated with $\cT^n_X$ to be the } empirical distribution of $x^n$. {For a broadcast channel $W_{YZ|X}$, we denote its marginal channels to the receivers by $W_{Y|X}$ and $W_{Z|X}$ respectively.} For a channel $W$, its $n$-fold product (memoryless use) is denoted by $W^n$. For any number $a$, we will use $\exp{a}$ to denote $2^a$ and $\log{a}$ to denote $\log_2{a}$. All information theoretic quantities (KL divergence, entropy and mutual information) are in base 2. 

\section{Setup and Preliminaries}\label{sec:setup}

Consider a two-receiver\footnote{Our focus is on the {\em two}-receiver model. However, the results generalize as we discuss in Section~\ref{sec:disc}.} memoryless broadcast channel $\ch_{YZ|X}$ from a sender (Alice) with input alphabet $\cX$ to receivers, Bob and Carol, resp., with output alphabets $\cY$ and $\cZ$, resp. We consider finite alphabets. An $(n,K)$ {\em consensus code} consists of:
\begin{enumerate}
	\item[(i)] an encoder: $\enc:[1:K] \rightarrow \cX^n$, and
	\item[(ii)] decoders: $\decB:\cY^n \rightarrow [1:K]\cup\{\perp\}$ \&
		$\decC:\cZ^n \rightarrow [1:K]\cup\{\perp\}$.
\end{enumerate} The rate of the encoder is $\log(K)/n$.
The encoder and the decoders are deterministic; we comment on this and other choices we make in setting up the problem in Section~\ref{sec:setup-remarks}.
\paragraph{Error probability} 
An {\em error} is said to occur when either of the following conditions (or both) hold: 
\begin{enumerate}
	\item[(i)] the outputs of the decoders do not match\footnote{Note that $(\decB(Y^n),\decC(Z^n))=(\bot,m)$, $m\in[1:K]$ counts as an error.} (i.e., $\decB(Y^n)\neq \decC(Z^n)$) irrespective of what the sender transmitted;
	\item[(ii)] if the sender transmitted the codeword $f(m)$ corresponding to a message $m\in[1:K]$ and the output of at least one of the decoders does not match the message $m$. 
\end{enumerate}
We will refer to a sender whose transmission is not from the codebook as a {\em malicious sender}.
We define
\begin{align*}
	\lambda_m &= 1 - \Pr\left(\decB(Y^n)=\decC(Z^n)=m|f(m)\right),\,m\in[1:K]\\
	\eta_{x^n} &= \Pr\left(\decB(Y^n)\neq \decC(Z^n)|x^n\right),\;x^n \in \cX^n,
\end{align*}
where we use the shorthand notation $\Pr(.|x^n)$ to denote $\Pr(.|X^n=x^n)$.
Let
\[	\lambda = \max_{m\in[1:K]} \lambda_m,
	\qquad\text{and}\qquad
	\eta = \max_{x^n\in \cX^n} \eta_{x^n}.\]
The probability of error of the code $(f,\decB,\decC)$ is defined as: 
\begin{align} \err = \max(\lambda,\eta).\label{eq:errordefn}\end{align}
We write $\err^{(n)}$ when we want to explicitly show the dependence on the block length $n$.
\paragraph{Achievable rates, capacity} 
We say rate $R$ is {\em achievable with consensus} if there is an $\epsilon>0$ such that for all sufficiently large $n$ there is an $(n,\lfloor{2^{nR}}\rfloor)$ consensus code with $\err^{(n)}\leq 2^{-n\epsilon}$ (we suppress the floor function in the sequel). The {\em consensus capacity}
$\capstr$ is the supremum of all rates achievable with consensus.
\begin{remark}\label{rem:weak}Notice that the definition above demands $-\log(\err^{(n)}) = \Omega(n)$. It turns out that the capacity remains unchanged even if this is relaxed to $\err^{(n)}=o(1/n)$, the condition under which we prove our converse. Surprisingly, 
it turns out that a converse cannot be shown if this is further relaxed to $\err^{(n)}=o(1)$. 
In Appendix~\ref{app:weak} we show an example where $\capstr= 0$, but a positive rate is achievable with $\err^{(n)}=o({1}/{n^{\frac{1}{2}-\epsilon}})$, for any $\epsilon>0$.
\end{remark}

\subsection{Some remarks on the definitions}\label{sec:setup-remarks}

\paragraph{Average error probability} We may also define a notion of ``average'' error probability $\erravg$ as the maximum of $\eta$ and $\lambda^\prime$ defined below:
\begin{align*}
	\lambda^\prime &= \frac{1}{K} \sum_{m=1}^K \lambda_m.
\end{align*}

We may argue that the capacity remains unchanged if we replace maximal error probability with average error probability in their definitions.  Clearly, $\lambda^\prime \leq \lambda$ 
and hence $\erravg \leq \err$. 
Thus, the consensus capacity for the average error criterion is no smaller than that for maximal error probability.  
Along the lines of the standard expurgation argument connecting maximal and average error capacities for point-to-point channels~\cite[page~204]{CoverT2006}, it is clear that given an $(n,2^{nR})$ consensus code with a certain $\lambda^\prime$, we can construct an $(n,2^{nR}/2)$ consensus code with $\lambda$ no larger than $2\lambda^\prime$ (and identical $\eta$) by discarding half the codewords with the worse $\lambda_m$'s (and replacing decoder outputs which map to discarded codewords by $\bot$). Hence, the consensus capacity for maximal error probability criterion is also no smaller than that for average error probability. Thus the consensus capacity is agnostic to the choice of maximal or average error probability in its definitions. 

\paragraph{Randomization} Allowing for common randomness shared by the sender and both receivers does not change the capacities\footnote{Notice that a malicious sender may choose its transmission depending on the realization of the common randomness. Hence the probability of error when the sender and the receivers share common randomness is the weighted average of probabilities of error (of the deterministic codes) under the different possible realizations of common randomness. Thus, there is a deterministic code whose probability of error is no worse than that of a code with common randomness.}. This also implies that private randomization by the sender does not alter the capacities (since turning the private randomness at the sender into common randomness by providing it to both the decoders cannot decrease the capacity). 

The presence of randomness shared by the decoders (or more generally, samples of correlated sources at the decoders independent of the channel) and unknown to the sender can be absorbed in the model as an additional component in the channel outputs $Y$ and $Z$ independent of the input and the rest of the channel outputs; so we do not introduce separate notation for this. Our results will show that this additional shared randomness has no effect on the consensus capacity $\capstr$ (see Remark~\ref{rem:sharedrandomness}). 
However, as we will discuss in Appendix~\ref{app:weak}, if, unlike our definition above, we only require that $\err^{(n)}\to 0$ as $n\to\infty$, common randomness shared by the decoders and unknown to the sender may affect the rate of communication with consensus. In Appendix~\ref{app:weak} we show an example with $\capstr= 0$, but a positive rate is achievable with $\err^{(n)}=o({1}/{n^{\frac{1}{2}-\epsilon}})$, for any $\epsilon>0$, if the decoders share common randomness unknown to the sender.

\subsection{Common channel}
The {\em common channel} of a broadcast channel will play a vital role in the characterization of its consensus capacity.
\begin{definition}[Common Channel, Common Channel Output Functions]\label{def:common-channel}
    The \emph{characteristic graph}\footnote{Not to be confused with the characteristic graph of a point-to-point channel~\cite{Shannon56ZeroError,AlonO95Repeated}.
} of a broadcast channel $\ch_{YZ|X}$ is the bipartite graph $G_\ch=(\cN,\cE)$, with vertex set $\cN=\cY\cup\cZ$ and edge set $\cE = \inb{\inb{y, z}: \ch(y,z|x)>0 \text{ for some }x\in \cX}$.
    Let $\cV$ be such that $G_v=(\cN_v,\cE_v), v\in\cV$ are the distinct connected components\footnote{A connected component of a graph is an induced subgraph in which every pair of vertices is connected by a path and which is not connected to any vertices in the rest of the graph. Without loss of generality, we assume that each letter in $\cY$ ($\cZ$, resp.) receive positive probability under $\ch_{Y|X}$ ($\ch_{Z|X}$, resp.) for some input letter so that none of the connected components consist of a single vertex.} of the characteristic graph $G_\ch$. 
    The {\em common channel} ${\ch}_{V|X}$ of $\ch_{YZ|X}$ is a point-to-point channel with input alphabet $\cX$ and output alphabet $\cV$ such that
    \begin{align}
        {\ch}_{V|X}(v|x) = \sum_{\{y,z\}\in \cE_v} \ch_{YZ|X}(y,z|x),\; v\in\cV, x\in\cX. \label{eq:commonch}
\intertext{Also define, for $(x,y,z,v)\in\cX\times\cY\times\cZ\times\cV$ such that $W_{V|X}(v|x)>0$,}
	\ch_{YZ|XV}(y,z|x,v) = 
	\begin{cases}
		\frac{\ch_{YZ|X}(y,z|x)}{\ch_{V|X}(v|x)}, & \{y,z\}\in \cE_v\\
		0, &\text{otherwise}. \label{eq:commonch2}
	\end{cases}
    \end{align}
We say that the common channel is {\em trivial} if its Shannon capacity is 0, i.e., if ${\ch}_{V|X}(.|x)$ is identical for all $x\in\cX$.
  
The {\em common channel output functions} $\phi_1:\cY\rightarrow\cV$ and $\phi_2:\cZ\rightarrow\cV$ map their argument to the index of the connected component to which the argument belongs. i.e., $\phi_1(y) = v$, where $\cN_v\ni y$, and $\phi_2(z) = v$, $\cN_v\ni z$.
\end{definition}
Clearly, both receivers can infer the common channel output. Specifically, $\phi_1(Y)=\phi_2(Z)=V$ irrespective of the channel input symbol $x$. Hence, $\capptop(\ch_{V|X})=\max_{P_X} I(X;V)$ is a lower bound on $\capstr$. 
An upper bound is the (non-byzantine) common message capacity,
\begin{align}\capcom(W_{YZ|X})=\max_{P_X}\min(I(X;Y),I(X;Z)).\label{eq:capcom}\end{align} 
Hence, \begin{align}\capptop(\ch_{V|X}) \leq \capstr(\ch_{YZ|X}) \leq \capcom(\ch_{YZ|X}).\label{eq:comparison}\end{align}
Our main result (Theorem~\ref{thm:Capacity}) will imply that $\capstr(\ch_{YZ|X})>0$ if and only if $\capptop(\ch_{V|X})>0$ and that the inequalities 
above are loose in general (see Figure~\ref{fig:bsc}). We note in passing that, when $|\cX|=1$, the common channel reduces to the common random variable (of the pair $Y,Z$) related to the notion of common information of G\'acs and K\"orner\cite{GacsK73}.




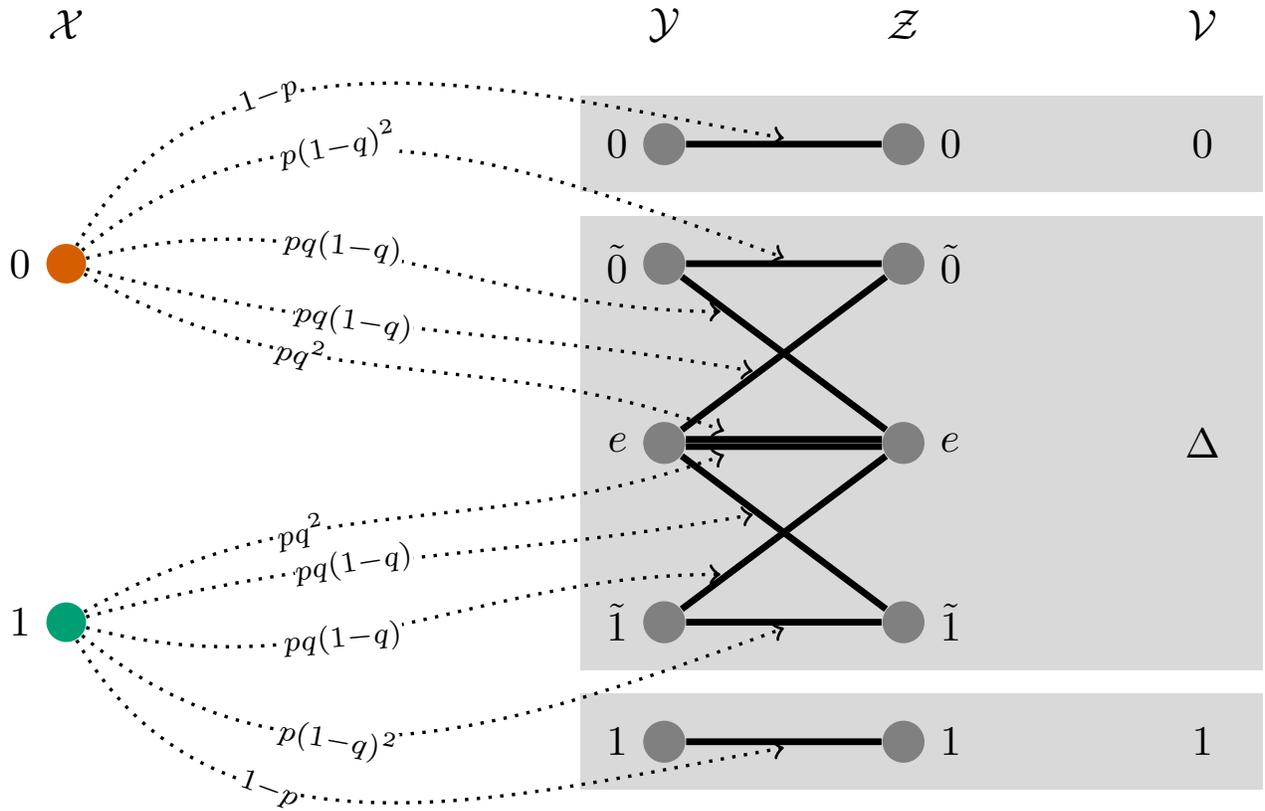
\begin{figure}
\centering
\resizebox{0.95\columnwidth}{!}{
 \begin{tikzpicture}[
    input/.style={circle},
    vertex/.style={draw=none, ellipse, fill=gray, minimum size=10pt, inner sep=0},
    lbl/.style={right, sloped, fill=white, inner sep=0cm},
    ed0/.style={ultra thick, auto, draw=bluishgreen},
    ed1/.style={ultra thick, auto, draw=vermillion},
    pointer/.style={thick, dotted, draw},]
    \node[draw=none] (X) at (-3.0,6.5) {$\cX$};
    \node[draw=none] (X) at (2.0,6.5) {$\cY$};
    \node[draw=none] (X) at (4.0,6.5) {$\cZ$};
    \node[draw=none] (X) at (6.5,6.5) {$\cV$};
    \draw[draw=none,fill=gray!30] ($(2,0.5)+(-0.7,0.4)$)  rectangle ($(5.3,0.5)+(1.7,-0.4)$);
    \draw[draw=none,fill=gray!30] ($(2,4.5)+(-0.7,0.4)$)  rectangle ($(5.3,1.5)+(1.7,-0.4)$);
    \draw[draw=none,fill=gray!30] ($(2,5.5)+(-0.7,0.4)$)  rectangle ($(5.3,5.5)+(1.7,-0.4)$);
    \node[draw=none] (v0) at (6.5,0.5) {$1$};
    \node[draw=none] (v1) at (6.5,3.0) {$\Delta$};
    \node[draw=none] (v2) at (6.5,5.5) {$0$};
    \node[input, fill=bluishgreen, label=left:$1$] (a) at (-3.0,1.5) {}; 
    \node[input, fill=vermillion, label=left:$0$] (b) at (-3.0,4.5) {}; 
    \node[vertex, label=left:$1$] (0l) at (2,0.5) {}; 
    \node[vertex, label=right:$1$] (0r) at (4,0.5) {}; 
    \node[vertex, label=left:$\tone$] (h0l) at (2,1.5) {}; 
    \node[vertex, label=right:$\tone$] (h0r) at (4,1.5) {}; 
    \node[vertex, label=left:$e$] (el) at (2,3) {}; 
    \node[vertex, label=right:$e$] (er) at (4,3) {}; 
    \node[vertex, label=left:$\tzero$] (h1l) at (2,4.5) {}; 
    \node[vertex, label=right:$\tzero$] (h1r) at (4,4.5) {}; 
    \node[vertex, label=left:$0$] (1l) at (2,5.5) {}; 
    \node[vertex, label=right:$0$] (1r) at (4,5.5) {}; 
    \path[ed0] (0l) edge (0r);
    \path[ed0] (h0l) edge (h0r);
    \path[ed0] (h0l) edge (er);
    \path[ed0] (el) edge (h0r);
    \path[ed0] ($(el)+(0.18,-0.03)$) edge ($(er)+(-0.18,-0.03)$);
    \path[ed1] ($(el)+(0.18,0.03)$) edge ($(er)+(-0.18,0.03)$);
    \path[ed1] (el) edge (h1r);
    \path[ed1] (h1l) edge (er);
    \path[ed1] (h1l) edge (h1r);
    \path[ed1] (1l) edge (1r);
    \path[pointer, in=-170, out=-60] (a) edge[->]  node[pos=0.3, lbl] (t) {$\scriptstyle 1-p$} ($(0l)+(1,-0.05)$);
    \path[pointer, in=-160, out=-40] (a) edge[->] node[pos=0.3, lbl] (t) {$\scriptstyle p(1-q)^2$} ($(h0l)+(1,-0.05)$);
    \path[pointer, in=180, out=-15] (a) edge[->] node[pos=0.3, lbl] (t) {$\scriptstyle pq(1-q)$} ($(h0l)+(0.47,0.4)$);
    \path[pointer, in=190, out=15] (a) edge[->] node[pos=0.3, lbl] (t) {$\scriptstyle pq(1-q)$} ($(h0l)+(0.74,0.9)$);
    \path[pointer, in=200, out=30] (a) edge[->] node[pos=0.3, lbl] (t) {$\scriptstyle pq^2$} ($(el)+(0.5,-0.1)$);
    \path[pointer, in=170, out=60] (b) edge[->] node[pos=0.3, lbl] (t) {$\scriptstyle 1-p$} ($(1l)+(1,+0.05)$);
    \path[pointer, in=160, out=40] (b) edge[->] node[pos=0.3, lbl] (t) {$\scriptstyle p(1-q)^2$} ($(h1l)+(1,0.05)$);
    \path[pointer, in=-180, out=15] (b) edge[->] node[pos=0.3, lbl] (t) {$\scriptstyle pq(1-q)$} ($(h1l)+(0.47,-0.4)$);
    \path[pointer, in=-190, out=-15] (b) edge[->] node[pos=0.3, lbl] (t) {$\scriptstyle pq(1-q)$} ($(h1l)+(0.74,-0.9)$);
    \path[pointer, in=-200, out=-30] (b) edge[->] node[pos=0.3, lbl] (t) {$\scriptstyle pq^2$} ($(el)+(0.5,0.1)$);
 \end{tikzpicture} 
}
\caption{Two-step binary erasure broadcast channel. The channel erases in two steps -- with probability $1-p$, both receivers (simultaneously) receive the input symbol unerased; with the remaining probability $p$, the input symbol is passed further through independent binary erasure channels which erase with probability $q$ and whose unerased output symbols acquire a $\tilde{\phantom{a}}$. The characteristic graph has three connected components (unless $p=1$ when there is only one connected component). The common channel is a binary erasure channel with erasure probability $p$ and erasure symbol $\Delta$.} \label{fig:bec}
 \end{figure}
\begin{example}[Two-step binary erasure broadcast channel] 
Let $p\in[0,1], q\in(0,1]$ and $\cX=\{0,1\}$, $\cY=\cZ=\{0,1,\tzero,\tone,e\}$. See Figure~\ref{fig:bec}. 
\[\ch_{YZ|X}(y,z|x)=\begin{cases} 1-p,&y=z=x, x\in\cX \\ pQ(y|x)Q(z|x), &y,z\in\{\tzero,\tone,e\}, x\in\cX,\end{cases}\]
where $Q(e|x)=1-Q(\tx|x)=q,x\in\cX$ with a slight abuse notation to denote $\tzero$ ($\tone$, resp.) by $\tx$ when $x$ is 0 (1, resp.). 
From Figure~\ref{fig:bec}, the common channel can be seen to be $\ch_{V|X}(\Delta|x)=1-\ch_{V|X}(x|x)=p$, where the common channel output alphabet is $\cV=\{0,1,\Delta\}$. The common channel is trivial when $p=1$.
Note that $q=0$ above amounts to the noiseless channel (with an additional Bernoulli-$p$ common random variable output independent of the input). Clearly all capacities are 1 in this case. We do not include this straightforward case in our parametrization so that the discussion below can be kept general.
\end{example}

\section{Consensus Capacity of the Two-Step Binary Erasure Broadcast Channel}\label{sec:short}
We first illustrate some of the key ideas behind the proof of our capacity theorem (Theorem~\ref{thm:Capacity}) by considering the special case of the two-step binary erasure broadcast channel.
Notice that the marginal channel to each receiver is, effectively, a binary erasure channel (BEC) with erasure probability $pq$. Hence, the point-to-point capacity of the marginals channels is $1-pq$. Since the uniform $P_X$ simultaneously maximizes $I(X;Y)$ and $I(X;Z)$, by \eqref{eq:capcom}, the common message capacity is also $\capcom=1-pq$. 
In this section, we will show that the consensus capacity $\capstr$ is also equal to $1-pq$ {if and only if} $p<1$. Moreover, if $p=1$, we will show that $\capstr=0$. Note that the capacity of the common channel is $\capptop(\ch_{V|X})=1-p$. Hence, the common channel is trivial (i.e., its point-to-point capacity is 0) precisely when $p=1$, the condition under which $\capstr=0$. Also note that, for $p<1$, the consensus capacity, common message capacity, and the common channel capacity are related by $\capstr=\capcom=1-pq\geq 1-p=\capptop(\ch_{V|X})$, where the inequality is strict if $p>0, q<1$, i.e.,$\capstr$ is strictly greater than $\capptop(\ch_{V|X})$ in this regime (cf. \eqref{eq:comparison}).  



\subsection{Converse: $\capstr=0$ if $p=1$} \label{sec:shortconverse}
With $p=1$, the channel is the independent binary erasure broadcast channel $\ch_{YZ|X}(y,z|x)=Q(y|x)Q(z|x)$, where $Q(e|x)=1-Q(\tx|x)=q$. Let $q\in(0,1)$ as the case of $q=1$ is obvious. Note that the characteristic graph has a single connected component.
Consider an $(n,2^{nR})$ consensus code $(\enc,\decB,\decC)$ with error probability $\err$. 
The key ingredient will be the following claim which states that changing the channel input vector at one location ($k$-th, say) should only produce a small effect on the decisions of the decoders.  
\begin{claim}\label{cl:shortconverse}
Suppose $k\in[1:n]$, $m\in[1:2^{nR}]$, and $\hx_1,\ldots,$ $\hx_{k-1},\hx_k,x_k,x_{k+1},\ldots,x_n \in \cX$. Let $\cA_m$ be the event $(\decB(Y^n)=\decC(Z^n)=m)$.
\begin{align}
&\Pr(\cA_m|X^n=(\hx_1^{k-1},x_k,x_{k+1}^n)) \notag\\
&\qquad-\Pr(\cA_m|X^n=(\hx_1^{k-1},\hx_k,x_{k+1}^n)) \leq \err\rho, \label{eq:shortconverse}
\end{align}
where $\rho = \frac{5}{\min(q^2,(1-q)^2)}$. 
\end{claim}
We prove this claim later, but the intuition for it can be summarized as follows: As the channel is memoryless, any change in the decisions of the decoders must be based on the channel outputs $Y_k$ and $Z_k$, resp. (i.e, outputs at the location with the change in input). Since the characteristic graph has a single connected component, the decoders cannot extract a non-trivial common part from $Y_k,Z_k$ and, as their decisions must agree with high probability for any input, the effect on their decisions must be small. 

To complete the proof, consider two distinct messages $m,\hm$ with codewords $x^n=f(m)$ and $\hx^n=f(\hm)$. Summing \eqref{eq:shortconverse} over $k=1,\ldots,n$,  
\begin{align*} \Pr(\cA_m|X^n=x^n) - \Pr(\cA_m|X^n=\hx^n) \leq n\err\rho.\end{align*}
Since $\Pr(\cA_m|X^n=x^n)=\Pr(\decB(Y^n)=\decC(Z^n)=m|X^n=x^n)\geq 1 - \err$,
\[ \Pr(\decB(Y^n)=\decC(Z^n)=m|X^n=\hx^n) \geq 1 - \err(1+n\rho).\]
Thus, $\err\geq\lambda_{\hm}\geq 1 - \err(1+n\rho)$ and hence $\err \geq 1/(2+n\rho)$. Since the definition of $\capstr$ requires $\err$ to decay faster than this as $n\rightarrow \infty$,  $\capstr=0$.
\proof[Proof of Claim~\ref{cl:shortconverse}]{
Consider the random variables $\hY_1^k,\hZ_1^k,Y_k^n,Z_k^n$ jointly distributed as
\begin{align}
&p_{\hY_1^{k},\hZ_1^{k},Y_{k}^n,Z_{k}^n}(\hy_1^{k},\hz_1^{k},y_{k}^n,z_{k}^n)\notag\\
 &=\left(\prod_{i=1}^{k} \ch_{YZ|X}(\hy_i,\hz_i|\hx_i)\right)
 \prod_{i=k}^n \ch_{YZ|X}(y_i,z_i|x_i).
\end{align}
Notice that we are defining a coupling where $(\hY_1^{k-1},Y_k^n,\hZ_1^{k-1},Z_k^n)$ and $(\hY_1^k,Y_{k+1}^n,\hZ_1^k,Z_{k+1}^n)$ have the same distributions as $(Y^n,Z^n)$ in the first and second terms, resp., of \eqref{eq:shortconverse}. Let $S=(\hY_1^{k-1},Y_{k+1}^n), T=(\hZ_1^{k-1},Z_{k+1}^n)$. For $s=(\hy_1^{k-1},y_{k+1}^n)$, abusing notation, we will write $\decB(s,y)$ to mean $\decB((\hy_1^{k-1},y,y_{k+1}^n))$. Similarly, we will also use $\decC(t,z)$.
We have (by \eqref{eq:errordefn}), for $x\in\cX$,
\begin{align*}
\err &\geq \Pr(\text{decoders disagree}|X^n=(\hx_1^k,x,x_{k+1}^n))\\
      &= \sum_{y,z} \ch_{YZ|X}(y,z|x) \Pr(\decB(S,y)\neq\decC(T,z)).
\end{align*}
Hence, for every edge $\{y,z\}\in\cE$ in the characteristic graph (i.e., $\ch_{YZ|X}(y,z|x)$ for some $x$), 
\begin{align}
\Pr(\decB(S,y)&\neq\decC(T,z))\notag\\
 &\leq \frac{\err}{\min_{(x,y,z):\ch_{YZ|X}(y,z|x)>0} \ch_{YZ|X}(y,z|x)}\notag\\
 &=\frac{\err}{\min(q^2,(1-q)^2)}. \label{eq:shortclaimstep}
\end{align}
Consider the event $\cE$ in which the decoder outputs do not depend on the $k$-th element of their channel output vectors,
\[ \cE=\left( \left|\bigcup_{\{y,z\}\in\cE} \left\{\decB(S,y),\decC(T,z)\right\} \right| = 1\right).\]
Since the characteristic graph of the channel is connected and has a spanning tree with 5 edges, from \eqref{eq:shortclaimstep}, we may conclude using a union bound that 
\begin{align}
\Pr(\cE)\geq 1- \frac{5\err}{\min(q^2,(1-q)^2)} = 1-\err\rho.
\end{align}
i.e., under the distribution of $(S,T)$, with probability at least $1-\err\rho$, the decoder outputs do not depend on the $k$-th element of their channel output vectors. Then, 
\begin{align*}
&\Pr(\decB(S,Y_k)=\decC(T,Z_k)=\decB(S,\hY_k)=\decC(T,\hZ_k))\\
&= \sum_{y,z,\hy,\hz} \ch_{YZ|X}(y,z|x_k)\ch_{YZ|X}(\hy,\hz|\hx_k)\\
  &\qquad\qquad \Pr(\decB(S,\hy)=\decC(T,\hz)=\decB(S,y)=\decC(T,z))\\
&\geq \sum_{y,z,\hy,\hz} \ch_{YZ|X}(y,z|x_k)\ch_{YZ|X}(\hy,\hz|\hx_k) \Pr(\cE)\\
&\geq 1 - \err\rho.\\
\iftoggle{long}{\intertext{Hence,}
&\Pr(\decB(S,Y_k)=\decC(T,Z_k)=m) \\&\qquad\qquad- \Pr(\decB(S,\hY_k)=\decC(T,\hZ_k)=m) \leq \err\rho.
}{
&\text{Hence, }\Pr(\decB(S,Y_k)=\decC(T,Z_k)=m) \\&\qquad\qquad- \Pr(\decB(S,\hY_k)=\decC(T,\hZ_k)=m) \leq \err\rho.
}
\end{align*}
}
\begin{remark}\label{rem:Witsenhausen}
\nottoggle{long}{
In \cite{long} we strengthen the converse to show that even a single bit cannot be communicated with consensus over this channel with $\err^{(n)}\rightarrow 0$. This also means that for this channel the converse does not require the more restrictive $\err^{(n)}=o(1/n)$. However, as mentioned in Remark~\ref{rem:weak}, in general, such a requirement is necessary and our proof of the converse of Theorem~\ref{thm:Capacity} generalizes the proof idea above.
}{ 
Below we strengthen the converse to show that even a single bit cannot be communicated with consensus over this channel with $\err^{(n)}\rightarrow 0$. This also means that for this channel the converse does not require the more restrictive $\err^{(n)}=o(1/n)$. However, as mentioned in Remark~\ref{rem:weak}, in general, such a requirement is necessary and our proof of the converse of Theorem~\ref{thm:Capacity} generalizes the proof idea above.

We will show that for the independent binary erasure broadcast channel with erasure probability $q>0$, there exists $\epsilon>0$ such that $\err^{(n)}\geq \epsilon$ for any $(n,2)$ consensus code, $n\in\bbN$. Consider an $(n,2)$ consensus code $(\enc,\decB,\decC)$, with codewords $\enc(1)=(x_1,\ldots,x_n)=:x^n$ and $\enc(2)=(\hx_1,\ldots,\hx_n):=\hx^n$. We have
\begin{align}
   \Pr\inp{\decB(Y^n)=\decC(Z^n)=2|X^n=\hx^n}&\ge 1-\err, \label{eq:Witn}\\
   \Pr(\decB(Y^n)=\decC(Z^n)=2|X^n=x^n)&\le1-\Pr(\decB(Y^n)=\decC(Z^n)=1|X^n=x^n)\le\err. \label{eq:Wit0}
\end{align}
Furthermore, by Claim~\ref{cl:shortconverse}, for all $k\in[1:n]$, $m\in\{1,2\}$,
\begin{align}
&\Pr(\decB(Y^n)=\decC(Z^n)=2|X^n=(\hx_1^{k-1},x_{k}^n)) 
-\Pr(\decB(Y^n)=\decC(Z^n)=2|X^n=(\hx_1^{k},x_{k+1}^n)) \leq \err\rho \label{eq:Witstepk}
\end{align}
i.e., $\Pr(\decB(Y^n)=\decC(Z^n)=2|X^n=(\hx_1^{k},x_{k+1}^n))$ is at most $\err$ for $k=0$ (by \eqref{eq:Wit0}); changes by at most $\err\rho$ at each step as $k$ increases from $0$ to $n$ in steps of 1 (by \eqref{eq:Witstepk}); and is at least $1-\err$ at $k=n$ (by \eqref{eq:Witn}). Hence, there must be a $k\in[1:n]$ such that
    \begin{align}\label{eq:wit}
        \frac{1}{2}-\err\rho \le \pr(\decB(Y^n)=\decC(Z^n)=2|X^n=(\hx_1^{k},x_{k+1}^n))\le \frac{1}{2}+\err\rho.
    \end{align}
For this $k$, fix $X^n=\bx^n:=(\hx_1^{k},x_{k+1}^n)$. Then $(Y^n,Z^n)$ have the following joint distribution: $(Y_i,Z_i)$ are independent over $i=1,\ldots,n$, with the joint distribution of $(Y_i,Z_i)$ given by 
\begin{align*}
P_{Y_i,Z_i}(y,z) =\begin{cases} q^2, &(y,z)=(e,e)\\ (1-q)^2, &(y,z)=(\bx_i,\bx_i)\\ q(1-q), &(y,z)\in\{(e,\bx_i),(\bx_i,e)\}\\ 0 &\text{otherwise}.\end{cases}
\end{align*}
For all $i\in[1:n]$, $P_{Y_i,Z_i}$ has zero G\'{a}cs-K\"orner common information~\cite{GacsK73} (i.e., their maximum correlation~\cite{Hirschfeld35,Gebelein41,Renyi59} is less than unity). Hence, by a result of Witsenhausen~\cite{Witsenhausen75}, there exists $\epsilon'>0$ (which depends only on $q$) such that for all deterministic functions\footnote{Witsenhausen~\cite{Witsenhausen75} considers functions which make a binary decision. Here, we may view the decoders $\decB,\decC$ as making a binary decision returning either the symbol $2$ or a symbol from $\{1,\bot\}$.} (specifically, $\decB$ and $\decC$), \eqref{eq:wit} holds only if $\err\rho\geq\epsilon'$. Thus, for arbitrarily small $\err$, $(n,2)$ consensus codes do not exist for any $n\in\mathbb{N}$.

Indeed, using this argument, we can prove such an impossibility for any channel that satisfies the following properties: 
\begin{enumerate}
\item The characteristic graph has a single connected component (it is easy to see that our proof of Claim~\ref{cl:shortconverse} made use of only this property of the channel); and
\item For each of its input symbols $x\in\cX$, the joint distribution $\ch_{YZ|X}(.,.|x)$ induced at the output by the channel has zero common information (so that the impossibility in~\cite{Witsenhausen75} applies).
\end{enumerate}
}
\end{remark}

\subsection{Achievability: $\capstr=1-pq$ if $p<1$}\label{sec:short_ach}
We introduce some notation and describe our decoder before giving the intuition behind our scheme.
For $x^n\in\cX^n$ and $v^n\in\cV^n$, we write $x^n \vartriangleright v^n$ if $v^n$ is an ``erased'' version of $x^n$, i.e., if $v_i\in\{x_i,\Delta\}$, $i\in[1:n]$. Similarly,  we write $x^n \blacktriangleright y^n$ if $y_i\in\{x_i,\tx_i,e\}$, $i\in[1:n]$. Let $d(x^n,\hx^n) := \frac{1}{n}(d_{\text{Hamming}}(x^n,\hx^n))=\frac{1}{n}\sum_{i=1}^n 1_{x_i = \hx_i}$ be the relative distance between $x^n,\hx^n\in\cX^n$.

Let $\rad>0$.
For an encoder $f$ of rate $R$, the decoder outputs $\decB(y^n)=m$ if it is the unique $m\in[1:2^{nR}]$ such that
\begin{enumerate}
\item[(i)] $f(m)\vartriangleright \phi_1(y^n)$, where $\phi_1(y^n):=(\phi_1(y_i))_{i\in[1:n]}$,
\item[(ii)] there is a $\bx^n\in\cX^n$ such that $d(f(m),\bx^n)< \rad$ and $\bx^n\blacktriangleright y^n$.
\end{enumerate}
$\decB(y^n)=\bot$ if no such unique $m$ exists. $\decC$ is similarly defined (with $\phi_2$ in lieu of $\phi_1$).
The first decoding condition requires the codeword to match the bits left unerased by the common channel; we denote this by  $m \lozenge y^n$. The second condition, denoted by $m \blacklozenge y^n$, requires an ``explaining'' vector $\bx^n$ which is $\delta$-close to the codeword and matches the bits left unerased in $y^n$. 

The intuition behind our coding scheme is as follows: Since the first decoding condition above only depends on the common channel output, both decoders will make the same decision on this. However, if they were to rely only on this condition, they cannot achieve rates above the common channel capacity $1-p$. Instead, if they were to use the decoding condition $f(m)\blacktriangleright y^n$ (resp., $f(m)\blacktriangleright z^n$) which (with an erasure code) can achieve all rates below $1-pq$ in the non-byzantine setting, there is a simple attack for the byzantine sender -- send $f(m)$ with one of the bits flipped. A receiver for which this bit is erased by the channel may accept $m$ while one for which this bit is left unerased will reject $m$; since there is a finite probability $(2pq(1-q))$ that this bit is erased for exactly one of the receivers, with non-vanishing probability they may disagree. The second decoding condition above circumvents this by tolerating some errors. A malicious sender may still try to get the receivers to disagree by sending a vector which is close to the boundary of tolerance with the hope that (with non-vanishing probability) channel erasures push one of the receivers to accept and the other to reject. However, since this attack requires the sent vector to be sufficiently far away from a codeword, the first condition, which tolerates no errors, will be able to detect it. This has the side-effect that when a vector which is close to a codeword (but is not the codeword itself) is sent, there is a significant probability that the message corresponding to that codeword is rejected; but this rejection (based on the first decoding condition) is carried out by both the receivers simultaneously so that their decisions still agree (see case~(iii) below).

Turning to the formal proof, for $x^n\in\cX^n$, define the event
\[\cB_{x^n}=(\exists m\in[1:2^{nR}]: d(f(m),x^n)\geq \rad, m \lozenge Y^n, m \blacklozenge Y^n).\]

\begin{claim}\label{cl:shortachievability}
Let $R<1-pq$. There are positive $\rad,\epsilon$ such that, for sufficiently large $n$, there is an encoder $f:[1:2^{nR}]\rightarrow\cX^n$ with $d(f(m),f(m'))\geq 2\rad$ for every pair 
$m\neq m'$ and
\begin{align}
\Pr( \cB_{x^n} | x^n) &\leq 2^{-n \epsilon},\text{ for all } x^n\in\cX^n. \label{eq:shortachievability}
\end{align}
\end{claim}
Before proving the claim, let us see that it implies the theorem. This will follow from a case analysis for the transmitted vectors $x^n\in\cX^n$:\\
{\em Case~(i):} $x^n=f(m)$ for some $m$.
 Then, $m\lozenge Y^n$ and $m\blacklozenge Y^n$. Moreover, for all $m'\neq m$, the encoder in Claim~\ref{cl:shortachievability} has $d(f(m),f(m'))\geq2\delta$. Hence,
\begin{align*} &\Pr(\decB(Y^n)\neq m|f(m))\\
&= \Pr(\exists m'\neq m \text{ s.t } m'\lozenge Y^n,m'\blacklozenge Y^n|f(m))\\ 
&= \Pr(A_{f(m)}|f(m)) \leq 2^{-n\epsilon},
\end{align*}
where the last inequality follows from \eqref{eq:shortachievability}. Similarly, $\Pr(\decC(Z^n)\neq m|f(m)) \leq 2^{-n\epsilon}$. By a union bound, 
$\Pr(\decB(Y^n)=\decC(Z^n)=m|f(m))\geq1-2^{-n\epsilon+1}$.\\
{\em Case~(ii):}  $d(f(m),x^n)\geq\delta$ for all $m\in[1:2^{nR}]$. Then, $\Pr(\decB(Y^n)\neq\bot|x^n)\leq\Pr(\cB_{x^n}|x^n)\leq 2^{-n\epsilon}$, where the last step is from~\eqref{eq:shortachievability}. Hence, by a union bound $\Pr(\decB(Y^n)=\decC(Z^n)=\bot|x^n)\geq 1 - 2^{-n\epsilon+1}$.\\
{\em Case~(iii):}  there is an $m$ such that $d(f(m),x^n)<\rad$, but $x^n\neq f(m)$. Since $d(f(m),x^n)<\rad$, by triangle inequality, $d(f(m'),x^n) \geq d(f(m'),f(m)) - d(f(m),x^n) >$ $2\rad - \rad = \rad$ for all $m'\neq m$. Hence, $\Pr(\decB(Y^n)\notin\{m,\bot\}|x^n) \leq \Pr(\cB_{x^n}|x^n)$ $\leq 2^{-n\epsilon}$, where the last step follows from~\eqref{eq:shortachievability}. By the union bound, $\Pr(\decB(Y^n),\decC(Z^n)\in\{m,\bot\}|x^n) \geq 1 - 2^{-n\epsilon+1}$. We will argue that, for this $x^n$ and under the event $(\decB(Y^n),\decC(Z^n)\in\{m,\bot\})$, the decoder outputs must match which will complete the proof. $m\blacklozenge Y^n$ as $x^n$ may serve as the explaining vector $\bx^n$ since $d(f(m),x^n)<\rad$ and $x^n \blacktriangleright Y^n$. Similarly, $m\blacklozenge Z^n$. Hence, the second decoding condition for message $m$ is met for both decoders. Since $\phi_1(Y^n)=\phi_2(Z^n)$, either the first condition for message $m$ is met or not met together for both decoders. Hence, $\Pr(\decB(Y^n)=\decC(Z^n)|x^n)\geq 1 - 2^{-n\epsilon+1}$.

It only remains to prove Claim~\ref{cl:shortachievability}. 
\proof[Proof of Claim~\ref{cl:shortachievability}]{
We use the method of types and follow the notation from~\cite{Csiszar98} (also see Section~\ref{sec:notation}). 
Let $P$ be the uniform type\footnote{For simplicity, we assume $n$ is even; the case of odd $n$ is easily handled by perturbing $P$ (or leaving unused, say, the last bit).} on $\cX$, $P(0)=P(1)=1/2$, and $H_2$ denote the binary entropy function.
\begin{lemma}\label{lem:shortcodebook}
For $\rad<1/4$, $0<\epsilon\leq R\leq 1-H_2(2\rad)-\epsilon$ and sufficiently large $n$, there exists an encoder $f:[1:2^{nR}]\rightarrow \{0, 1\}^n$  whose codewords $f(m), m\in [1:2^{nR}]$ are of type $P$ such that 
\begin{align}
d(f(m), f(m'))\geq  2\rad \text{ for all }m\neq m', \nonumber
\end{align} and for every joint type $P_{X'X}\in \cP^n\inp{\cX\times \cX}$ and $x^n\in \cX^n$, 
\begin{align}
|\{m:f(m)\in \cT^{n}_{X'|X}(x^n)\}|\leq2^{n\left(\left|R-I(X';X)\right|^{+}+\epsilon\right)}.\label{eq:code2_example}
\end{align}\end{lemma}
We can show the above lemma using a random coding argument.
The first property is similar to the Gilbert-Varshamov bound and gives a minimum distance guarantee (also see \cite[Problem~10.1(c)]{CsiszarK2011}). The second property is similar to \cite[(V.10)]{Csiszar98}. 
The lemma follows from Lemma~\ref{lemma:codebook}~(page~\pageref{lemma:codebook}) where we take $\cU =\cX$ and $P$ to be the uniform type, i.e., $P(0)=P(1) = 0.5$. 

To proceed with the proof of Claim~\ref{cl:shortachievability}, let \[\cP:= \inb{P_{XX'}\in \cP^n\inp{\cX\times\cX}: P_X=P_{X'}=P, \, \Pr\inp{X\neq X'}< 2\rad}.\] Then, $\min_{P_{X X'}\in \cP}I(X;X') = \min_{P_{X X'}\in \cP} 1-H_2(\Pr\inp{X\neq X'})\leq 1-H_2(2\rad)$, 
where is last inequality follows from $\rad<1/4$.
Choose $\rad>0$ sufficiently small so that $\rad<1/4$, $2(H_2(\rad)+\rad) < (1-pq-R)$, and $R<1-H_2(2\rad)$. Further, choose sufficiently small $\epsilon>0$ so that $\epsilon\leq R\leq 1-H_2(2\rad)-\epsilon$ and consider the codebook from Lemma~\ref{lem:shortcodebook}. Let $\cD$ be the set of all joint types $P_{XX'Y}\in \cP^n\inp{\cX\times \cX\times\cY}$ such that 
\begin{enumerate}
\item[(i)] $X\vartriangleright \phi_1(Y)$,\qquad\qquad  (ii) $\Pr\inp{X'\neq X}\geq\rad$, and
\item[(iii)] $X'\vartriangleright \phi_1(Y), \exists \, P_{\bX|XX'Y}$ s.t. $\Pr\inp{X'\neq\bX}< \rad, \bX \blacktriangleright Y$.
\end{enumerate}
Consider the definition of $\cB_{x^n}$. Suppose, for channel input $x^n$, output $Y^n$ and $m\in[1:2^{nR}]$ are such that $d(f(m),x^n)\geq\rad$ and $m\lozenge Y^n, m\blacklozenge Y^n$. Then, the joint type of $(x^n,f(m),Y^n)$ belongs to $\cD$, i.e., $(x^n,f(m),Y^n) \in \cT^n_{XX'Y}$ for some $P_{XX'Y} \in \cD$, since (i) $x^n\vartriangleright \phi_1(Y^n)$, (ii) $d(f(m),x^n)\geq\rad$, and (iii) $m\lozenge Y^n, m\blacklozenge Y^n$. Further, let $\cD_{\textup{typical}}$ be the set of joint types $P_{XX'Y}$ s.t. $P_Y(e) \leq pq + \sqrt{\epsilon}$ and 
$\Pr(\phi_1(Y)=\Delta|X=x)=\sum_{y\in\{0,1,e\}}P_{Y|X}(y|x) \leq  p + 2\sqrt{\epsilon/\rad}$ for all $x\in\cX$ for which $P_X(x)\geq\rad/4$. Here, we further restrict $\epsilon>0$ to be small enough so that $pq+\sqrt{\epsilon},p+2\sqrt{\epsilon/\rad}<1$ (recall, $p<1$).
By a union bound over $\cD$,
\begin{align*}
\Pr\inp{\cB_{x^n}|x^n}&\leq\sum_{\substack{P_{XX'Y}\in\\\cD\cap \cD_{\textup{typical}}}}\Pr\inp{\exists m:(x^n, f(m), Y^n)\in \cT^n_{XX'Y}}\\
&\qquad\quad+ \sum_{P_{XX'Y}\in \cD\cap \cD^c_{\textup{typical}}}\Pr\inp{(x^n, Y^n)\in \cT^n_{XY}}.
\end{align*}
The second term can be upper bounded by $3(n+1)^22^{-2n\epsilon}$ using Sanov's theorem.
To see this, let $J:=1_{\inp{\phi_1(Y)=\Delta}}$ be the indicator random variable of $(\phi_1(Y)=\Delta)$. Let $\hJ:=1_{\inp{Y=e}}$ and $J^n=(1_{\inp{\phi_1(Y_i)=\Delta}})_{i=1}^{n}$, $\hJ^n=(1_{\inp{Y_i=e}})_{i=1}^{n}$. Let $\bP$ be the type of $x^n$.
\begin{align*}
&\sum_{\substack{P_{XX'Y}\in\\ \cD\cap \cD^c_{\textup{typical}}}}\Pr\inp{(x^n, Y^n)\in \cT^n_{XY}|x^n}\\
&\qquad\leq \sum_{x\in\bP(x)\geq\rad/4}\;\sum_{\substack{P_{J|X}:P_{J|X}(1|x)\\>p+2\sqrt{\epsilon/\rad}}}\Pr\inp{J^n\in\cT^n_{J|X}(x^n)}\\
&\qquad\qquad+\sum_{P_{\hJ}:P_{\hJ}(1)>pq+\sqrt{\epsilon}}\Pr\inp{\hJ^n\in\cT^n_{\hJ}}\\
&\qquad\leq 2(n+1)^2 2^{-2n\epsilon} + (n+1)^2 2^{-2n\epsilon}\leq 3(n+1)^22^{-2n\epsilon}.
\end{align*} 
The second inequality follows from Sanov's theorem and Pinsker's inequality. Specifically, suppose $x\in\cX$ is such that $\bP(x)\geq\rad/4$. Then, for $\sum_{P_{J|X}:P_{J|X}(1|x)>p+2\sqrt{\epsilon/\rad}}\Pr\inp{J^n\in\cT^n_{J|X}(x^n)}$, the exponent of the upper bound from Sanov's theorem is 
\begin{align*}
-(n\bP(x))D\inp{p+2\sqrt{\epsilon/\rad}\middle\|p} \leq -n\bP(x)\frac{8\epsilon}{\rad\ln2}\leq -2n\epsilon,
\end{align*}
where the first inequality follows from Pinsker's inequality. Similarly, for $\sum_{P_{\hJ}:P_{\hJ}(1)>pq+\sqrt{\epsilon}}\Pr\inp{\hJ^n\in\cT^n_{\hJ}}$, the exponent is 
\begin{align*}
-nD\inp{pq+\sqrt{\epsilon}\middle\|pq} \leq -n\frac{2\epsilon}{\ln2}\leq -2n\epsilon.
\end{align*}

For the first term of the upper bound on $\Pr\inp{\cB_{x^n}|x^n}$ above,
\begin{align*}
&\sum_{P_{XX'Y}\in \cD\cap \cD_{\textup{typical}}}\Pr\inp{\exists\, m: (x^n, f(m), Y^n)\in \cT^n_{XX'Y}}\\
&=\sum_{\substack{P_{XX'Y}\in\\ \cD\cap \cD_{\textup{typical}}}} \; \sum_{\substack{m:f(m)\in\\ \cT^n_{X'|X}(x^n)}} \;\sum_{\substack{y^n\in\cT^n_{Y|XX'}(x^n,f(m))}}W(y^n|x^n)\\
&\stackrel{(a)}{\leq}\sum_{P_{XX'Y}\in \cD\cap \cD_{\textup{typical}}}2^{n\left(\left|R-I(X';X)\right|^{+}+\epsilon\right)}2^{{-n\inp{I(Y;X'|X)-\epsilon}}},
\end{align*}
where $(a)$ follows from \eqref{eq:code2_example}, the fact that $|\cT^n_{Y|XX'}(x^n,f(m))|\leq 2^{nH(Y|XX')}$, and  $W(y^n|x^n)$ is the same for each $y^n\in \cT^n_{Y|XX'}(x^n,f(m))$ and is upper bounded by $1/|\cT^n_{Y|X}(x^n)|$ 
which in turn is upper bounded by $2^{-nH(Y|X)+n\epsilon}$ for sufficiently large $n$. For $P_{XX'Y}\in \cD\cap \cD_{\textup{typical}}$, let $\zeta_{XX'Y}:=2^{n\inp{|R-I(X';X)|^{+}-I(Y;X'|X)+2\epsilon}}$. Below, we use a case analysis to argue that $\zeta_{XX'Y}\leq 2^{{-2n\epsilon}}$ for sufficiently small $\epsilon>0$:

{\em Case}~(i): $R\leq I(X';X)$. 
Recall that $J:=1_{\inp{\phi_1(Y)=\Delta}}$.
Then $I(Y;X'|X)\geq I(J;X'|X)$. We show that $I(J;X'|X)\geq 4\epsilon$. To see the intuition behind the argument, suppose $I(J;X'|X)=0$.
Since $X\vartriangleright \phi_1(Y)$ and $X'\vartriangleright \phi_1(Y)$, we have $X=X'$ whenever $J=0$. Also, since $P_{XX'Y}\in\cD_{\textup{typical}}$, for $x$ s.t. $P_X(x)\geq\rad/4$, we have $P_{J|X}(1|x)\leq p + 2\sqrt{\epsilon/\rad}$ which is strictly smaller than 1 (by the further restriction on $\epsilon$ we imposed); hence for such $x$, $P_{J|X}(0|x)>0$. Together with our supposition that 
$I(J;X'|X)=0$, this implies that for such $x$, $P_{X'|X}(x|x)=P_{X'|X,J}(x|x,0)=1$. Thus, $\Pr(X'\neq X) \leq \sum_{x:P_X(x)<\rad/4} P_X(x) \leq \rad/2$ which contradicts $\Pr(X'\neq X)\geq \rad$. This intuition can be extended to obtain a contradiction for $I(J;X'|X)\leq 4\epsilon$ for a sufficiently small choice of $\epsilon>0$.
To see this, suppose $I(J;X'|X)\leq 4\epsilon$. By, Pinsker's inequality, this implies that
\begin{align*}
\sum_{j, x',x}P_{JX}(j,x)|P_{X'|JX}(x'|j,x)-P_{X'|X}(x'|x)|\leq \sqrt{8\epsilon\ln2}.
\end{align*} 
As argued above, for $x\in\cX$ such that $P_X(x)\geq\rad/4$, $P_{X'|JX}(x|0,x)=1$. Hence, for such $x$, $P_{X'|X}(x|x)\geq 1-\frac{\sqrt{8\epsilon\ln2}}{P_{JX}(0,x)}$. Thus,
\begin{align*}
\Pr(X'\neq X) &= 1-\sum_{x}P_{XX'}(x,x) \\
&\leq \sum_{x:P_X(x)\geq\rad/4} \frac{\sqrt{8\epsilon\ln2}}{P_{J|X}(0|x)}+\sum_{x:P_X(x)<\rad/4}\rad/4\\
&\leq \frac{2\sqrt{8\epsilon\ln2}}{1-p-2\sqrt{\epsilon/\rad}}+\rad/2,
\end{align*}
where in the last step we used the fact that for $x$ s.t. $P_X(x)\geq \rad/4$, $P_{J|X}(0|x)=1-P_{J|X}(1|x) \geq 1 - p - 2\sqrt{\epsilon/\rad}$. This contradicts $\Pr(X\neq X')\geq \rad$ if $\epsilon>0$ is chosen sufficiently small such that $\frac{2\sqrt{8\epsilon\ln2}}{1-p-2\sqrt{\epsilon/\rad}}<\rad/2$ (recall, $p<1$). Thus, $I(J;X'|X)\geq 4\epsilon$ which implies that $\zeta_{XX'Y}\leq 2^{{-2n\epsilon}}$.

{\em Case}~(ii): $R> I(X';X)$. Here, $|R-I(X';X)|^{+}-I(Y;X'|X) = R-I(X';XY)\leq R-I(X';Y)$. We will show that $R\leq I(X';Y)-4\epsilon$ to obtain $\zeta_{XX'Y}\leq 2^{-2n\epsilon}$. To see the intuition, suppose  the conditional distribution $P_{\bX|XX'Y}$ is such that $\bX=X'$ (instead of just $\Pr(\bX\neq X')\leq\rad$). Since $\bX \blacktriangleright Y$, this implies  $X' \blacktriangleright Y$. Hence,  $I(X';Y) = H(X')-H(X'|Y)= 1 - \Pr(Y=e)H(X'|Y=e) \geq 1-\Pr(Y=e)\geq 1-pq-\sqrt{\epsilon}$, where in the second equality we used $X' \blacktriangleright Y$ and the fact that $P_{X'}$ is uniform since all codewords are of uniform type. For $\epsilon>0$ small enough such that $4\epsilon+\sqrt{\epsilon}<1-pq-R$ (recall, $R<1-pq$), $R\leq I(X';Y)-4\epsilon$. We can extend this argument to the case when $\Pr(X'\neq \bX)<\rad$.
To see this,
note that
 \begin{align}H(X'|Y) &\leq H(X'\bX |Y)\notag\\&= H(\bX|Y)+H(X'|\bX Y )\notag\\&\leq H(\bX|Y) + H(X'|\bX).\label{eq:shortachive-case2-step}\end{align}
Hence,
\begin{align*}
I(X';Y)  &= H(X') - H(X'|Y)\\
&\stackrel{\text{(a)}}{=} 1 - H(X'|Y)\\
&\stackrel{\text{(b)}}{\geq} 1 - H(\bX|Y) - H(X'|\bX)\\
&\stackrel{\text{(c)}}{\geq} 1 - \Pr(Y=e)H(\bX|Y=e) - H(X'|\bX)\\
&\stackrel{\text{(d)}}{\geq} 1 - \Pr(Y=e) - H(X'|\bX)\\
&\stackrel{\text{(e)}}{\geq} 1 - \Pr(Y=e) - (H_2(\rad) - \rad)\\
&\stackrel{\text{(f)}}{\geq} 1-pq-\sqrt{\epsilon}  -H_2(\rad) - \rad,
\end{align*}
where (a) follows from the fact that $P_{X'}$ is uniform, (b) follows from \eqref{eq:shortachive-case2-step}, (c) from $\bX \blacktriangleright Y$, (d) from the fact that the entropy of a binary random variable is at most 1, (e) from $\Pr(\bX\neq X')\leq\rad$ and Fano's inequality, and (f) from $P_Y(e)\leq pq+\sqrt{\epsilon}$.
Choosing $\epsilon>0$ sufficiently small such that $\sqrt{\epsilon}+ 4\epsilon+H(\rad) +\rad \leq 1-pq-R$ (recall, we chose $\rad>0$ such that $2(H_2(\rad) +\rad) \leq 1-pq-R$), we can conclude that $R\leq I(X';Y)-4\epsilon$. 
Thus,
\[
\Pr\inp{\cB_{x^n}|x^n}\leq 3(n+1)^22^{-2n\epsilon}+|\cD|2^{-2n\epsilon}\leq 2^{-n\epsilon}
\] for sufficiently large $n$ since $|\cD|\leq |\cP^n(\cX\times\cX\times\cY)|$ is polynomial in $n$.
}


\section{Consensus Capacity of Broadcast Channels}\label{sec:general}

The proof ideas from Section~\ref{sec:short} generalize to allow a characterization of the consensus capacity of all broadcast channels. 
\begin{definition}[Effective Input Alphabet $\cU$ and {\RepresentationChannel} $\tP_{U|X}$]\label{def:effective-input}
Recall the definition of the common channel $W_{V|X}$ of a broadcast channel $W_{YZ|X}$ (Definition~\ref{def:common-channel}). We define the (bounded) convex {\em common channel polytope} $\sS\subset\bbR^{|\cV|}$ of $W_{YZ|X}$ as the convex hull of the $|\cV|$-dimensional vectors $\{\ch_{V|X}(.|x), x\in\cX\}$. The {\em effective input alphabet} $\cU\subseteq\cX$ of $W_{YZ|X}$ is defined as follows: each $u\in\cU$ is such that $\ch_{V|X}(.|u)$ is a distinct vertex of $\sS$; if more than one $x\in\cX$ correspond to the same vertex, one among them is arbitrarily chosen to represent that vertex in $\cU$. 
In other words, let $\cU \subseteq \cX$ be a smallest-sized set such that there is a conditional distribution $\tP_{U|X}$  satisfying
\begin{align} \ch_{V|X}(v|x) = \sum_{u\in\cU} \tP_{U|X}(u|x)\ch_{V|X}(v|u), \label{eq:Ptilde}\end{align}
for all $v\in\cV, x\in\cX$. Thus, the cardinality of $\cU$ is the number of vertices of the polytope $\sS$. We will refer to $\tP_{U|X}$ as the {\em \representationchannel}. While the choice of $\tP_{U|X}$ and $\cU$ may not be unique, we choose one among the valid ones for the rest of the discussion. As we will see (Remark~\ref{rem:uniquenessofU}), our results will not depend on this choice.

For $u\in\cU$, we define $\cX_u = \{x\in\cX: W_{V|X}(.|x) = W_{V|X}(.|u)\}$, i.e., $\cX_u$ consists of all input letters $x$ corresponding to the vertex of $\sS$ associated with $u$.
\end{definition}

\begin{figure}
\centering
\resizebox{0.58\columnwidth}{!}{
 \begin{tikzpicture}[
    input/.style={circle},
    vertex/.style={draw=none, ellipse, fill=gray, minimum size=10pt, inner sep=0},
    lbl/.style={right, sloped, fill=white, inner sep=0cm},
    ed0/.style={ultra thick, auto, draw=vermillion},
    ed1/.style={ultra thick, auto, draw=bluishgreen},
    ed2/.style={ultra thick, auto, draw=myblue},
    pointer/.style={thick, dotted, draw},]
    \draw[draw=none,fill=gray!30] ($(2,3)+(-0.7,0.4)$)  rectangle ($(4,2)+(1.9,-0.4)$);
    \draw[draw=none,fill=gray!30] ($(2,1)+(-0.7,0.4)$)  rectangle ($(4,0)+(1.9,-0.4)$);
    \node[input, fill=vermillion, label=left:$0$] (0) at (-1.0,2.5) {}; 
    \node[input, fill=myblue, label=left:$e$] (2) at (-1.0,1.5) {}; 
    \node[input, fill=bluishgreen, label=left:$1$] (1) at (-1.0,0.5) {}; 
    \node[draw=none] (X) at (-1.0,4.0) {$\cX$}; 
    \node[draw=none] (Y) at (2.0,4.0) {$\cY$}; 
    \node[draw=none] (Z) at (4.0,4.0) {$\cZ$}; 
    \node[draw=none] (V) at (5.5,4.0) {$\cV$}; 
    \node[draw=none] (V) at (5.5,2.5) {$0$}; 
    \node[draw=none] (V) at (5.5,0.5) {$1$}; 
    \node[vertex, label=left:$a$] (al) at (2,3) {}; 
    \node[vertex, label=right:$a$] (ar) at (4,3) {}; 
    \node[vertex, label=left:$b$] (bl) at (2,2) {}; 
    \node[vertex, label=right:$b$] (br) at (4,2) {}; 
    \node[vertex, label=left:$c$] (cl) at (2,1) {}; 
    \node[vertex, label=right:$c$] (cr) at (4,1) {}; 
    \node[vertex, label=left:$d$] (dl) at (2,0) {}; 
    \node[vertex, label=right:$d$] (dr) at (4,0) {}; 
    \path[ed0] (al) edge (ar);
    \path[ed1] (bl) edge (br);
    \path[ed0] (cl) edge (cr);
    \path[ed1] (dl) edge (dr);
    \path[ed2] (br) edge (al);
    \path[ed2] (dr) edge (cl);
    \path[pointer, in=-200, out=20] (0) edge[->]  node[pos=0.3, lbl] (t) {$\scriptstyle 1-p$} ($(al)+(1,0.05)$);
    \path[pointer, in=-190, out=-30] (0) edge[->]  node[pos=0.3, lbl] (t) {$\scriptstyle p$} ($(cl)+(1,0.05)$);
    \path[pointer, in=200, out=-20] (1) edge[->]  node[pos=0.3, lbl] (t) {$\scriptstyle 1-p$} ($(dl)+(1,-0.05)$);
    \path[pointer, in=190, out=30] (1) edge[->]  node[pos=0.3, lbl] (t) {$\scriptstyle p$} ($(bl)+(1,-0.05)$);
    \path[pointer, in=160, out=30] (2) edge[->]  node[pos=0.3, lbl] (t) {$\scriptstyle \frac{1}{2}$} ($(bl)+(1,0.5)$);
    \path[pointer, in=-160, out=-30] (2) edge[->]  node[pos=0.3, lbl] (t) {$\scriptstyle \frac{1}{2}$} ($(cl)+(1,-0.5)$);
 \end{tikzpicture}
}%
\resizebox{0.42\columnwidth}{!}{\input{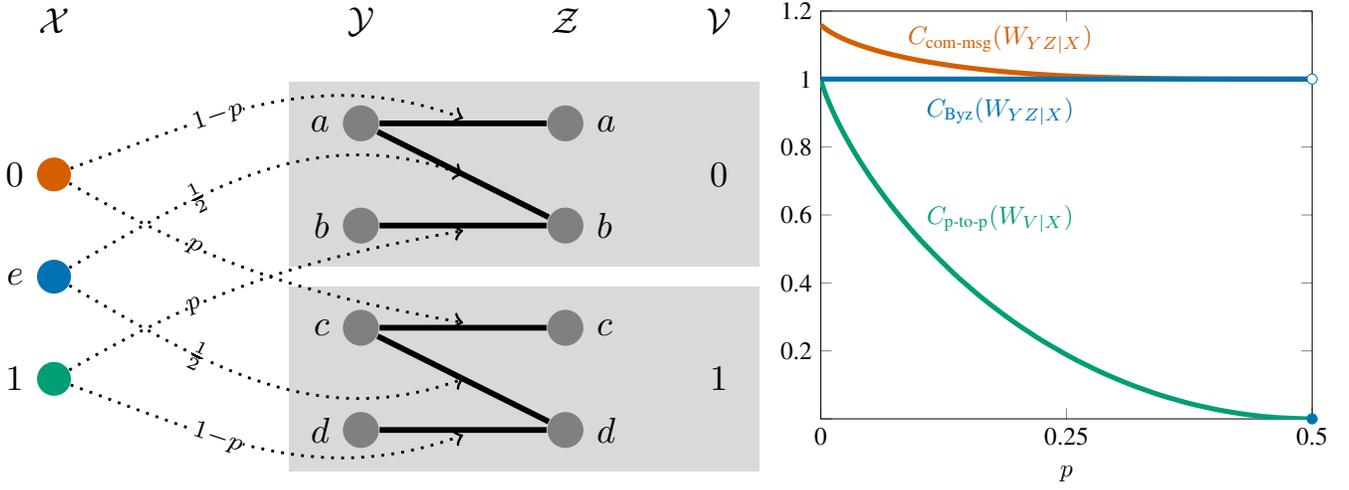}}
\caption{An example to show that $\capstr$ could be strictly in between the point-to-point capacity of the common channel and the common message capacity of the broadcast channel. For all values of $p$ except $p=0.5$, $\capstr=1$. For $p=0.5$, when the common channel is trivial, $\capstr=0$. Here $\cU=\{0,1\}\subsetneq\cX$ (except for $p=0.5$ when $\cU$ is a singleton).}\label{fig:bsc}
\end{figure}

\begin{theorem}\label{thm:Capacity}
The consensus capacity of $W_{YZ|X}$ is
\begin{align}
\capstr=\max_{P_U} \min_{\substack{P_{X|U}: P_{X|U}(x|u)>0\\\text{only if }\ch_{V|X}(.|x)=\ch_{V|X}(.|u)}} \min(I(U;Y),I(U;Z)),\label{eq:capstr}
\end{align} 
where the maximization is over p.m.f.s $P_U$ over the effective input alphabet $\cU$ and the mutual informations are evaluated under $P_{UXYZ}(u,x,y,z) = P_U(u)P_{X|U}(x|u)\ch_{YZ|X}(y,z|x)$.
\end{theorem}
We prove this in Section~\ref{sec:proof}. The expression for capacity can be interpreted as follows: Unlike the expression \eqref{eq:capcom} for $\capcom$, the input distribution $P_U$ avoids using letters $x$ with $\ch_{V|X}(.|x)$ which are not vertices of $\sS$. A byzantine sender can attack by replacing such an $x$ by sending a letter from $\cU$ picked according to an appropriate distribution that induces the same common channel output distribution. The minimization in the capacity expression represents a similar attack where, for each letter $u$, the sender randomly chooses among the letters $x$ which correspond to the same vertex in $\sS$ as $u$. It turns out that both these attacks cannot be detected by the receivers in a manner which permits consensus. 
As we show in Example~\ref{ex:Zexample} below, $\capstr$ may lie strictly in between $\capcom$ and the point-to-point capacity of the common channel $W_{V|X}$ (see Figure~\ref{fig:bsc}).
\begin{remark}\label{rem:commonrand}
Notice that when the common channel has capacity $\capptop(\ch_{V|X})=0$, $\cU$ is a singleton set and $\capstr=0$. Since $\capstr\geq\capptop(\ch_{V|X})$ (a fact which can also be verified from Theorem~\ref{thm:Capacity} using $V=\phi_1(Y)=\phi_2(Z)$), $\capstr>0$ if and only if $\capptop(\ch_{V|X})>0$. 
\end{remark}
\begin{remark}\label{rem:sharedrandomness}
It is easy to see from Theorem~\ref{thm:Capacity} that $\capstr$ remains unchanged if the receivers are provided additional correlated randomness unknown to the sender which they can use to coordinate their actions --- augment the channel outputs to $(Y,S)$ and $(Z,T)$ where $(S,T)$ is independent of $(X,Y,Z)$ and notice that $I(U;Y,S)=I(U;Y)$ and $I(U;Z,T)=I(U;Z)$. 
Recall that the converse of Theorem~\ref{thm:Capacity} is shown under $\err^{n}=o(1/n)$ (see Remark~\ref{rem:weak}).
However, we show in Appendix~\ref{app:weak} that for the example of Section~\ref{sec:shortconverse} which has $\capstr=0$ (in fact, even in a stronger sense; see Remark~\ref{rem:Witsenhausen}), when common randomness unknown to the sender is available to the decoders, a positive rate can be achieved with $\err^{(n)}=o(n^{-\frac{1}{2}+\epsilon})$, for any $\epsilon>0$.  
\end{remark}

\begin{remark}\label{rem:uniquenessofU}
It is clear from \eqref{eq:capstr} that if there was a choice in selecting $\cU$, the expression for capacity on the right-hand-side of \eqref{eq:capstr} does not depend on this choice --  the minimization is over $P_{X|U}$ such that, for each $u$, $P_{X|U}(.|u)$ has support (only) over all letters $x\in\cX$ which correspond to the same vertex of the polytope $\sS$ as $u$. 
\end{remark}

\begin{example}[Capacity of the channel in Figure~\ref{fig:bsc}]\label{ex:Zexample}
Let $p\in[0,0.5]$. Consider the following channel with $\cX=\{0,e,1\}$, $\cY=\cZ=\{a,b,c,d\}$
\[ \ch_{YZ|X}(y,z|x)=\begin{cases} 
                             1-p, & (x,y,z)\in\{(0,a,a),(1,d,d)\}\\
                             p, & (x,y,z)\in\{(0,c,c),(1,b,b)\}\\
                             1/2, & (x,y,z)\in\{(e,a,b),(e,c,d)\}
                     \end{cases}
\]
The characteristic graph of $\ch_{YZ|X}$ in Figure~\ref{fig:bsc} has two connected components: $G_0$ on the vertices $\{a,b\}\cup\{a,b\}$ and $G_1$ on $\{c,d\}\cup\{c,d\}$. Denote the connected components by $\cV=\{0,1\}$.
The common channel is $\ch_{V|X}$ with $\ch_{V|X}(1-x|x)=1-\ch_{V|X}(x|x)=p$ for $x\in\{0,1\}$, and $\ch_{V|X}(0|e)=\ch_{V|X}(1|e)=\frac{1}{2}$.
Hence, the capacity of the common channel is that of a binary symmetric channel with cross over probability $p$
\[ \capptop(\ch_{V|X})=1-H(p).\]

The consensus capacity of this channel is computed as follows: For $p\neq 0.5$, the vertices of the convex hull of $\{\ch_{V|X}(\cdot|x), x\in\{0,1,e\}\}$ correspond to the symbols $0$ and $1$ of $\cX$, i.e., $\cU=\{0,1\}$, since
\[ \ch_{V|X}(v|e) = \frac{1}{2}\ch_{V|X}(v|0) + \frac{1}{2}\ch_{V|X}(v|1), v\in\cV.\]
For $u\in\{0,1\}$, $\ch_{V|X}(\cdot|x)=\ch_{V|U}(\cdot|u)$ if and only if $x=u$.
Hence, the consensus capacity is \eqref{eq:capstr} evaluated under $P_{UYZ}(u,y,z)=P_U(u)\ch_{YZ|X}(y,z|u)$.
\begin{align*}
\capstr &=\max_{P_U} \min(I(U;Y),I(U;Z))\\
        &\stackrel{\text{(a)}}{=} \max_{P_U} H(U)\\
        &= 1,
\end{align*}
where (a) follows from the fact that $H(U|Y)=H(U|Z)=0$ under $P_{UYZ}(u,y,z)=P_U(u)\ch_{YZ|X}(y,z|u)$.
For $p=0.5$, the vectors $\ch_{V|X}(.|x)$ are identical for $x\in\cX$ and the polytope collapses to a point. Then, $\cU$ is singleton and $\capstr=0$.

Finally, we compute the common message capacity 
\[ \capcom=\max_{P_X}\min(I(X;Y),I(X;Z)).\]
Suppose $P_X$ is a maximizer. By symmetry, $P'_X$ such that $P'_X(x)=P_X(1-x)$ for $x\in\{0,1\}$ and $P'_X(e)=P_X(e)$ also achieves the maximum. Furthermore, since $I(X;Y)$ and $I(X;Z)$ are concave functions of $P_X$ and so is $\min(I(X;Y),I(X;Z))$, any convex combination (specifically, the uniform convex combination) of $P_X$ and $P'_X$ also achieves the maximum.
Hence, without loss of generality, we make take the maximizing $P_X$ to be of the form $P_X(0)=P_X(1)=\frac{1-q}{2}, P_X(e)=q$, where $q\in[0,1]$. With the mutual informations evaluated under this,
    \begin{align*}
        \capcom&=\max_{q\in[0,1]} \min(I(X;Y),I(X;Z))\\
        &\stackrel{(a)}{=} \max_{q\in[0,1]} I(X;Y)\\
        &\stackrel{(b)}{=} \max_{q\in[0,1]} I(X;Y, 1_{Y\in\{b,d\}})\\
        &= \max_{q\in[0,1]} I(X;1_{Y\in\{b,d\}}) + I(X;Y|1_{Y\in\{b,d\}})\\
        &\stackrel{(c)}{=} \max_{q\in[0,1]} I(X;1_{X=1}) + P(X\neq 1) I(X;Y|X\neq 1)\\
        &\stackrel{(c)}{=} \max_{q\in[0,1]} H(1_{X=1}) + P(X\neq 1) I(X;Y|X\neq 1)\\
        &= \max_{q\in[0,1]} H_2\left(\frac{1-q}{2}\right) +
           \frac{1+q}{2} \bigg(H(Y|X\neq 1) -\\ &\qquad\qquad\quad  \sum_{x\in\{0,e\}} P(X=x|X\neq 1) H(Y|X=x)\bigg)\\
        &= \max_{q\in[0,1]} H_2\left(\frac{1-q}{2}\right) +
           \frac{1+q}{2} \bigg(H_2\left(\frac{p(1-q)+q}{1+q}\right) \\ &\qquad\qquad\qquad\qquad\qquad-\left(\frac{1-q}{1+q}H_2(p) + \frac{2q}{1+q}\right)\bigg),
    \end{align*}
    where (a) follows from the fact that $I(X;Y)=I(X;Z)$ when $P_X$ is of the form we are working with,
    (b) from the fact that the indicator function $1_{Y\in\{b,d\}}$ is a function of $Y$,
    and, (c) follows from the fact that $(Y\in\{b,d\})$ if and only if $(X=1)$.
We evaluate this expression numerically and compare with the other capacities in Figure~\ref{fig:bsc}.
\end{example}


\section{Proof of Theorem~\ref{thm:Capacity}}\label{sec:proof}
Before proving the converse (in Section~\ref{sec:converse}) and achievability (in Section~\ref{sec:achievability}) of Theorem~\ref{thm:Capacity}, we state a lemma which is an immediate consequence of the fact that a vertex of a convex polytope (specifically, the common channel polytope $\sS$) cannot be expressed as a non-trivial convex combination of other vertices or interior points of the polytope. We prove the lemma in Appendix~\ref{app:polytope}.
\begin{lemma}\label{lem:polytope}
Consider a broadcast channel $W_{YZ|X}$ with $W_{V|X}$ as its common channel (Definition~\ref{def:common-channel}) and $\cU$ as a common message alphabet (Definition~\ref{def:effective-input}). 
Let $\lambda_x,x\in\cX$ be a p.m.f., i.e., $\lambda_x\geq0, x\in\cX$ and $\sum_{x\in\cX} \lambda_x=1$. Suppose for some $u\in\cU$,
\begin{align}\sum_{x\in\cX} \lambda_xW_{V|X}(v|x)= \ch_{V|X}(v|u), \text{ for all } v\in\cV.\label{lemma5:main_statement}
\end{align}
Then $\lambda_x>0$ only if $\ch_{V|X}(v|x)=\ch_{V|X}(v|u)$ for all $v\in \cV$.
Additionally, the representation channel $\tP_{U|X}(u|x)$ in Definition~\ref{def:effective-input} is such that for $u\in \cU$
\begin{align}
\tP_{U|X}(u|x) = 1 \text{ if and only if } x\in \cX_u \label{eq:cX_u}
\end{align} where 
$\cX_u =\{x\in \cX:W_{V|X}(v|x)=W_{V|X}(v|u),\, \forall\,v\in \cV\}.$ In particular, for $\tilde{u}, u'\in \cU$,
\begin{align}\tP_{U|X}(\tilde{u}|u') = 1_{\inb{\tilde{u} = u'}}.\label{eq:cX_u_particular}\end{align}
\end{lemma}

\subsection{Converse of Theorem~\ref{thm:Capacity}}\label{sec:converse}
The following is the main converse claim.
\begin{claim}\label{cl:converseclaim}
\begin{align}
\capstr\leq \max_{P_X} \min_{P_{X'|X}\in\cP_{X'|X}}  \min(I(X;Y),I(X;Z)),\label{eq:converseclaim1}
\end{align}
where $\cP_{X'|X}$ is the set of all $P_{X'|X}$ such that 
\begin{align}\sum_{x'\in\cX} P_{X'|X}(x'|x)W_{V|X}(v|x')= \ch_{V|X}(v|x)\label{eq:cP-X'|X} 
\end{align}
for all $v\in\cV,x\in\cX$, and the mutual informations are evaluated under the joint
distribution\\ $P_{XX'YZ}(x,x',y,z)=P_X(x)P_{X'|X}(x'|x)\ch_{YZ|X}(y,z|x')$.
\end{claim}
Before proving the claim, we will argue that it implies the converse of Theorem~\ref{thm:Capacity}. To this end, we now argue that the maximization in \eqref{eq:converseclaim1} may be restricted to $P_X$ which have support only over $\cU$ and, hence,
\begin{align}
\capstr\leq \max_{P_U} \min_{P_{X'|U}\in\cP_{X'|U}} \min(I(U;Y),I(U;Z)),\label{eq:converseclaim2}
\end{align}
where $\cP_{X'|U}$ is defined analogously to \eqref{eq:cP-X'|X}, i.e., $\cP_{X'|U}$ is the set of all $P_{X'|U}$ such that
\[ \sum_{x'\in\cX} P_{X'|X}(x'|x)W_{V|X}(v|x')= \ch_{V|X}(v|x), \quad v\in\cV,u\in\cU,\] 
and the joint distribution $P_{UX'YZ}$ for the RHS of \eqref{eq:converseclaim2} is given by \[P_{UX'YZ}(u,x',y,z)=P_U(u) P_{X'|U}(x'|u) \ch_{YZ|X}(y,z|x').\]
Given any $P_X$, consider $\tP_{XU}$ defined by $\tP_{XU}(x,u)=P_X(x)\tP_{U|X}(u|x)$, where $\tP_{U|X}$ is the \representationchannel{} in~Definition~\ref{def:effective-input} (see~\eqref{eq:Ptilde}). We will show that under the induced marginal distribution $\tP_U$,
\begin{align*} &\min_{P_{X'|U}\in\cP_{X'|U}} \min(I(U;Y),I(U;Z)) \\
 &\qquad\qquad\geq \min_{P_{X'|X}\in\cP_{X'|X}}  \min(I(X;Y),I(X;Z)), \end{align*}
where the RHS is evaluated with $P_X$. To do this, suppose $\pP_{X'|U}\in\cP_{X'|U}$ is a minimizer for the LHS. We need only show that there is a $P_{X'|X}\in\cP_{X'|X}$ such that $I(X;Y)$ (resp., $I(X;Z)$) is no larger than $I(U;Y)$ (resp., $I(U;Z)$) of the LHS. Consider the $P_{XUX'YZ}$ defined by
\begin{align} 
&P_{XUX'YZ}(x,u,x',y,x)\notag\\
& = P_X(x)\tP_{U|X}(u|x)\pP_{X'|U}(x'|u)\ch_{YZ|X}(y,z|x')
\label{eq:converse-x'|x}
\end{align}
which induces the following $P_{X'|X}$
\begin{align*}
P_{X'|X}(x'|x)=\sum_{u} \tP_{U|X}(u|x)\pP_{X'|U}(x'|u), x,x'\in\cX.
\end{align*}
We have $P_{X'|X}\in\cP_{X'|X}$ (see \eqref{eq:cP-X'|X}) since, for all $v,x$,
\begin{align*}
&\sum_{x'} P_{X'|X}(x'|x)\ch_{V|X}(v|x')\\
&= \sum_{x'} \left(\sum_{u}\tP_{U|X}(u|x)\pP_{X'|U}(x'|u)\right)\ch_{V|X}(v|x')\\
&= \sum_u \tP_{U|X}(u|x) \left(\sum_{x'} \pP_{X'|U}(x'|u)\ch_{V|X}(v|x')\right)\\
&\stackrel{\text{(a)}}{=} \sum_u \tP_{U|X}(u|x) \ch_{V|X}(v|u)\\
&\stackrel{\text{(b)}}{=} \ch_{V|X}(v|x),
\end{align*}
where (a) follows from $\pP_{X'|U}\in\cP_{X'|U}$ and (b) from~\eqref{eq:Ptilde}. Further, since $X-U-(Y,Z)$ is a Markov chain under \eqref{eq:converse-x'|x}, we have $I(X;Y)\leq I(U;Y)$ and $I(X;Z)\leq I(U;Z)$. Thus, we have shown~\eqref{eq:converseclaim2}. The converse now follows from observing that if $P_{X'|U}\in\cP_{X'|U}$, then, by \eqref{lemma5:main_statement} of Lemma~\ref{lem:polytope}, $P_{X'|U}(x|u)>0$ only if $\ch_{V|X}(.|x)=\ch_{V|X}(.|u)$. 

It only remains to prove Claim~\ref{cl:converseclaim}. Suppose $(\enc,\decB,\decC)$ is an $(n,2^{nR})$ byzantine agreement code with error probability $\err$. Consider any $P_{X'|X}\in\cP_{X'|X}$. We shall argue that for a uniformly chosen message $M$, if the sender passes the codeword $X^n=f(M)$ through the discrete memoryless channel (DMC) $P_{X'|X}$ and sends the output $X'^n$ through the broadcast channel $\ch_{YZ|X}$, the decoders acting on the output vectors $Y^n$ and $Z^n$ must output the message $M$ with sufficiently high probability. Specifically, we will show the following:
\begin{claim}\label{cl:converse2ndclaim}
Under the above experiment,
\[  \Pr(\decB(Y^n)=\decC(Z^n)=M) \geq 1 - (n+1)\rho\err,\]
where 
\begin{align}
\rho = \frac{|\cY|+|\cZ|-1}{\min_{(x,y,z):\ch_{YZ|X}(y,z|x)>0} W_{YZ|X}(y,z|x)}. \label{eq:rho}
\end{align}
\end{claim}
Assuming Claim~\ref{cl:converse2ndclaim} for the moment, suppose $R$ is achievable, i.e., there is a sequence of $(n,2^{nR})$ byzantine agreement codes such that $\err^{(n)}=o(1/n)$. Then, under the above experiment, $\Pr(\decB(Y^n)\neq M)$ and $\Pr(\decC(Z^n)\neq M)$, which are upper bounded by $(n+1)\rho\err^{(n)}$, approach 0 as $n\rightarrow\infty$. Note that here we make use of the requirement that the error probability $\err^{(n)}$ must fall super-linearly in the blocklength $n$. Then, using Fano's inequality and following standard single-letterization steps (e.g.,~\cite[Sec.~7.9]{CoverT2006}),
\begin{align*}
R &\leq \frac{1}{n}\sum_{i=1}^nI(X_i;Y_i) + R\Pr(\decB(Y^n)\neq M) + \frac{1}{n}\\
  &\leq \frac{1}{n}\sum_{i=1}^nI(X_i;Y_i) + (n+1)R\rho\err^{(n)} + \frac{1}{n}.
\intertext{Similarly,}
R &\leq \frac{1}{n}\sum_{i=1}^nI(X_i;Z_i) + (n+1)R\rho\err^{(n)} + \frac{1}{n}.
\end{align*}
Since $\err^{(n)}=o(1/n)$, for any $\epsilon>0$, there is a sufficiently large $n$ such that
\begin{align*}
R &\leq \min\left( \frac{1}{n}\sum_{i=1}^nI(X_i;Y_i), \frac{1}{n}\sum_{i=1}^nI(X_i;Z_i) \right) + \epsilon\\
  &\leq \min\left( I(X;Y), I(X;Z) \right) + \epsilon,
\end{align*}
where we used Jensen's inequality in the last step whose RHS is evaluated under the joint distribution
\[ P_{XX'YZ}(x,x',y,z) = P_X(x) P_{X'|X}(x'|x) \ch_{YZ|X}(y,z|x')\]
with $P_X(x) = \frac{1}{n}\sum_{i=1}^n \Pr(X_i=x)$ and $X^n=\enc(M)$ where $M$ is a uniformly chosen message. Since this holds for any choice of $P_{X'|X}\in\cP_{X'|X}$, Claim~\ref{cl:converseclaim} follows.

We now prove Claim~\ref{cl:converse2ndclaim} to complete the proof of Claim~\ref{cl:converseclaim}.
\begin{proof}[Proof of Claim~\ref{cl:converse2ndclaim}]
It will suffice to show that, for each message $m\in\{1,\ldots,2^{nR}\}$, under the experiment in which the codeword $\enc(m)$ is first sent over a DMC $P_{X'|X}\in\cP_{X'|X}$ and its output sent over $\ch_{YZ|X}$ to produce $\hY^n,\hZ^n$,
\begin{align} \Pr(\decB(\hY^n)=\decC(\hZ^n)=m) \geq 1 - (n+1)\rho\err. \label{eq:step-n}\end{align}
Notice that, in this proof, we denote $Y^n,Z^n$ in the statement of Claim~\ref{cl:converse2ndclaim} by $\hY^n,\hZ^n$; and the unhatted versions refer to the result of sending the codeword as it is into the channel as follows: 
Let $\enc(m)=x^n=(x_1,x_2,\ldots,x_n)$. Let $Y^n,Z^n$ be the outputs of the channel when $x^n$ is transmitted, i.e.,
\begin{align}
 P_{Y^n,Z^n}(y^n,z^n)=\prod_{i=1}^n \ch_{YZ|X}(y_i,z_i|x_i).\label{eq:unhattedmarginals}
\end{align}
We are given (by \eqref{eq:errordefn})
\begin{align}
\Pr((\decB(Y^n)=\decC(Z^n)=m) \geq 1 - \err. \label{eq:step-0}
\end{align}
Let $\hX^n$ denote the output resulting from sending $x^n$ over the DMC $P_{X'|X}\in\cP_{X'|X}$, and let $\hY^n,\hZ^n$ denote the channel outputs from sending this $\hX^n$ over $\ch_{YZ|X}$. i.e.,
\begin{align}
P_{\hX^n,\hY^n,\hZ^n}&(\hx^n,\hy^n,\hz^n)\notag\\
 &= \prod_{i=1}^n P_{X'|X}(\hx_i|x_i)\ch_{YZ|X}(\hy_i,\hz_i|\hx_i).\label{eq:hattedmarginals}
\end{align}
We show \eqref{eq:step-n} by showing the following for $\epsilon>0$ and $k\in[1:n]$: 
\begin{align}
\intertext{If }
\Pr(&\decB((\hY_1,\ldots,\hY_{k-1},Y_k,\ldots,Y_n))=\notag\\
  &\decC((\hZ_1,\ldots,\hZ_{k-1},Z_k,\ldots,Z_n))=m) \geq 1 - \epsilon,\label{eq:step-t-1}\\
\intertext{then}
\Pr(&\decB((\hY_1,\ldots,\hY_{k},Y_{k+1},\ldots,Y_n))=\notag\\
  &\decC((\hZ_1,\ldots,\hZ_{k},Z_{k+1},\ldots,Z_n))=m) \geq 1 - \epsilon - \err\rho.\label{eq:step-t}
\end{align}
Since \eqref{eq:step-0} implies \eqref{eq:step-t-1} for $k=1$ with $\epsilon=\rho\err$ (since $\rho\geq 1$), applying the above recursively for $k=1,\ldots,n$ will give \eqref{eq:step-n}. Note that the joint distribution of the random variables in \eqref{eq:step-t-1} is
\begin{align}
&p_{\hY_1^{k-1},Y_k^n,\hZ_1^{k-1},Z_k^n}(\hy_1^{k-1},y_k^n,\hz_1^{k-1},z_k^n)\notag\\
 &=\prod_{i=1}^{k-1} \left(\sum_{\hx_i} P_{X'|X}(\hx_i|x_i)\ch_{YZ|X}(\hy_i,\hz_i|\hx_i)\right)\notag\\
     &\qquad\prod_{i=k}^n \ch_{YZ|X}(y_i,z_i|x_i), \label{eq:t-1marginals}
\end{align} 
while that of the random variables in \eqref{eq:step-t} is
\begin{align}
&p_{\hY_1^{k},\hZ_1^{k},Y_{k+1}^n,Z_{k+1}^n}(\hy_1^{k},\hz_1^{k},y_{k+1}^n,z_{k+1}^n)\notag\\
 &=\prod_{i=1}^{k} \left(\sum_{\hx_i} P_{X'|X}(\hx_i|x_i)\ch_{YZ|X}(\hy_i,\hz_i|\hx_i)\right) \notag\\
    &\qquad \prod_{i=k+1}^n \ch_{YZ|X}(y_i,z_i|x_i). \label{eq:tmarginals}
\end{align} 
To show that \eqref{eq:step-t-1} implies \eqref{eq:step-t}, we define the following joint distribution of these random variables
\begin{align}
&p_{\hY_1^{k},\hZ_1^{k},Y_{k}^n,Z_{k}^n}(\hy_1^{k},\hz_1^{k},y_{k}^n,z_{k}^n)\notag\\
 &=\prod_{i=1}^{k-1} \left(\sum_{\hx_i} P_{X'|X}(\hx_i|x_i)\ch_{YZ|X}(\hy_i,\hz_i|\hx_i)\right)
 \notag\\&\qquad
 P_{\hY_k,\hZ_k,Y_k,Z_k}(\hy_k,\hz_k,y_k,z_k) 
 \prod_{i=k+1}^n \ch_{YZ|X}(y_i,z_i|x_i),\label{eq:coupling-long}
\end{align}
where the coupling $P_{\hY_k,\hZ_k,Y_k,Z_k}$ is given by
\begin{align}
&P_{\hY_k,\hZ_k,Y_k,Z_k}(\hy_k,\hz_k,y_k,z_k)\notag\\
&=\sum_{\hx_k} P_{X'|X}(\hx_k|x_k) \sum_{v} \ch_{V|X}(v|\hx_k)\notag\\
&\qquad\qquad \ch_{YZ|XV}(\hy_k,\hz_k|\hx_k,v)\ch_{YZ|XV}(y_k,z_k|x_k,v).\label{eq:coupling-short}
\end{align}
i.e., under the coupling, both $(\hY_k,\hZ_k)$ and $(Y_k,Z_k)$ have the same common channel output. We first demonstrate that this is a valid coupling by verifying that it has the correct marginals for $(\hY_k,\hZ_k)$ and $(Y_k,Z_k)$.
\begin{align*}
&P_{\hY_k,\hZ_k}(\hy_k,\hz_k)\\
 &=\sum_{\hx_k} P_{X'|X}(\hx_k|x_k) \sum_{v} \ch_{V|X}(v|\hx_k) \ch_{YZ|XV}(\hy_k,\hz_k|\hx_k,v)\\
 &=\sum_{\hx_k} P_{X'|X}(\hx_k|x_k)\ch_{YZ|X}(\hy_k,\hz_k|\hx_k),
\end{align*}
which matches \eqref{eq:tmarginals}; we used Definition~\ref{def:common-channel} in the last step above. To verify that the marginals of $(Y_k,Z_k)$ in the coupling match that in \eqref{eq:t-1marginals},
\begin{align*}
&P_{Y_k,Z_k}(y_k,z_k)\\
 &=\sum_{\hx_k} P_{X'|X}(\hx_k|x_k) \sum_{v} \ch_{V|X}(v|\hx_k) \ch_{YZ|XV}(y_k,z_k|x_k,v)\\
 &=\sum_{v} \left(\sum_{\hx_k} P_{X'|X}(\hx_k|x_k) \ch_{V|X}(v|\hx_k)\right) \ch_{YZ|XV}(y_k,z_k|x_k,v)\\
 &\stackrel{\text{(a)}}{=}\sum_{v} \ch_{V|X}(v|x_k) \ch_{YZ|XV}(y_k,z_k|x_k,v)\\
 &\stackrel{\text{(b)}}{=}\ch_{YZ|X}(y_k,z_k|x_k),
\end{align*}
where (a) follows from \eqref{eq:cP-X'|X} since $P_{X'|X}\in\cP_{X'|X}$ and (b) from Definition~\ref{def:common-channel}.
Now \eqref{eq:step-t} will follow from \eqref{eq:step-t-1} if we show that, under the joint distribution of \eqref{eq:coupling-long},
\begin{align}
\Pr\big(&\decB((\hY_1,\ldots,\hY_{k-1},Y_k,Y_{k+1}\ldots,Y_n))=\notag\\
  &\decC((\hZ_1,\ldots,\hZ_{k-1},Z_k,Z_{k+1}\ldots,Z_n))=\notag\\
  &\decB((\hY_1,\ldots,\hY_{k-1},\hY_k,Y_{k+1},\ldots,Y_n))=\notag\\
  &\decC((\hZ_1,\ldots,\hZ_{k-1},\hZ_k,Z_{k+1},\ldots,Z_n))\big) \geq 1 - \err\rho \label{eq:coupled}
\end{align}
since, if we take $A,B,C$ to be the events whose probabilities are on the left-hand sides of \eqref{eq:step-t-1}, \eqref{eq:step-t} and \eqref{eq:coupled} resp., we have $B\supseteq A\cap C$, and hence $\Pr(B) \geq 1 - \Pr(A^c) - \Pr(C^c) = \Pr(A) + \Pr(C) - 1$. 

To this end, we state the following lemma which will be proved later.
\begin{lemma}\label{lem:hybrid}
Let $P_{ST}$ be a joint distribution over $\cS\times\cT$, $\delta>0$, and $\psi_{\sB}$ and $\psi_{\sC}$ be functions defined on ${\cY}\times\cS$ and $\cZ\times \cT$, respectively. Under the joint distribution 
    \begin{align} P_{STYZ|X}(s,t,y,z|x) = P_{ST}(s,t)\ch_{YZ|X}(y,z|x), \label{eq:hybrid-distr}
    \end{align}
suppose for every $x\in {\cX}$, 
    \begin{align}
        \Pr(\psi_{\sB}(Y,S)\neq \psi_{\sC}(Z,T)| X=x)\leq \delta. \label{eq:hybrid-error}
    \end{align} 
Then, for every $v\in \cV$, 
    \begin{align}
       \Pr\left( \left|\bigcup_{\{y,z\}\in\cE_v} \left\{\psi_{\sB}(y,S),\psi_{\sC}(z,T)\right\} \right| = 1\right)
       &\geq 1- \delta\rho, \label{eq:hybrid}
    \end{align}
where the probability is over $(S,T)\sim P_{ST}$ and $\rho$ is given by \eqref{eq:rho}. Recall that $\cE_v$ is the edge set of the connected component (of the characteristic graph of the broadcast channel $\ch_{YZ|X}$) corresponding to the letter $v\in\cV$ of the common channel output alphabet (Definition~\ref{def:common-channel}).
\end{lemma}
To show \eqref{eq:coupled} (under the joint distribution of \eqref{eq:coupling-long}), we invoke Lemma~\ref{lem:hybrid} with $S=(\hY_1^{k-1},Y_{k+1}^n)$, $T=(\hZ_1^{k-1},Z_{k+1}^n)$, and 
\begin{align*}
\psi_{\sB}(y,\inp{\hy_1^{k-1},y_{k+1}^n)}&=\decB(\hy_1^{k-1},y,y_{k+1}^n),\\
\psi_{\sC}(z,(\hz_1^{k-1},z_{k+1}^n))&=\decC(\hz_1^{k-1},z,z_{k+1}^n).
\end{align*}
With $\delta=\err$, as we argue below, \eqref{eq:hybrid-error} follows from the fact the decoder outputs must agree with probability at least $1-\err$ for all inputs and, specifically, the input $(\hX_1,\ldots,\hX_{k-1},x,x_{k+1},\ldots,x_n)$. i.e., 
\begin{align*}
&\Pr(\psi_{\sB}(Y,S)\neq \psi_{\sC}(Z,T)| X=x)\\
&\stackrel{\text{(a)}}{=}\sum_{\hx_1,\ldots,\hx_{k-1}} 
 \left(\prod_{i=1}^{k-1} P_{X'|X}(\hx_i|x_i)\right)\\
 &\qquad\qquad\qquad\Pr(\decB(Y^n)\neq\decC(Z^n)|X^n=(\hx_1^{k-1},x,x_k^n))\\
&\stackrel{\text{(b)}}{\leq}\sum_{\hx_1,\ldots,\hx_{k-1}} 
 \left(\prod_{i=1}^{k-1} P_{X'|X}(\hx_i|x_i)\right) \err\\
&=\err,
\end{align*}
where the probability in the RHS of (a) is over the $n$-uses of the memoryless channel $\ch_{YZ|X}$ and (b) follows from \eqref{eq:errordefn}. We need to show \eqref{eq:coupled} (under the joint distribution of \eqref{eq:coupling-long}) which translates to
\begin{align}
\Pr\big(&\psi_B(Y_k,S)=
\psi_C(Z_k,T)
\notag\\&
=\psi_B(\hY_k,S)=
\psi_C(\hZ_k,T)\big)
 \geq 1 - \err\rho, \label{eq:coupled-tr}
\end{align}
where 
$(S,T)$ is independent of $(\hY_k,\hZ_k,Y_k,Z_k)$ with $P_{\hY_k,\hZ_k,Y_k,Z_k}$ given by \eqref{eq:coupling-short}.
Let us define the events 
$\cA_{\hy_k,\hz_k,y_k,z_k}=\{\psi_B(\hy_k,S)= \psi_C(\hz_k,T) =\psi_B(y_k,S)= \psi_C(z_k,T)\big\}$.
\begin{align*}
&\Pr(\psi_B(Y_k,S)= \psi_C(Z_k,T) =\psi_B(\hY_k,S)= \psi_C(\hZ_k,T))\\
&\stackrel{\text{(a)}}{=}\sum_{\hy_k,\hz_k,y_k,z_k} P_{\hY_k,\hZ_k,Y_k,Z_k}(\hy_k,\hz_k,y_k,z_k) \Pr(\cA_{y_k,z_k,\hy_k,\hz_k})\\
&\stackrel{\text{(b)}}{=}\sum_{\hy_k,\hz_k,y_k,z_k} \sum_{\hx_k,v} P_{X'|X}(\hx_k|x_k)\ch_{V|X}(v|\hx_k)\\
     &\qquad\qquad\qquad
         \ch_{YZ|XV}(y_k,z_k|x_k,v) \ch_{YZ|XV}(\hy_k,\hz_k|\hx_k,v)\\
     &\qquad\qquad\qquad\qquad
          \Pr(\cA_{y_k,z_k,\hy_k,\hz_k})\\
&=\sum_{\hx_k,v} P_{X'|X}(\hx_k|x_k)\ch_{V|X}(v|\hx_k)\\
     &\qquad
       \sum_{\hy_k,\hz_k,y_k,z_k} 
         \ch_{YZ|XV}(y_k,z_k|x_k,v) \ch_{YZ|XV}(\hy_k,\hz_k|\hx_k,v)\\
     &\qquad\qquad\qquad\qquad
          \Pr(\cA_{y_k,z_k,\hy_k,\hz_k})\\
&\stackrel{\text{(c)}}{=}\sum_{\hx_k,v} P_{X'|X}(\hx_k|x_k)\ch_{V|X}(v|\hx_k)\\
     &
       \sum_{\{\hy_k,\hz_k\},\{y_k,z_k\}\in\cE_v} 
         \ch_{YZ|XV}(y_k,z_k|x_k,v) \ch_{YZ|XV}(\hy_k,\hz_k|\hx_k,v)\\
     &\qquad\qquad\qquad\quad\;
          \Pr(\cA_{y_k,z_k,\hy_k,\hz_k})\\
&\stackrel{\text{(d)}}{\geq} \sum_{\hx_k,v} P_{X'|X}(\hx_k|x_k)\ch_{V|X}(v|\hx_k)\\
     &
       \sum_{\{\hy_k,\hz_k\},\{y_k,z_k\}\in\cE_v} 
         \ch_{YZ|XV}(y_k,z_k|x_k,v) \ch_{YZ|XV}(\hy_k,\hz_k|\hx_k,v)\\
     &\qquad\qquad\qquad\quad\;
          (1-\err\rho)\\
&=1-\err\rho,
\end{align*}
where (a) follows from the independence of $(S,T)$ and $(\hY_k,\hZ_k,Y_k,Z_k)$ (notice that the probability of $\Pr(\cA_{y_k,z_k,\hy_k,\hz_k})$ is over the distribution of $(S,T)$), (b) from \eqref{eq:coupling-short}, and (c) from the fact that $\ch_{YZ|XV}(y,z|x,v)>0$ only if the edge $\{y,z\}$ lies in the edge set $\cE_v$ of the connected component $G_v$ corresponding to the common channel output letter $v$. Inequality (d) follows from Lemma~\ref{lem:hybrid} which, as discussed, we may invoke with $\delta=\err$. By~\eqref{eq:hybrid}, we may conclude that $\Pr(\cA_{y_k,z_k,\hy_k,\hz_k})\geq 1-\err\rho$ since under the event in \eqref{eq:hybrid} all $\{y,z\}\in\cE_v$ result in the same output. This completes the proof of Claim~\ref{cl:converse2ndclaim}.
\end{proof}
It only remains to prove Lemma~\ref{lem:hybrid}.
\proof[Proof of Lemma~\ref{lem:hybrid}]{
For $x\in\cX$, by \eqref{eq:hybrid-error},
    \begin{align*}
        \delta &\geq \Pr(\psi_{\sB}(Y,S)\neq \psi_{\sC}(Z,T)| X=x)\\
        &\stackrel{(a)}{=} \sum_{y,z} \ch_{YZ|X}(y,z|x) \Pr(\psi_{\sB}(S,y)\neq\psi_{\sC}(T,z)|X=x)\\
        &\stackrel{(b)}{=} \sum_{y,z} \ch_{YZ|X}(y,z|x) \Pr(\psi_{\sB}(S,y)\neq\psi_{\sC}(T,z)),
    \end{align*}
    where (a) and (b) use the description of $P_{STYZ|X}$ given in \eqref{eq:hybrid-distr}.
    Hence, for every edge $\{y,z\}\in\cE$ in the characteristic graph (i.e., $\ch_{YZ|X}(y,z|x)>0$ for some $x$), 
    \begin{align}
        &\Pr(\psi_{\sB}(S,y)\neq\psi_{\sC}(T,z)) \notag\\
            &\qquad\leq \frac{\delta}{\min_{(x,y,z):\ch_{YZ|X}(y,z|x)>0} \ch_{YZ|X}(y,z|x)}.\label{eq:hybrid-edge}
    \end{align}
    For every $v\in\cV$, since its corresponding connected component $G_v(\cN_v,\cE_v)$ is connected, it has a spanning tree, say, $G'_v(\cN_v,\cE'_v)$, $\cE'_v\subseteq\cE_v$.
    \begin{align*}
        &\Pr\left( \left|\bigcup_{\{y,z\}\in\cE_v} \left\{\decB(S,y),\decC(T,z)\right\} \right| = 1\right)\\
        &\stackrel{(a)}{=} \Pr\left( \left|\bigcup_{\{y,z\}\in\cE'_v} \left\{\decB(S,y),\decC(T,z)\right\} \right| = 1\right)\\
        &\stackrel{(b)}{\ge} 1-\sum_{\{y,z\}\in\cE'_v} \Pr\left( \decB(S,y)\neq\decC(T,z)\right)\\
        &\stackrel{(c)}{\ge} 1-\frac{\delta|\cE'|}{\min_{(x,y,z):\ch_{YZ|X}(y,z|x)>0} \ch_{YZ|X}(y,z|x)}\\
        &\stackrel{(d)}{\ge} 1-\delta\rho.
    \end{align*}
    Here, (a) follows from $G'_v$ being a spanning tree of $G_v$; (b) follows from a union bound; (c) follows from \eqref{eq:hybrid-edge}; and, (d) follows from $G'_v$ being a tree and hence $|\cE'_v|=\cN_v-1\le|\cX|+|\cY|-1$.
    This concludes the proof.
}

\subsection{Achievability of Theorem~\ref{thm:Capacity}}\label{sec:achievability}

We use the following notation for this section. Some of it is repeated from Section~\ref{sec:notation} for ready reference.\\
\underline{Notation specific to this subsection}:  Random variables are denoted by capital letters like $X, X', Y,$ etc. The corresponding alphabets are denoted by calligraphic letters in the same format, for example, the random variables $X$ and $X'$ have alphabet $\cX$. Its $n-$product set is denoted by $\cX^n$. For compactness, We use bold faced letters to denote $n-$length vectors, for example $\vecx$   denotes a vector in $\cX^n$ and $\vecX$ denotes a random vector in $\cX^n$. 
For an alphabet $\cX$, let $\mathcal{P}^n\inp{\cX}$ denote the set of all empirical distributions (types) of $n$ length strings from $\cX^n$. For a random variable $X$, we denote its distribution by $P_X$ and use the notation $X\sim P_{X}$ to indicate this. If $P_X \in \mathcal{P}^n\inp{\cX}$, we use $\cT^n_X$ to denote the set of all sequences with empirical distribution specified by $P_X$.  If $\vecx\in \cX^n$ has empirical distribution $P_X$, we say $\vecx$ is of type $P_X$ and write $\vecx\in \cT^n_X$.  When $P_X$ is not already defined, note that we write $\vecx\in \cT^n_{X}$ to implicitly define the type $P_X$ {associated with $\cT^n_X$ to be the } empirical distribution of $\vecx$. Similarly, for $(\vecx, \vecy)\in \cX^n\times\cY^n$, we write $(\vecx, \vecy)\in \cT^n_{XY}$ to define $P_{XY}$ to be the joint empirical distribution  the vectors $\vecx$ and $\vecy$. 
{For a broadcast channel $W_{YZ|X}$, we denote its marginal channels to the receivers by $W_{Y|X}$ and $W_{Z|X}$ respectively.} For a channel $W$, its $n$-fold product (memoryless use) is denoted by $W^n$. 
For any number $a$, we will use $\exp{a}$ to denote $2^a$ and $\log{a}$ to denote $\log_2{a}$. All information theoretic quantities (KL divergence, entropy and mutual information) are in base 2. 
For any set $\cS$, we use $\cS^c$ to denote its complement. For $s, s'\in \cS$ 
 we define 
 $$1_{\inb{s=s'}}=\begin{cases}1 \text{ if }s=s'\\0, \text{ otherwise.}\end{cases}$$

We state some properties of joint types from \cite[Chapter~2]{CsiszarK2011} which will be useful. Consider a joint type $P_{XY}\in \cP^n\inp{\cX\times\cY}$. For $\vecx\in \cT^n_X$, a  distribution $Q$ on $\cX$ and a discrete memoryless channel $W$ from $\cX$ to $\cY$, 
\begin{align}
|\cP^n\inp{\cX}|&\leq (n+1)^{|\cX|}\label{eq:type_property1}\\
(n+1)^{-|\cX|}\exp\inp{nH(X)}&\leq |\cT^n_{X}|\leq \exp\inp{nH(X)}\label{eq:type_property2}\\
(n+1)^{-|\cX||\cY|}\exp\inp{nH(Y|X)}&\leq |\cT^n_{Y|X}(\vecx)|\leq \exp\inp{nH(Y|X)}\label{eq:type_property3}\\
(n+1)^{-|\cX|}\exp\inb{-nD(P_X||Q)}&\leq \sum_{\vecx'\in \cT^n_{X}}Q^n(\vecx')\leq \exp\inb{-nD(P_X||Q)}\label{eq:type_property4}\\
\sum_{\vecy\in \cT^n_{Y|X}(\vecx)}W^n(\vecy|\vecx)&\leq \exp\inb{-nD(P_{XY}||P_XW)}\label{eq:type_property5}
\end{align}

Next, we state the following claim (proved in Appendix~\ref{app:polytope}, page~\pageref{sub:proof_of_claim_ref_claim_defn_gamma}) which will be used later. 
\begin{claim}\label{claim:defn_gamma}
Suppose $W_{YZ|X}$ is such that  $\capptop(\ch_{V|X})>0$, i.e. $|\cU|\geq 2$. Then, there exists $\gamma>0$ such that for any $u'\in \cU$ and $P_{U|U'}$,
\begin{align}
&\sum_{v\in \cV}\big|\sum_{u\in \cU}P_{U|U'}(u|u')W_{V|X}(v|u) - W_{V|X}(v|u')\big|{\geq} (1-P_{U|U'}(u'|u'))\gamma.\label{eq:defn_U'}
\end{align} This holds with equality if $P_{U|U'}(u'|u') = 1$. 
\end{claim}

For the achievability, we need to generalize the notion of relative distance between binary strings that was used in Section~\ref{sec:short_ach}. Recall the definition of $\cU$ (Definition~\ref{def:effective-input}). 
\begin{definition}\label{defn:dis}
Let $(\vecu, \vecx)\in \cU^n\times\cX^n$ with empirical distribution $P_{UX}$, i.e.,  $(\vecu, \vecx)\in \cT^n_{UX}$.
We define $d(\vecu, \vecx)$ and $d(P_{UX})$ as follows:\\
For the {\em representation channel} $\tP_{U|X}$ given by Definition~\ref{def:effective-input}, let $P_{UX\tilde{U}}(u, x, \tilde{u}) = P_{UX}(u, x)\tP_{U|X}(\tilde{u}|x)$ for $u, \tilde{u}\in \cU$ and $x\in \cX$. Then, 
\begin{align}
d(\vecu, \vecx) = d(P_{UX}) := \pr(U\neq \tilde{U})\label{eq:dis_blue}
\end{align}  under the $P_{UX\tilde{U}}$ above.
\end{definition}
The function $d$ defined above satisfies symmetry (for elements in $\cU^n$) and triangle inequality.
\begin{lemma}\label{lemma:dis_properties}
For $\vecu, \vecu'\in \cU^n$ and $\vecx\in \cX^n$, 
\begin{align}
d(\vecu, \vecu') &=  d(\vecu', \vecu).\label{dist:U} \\
\text{ and } d(\vecu, \vecu')&\leq d(\vecu, \vecx) + d(\vecu', \vecx) \label{eq:triangle}
\end{align}

\end{lemma}
This is proved in Appendix~\ref{app:polytope} (page~\pageref{proof:lemma10_prop}).
The next lemma uses a random coding argument to generate a codebook which will be used to show the achievability. It is proved in Appendix~\ref{app:polytope} (page~\pageref{app:proof:codebook}).
\begin{lemma}\label{lemma:codebook}
There exists $n_0(\cdot)$ such that for any  $\epsilon>0,\,  n\geq n_0(\epsilon)$, $K\geq \exp\inp{n\epsilon}$, $\alpha>0$, $\rad>0$ and type $P\in \cP^n\inp{\cU}$ satisfying {\em (i)} $\min_{u}P(u)\geq \alpha$, and {\em (ii)}  $\epsilon\leq R\leq \min_{P_{U U'}\in \cP}I(U;U')-\epsilon$ where $R:=\inp{\log{K}}/n$ and  $\cP:= \inb{P_{UU'}\in \cP^n\inp{\cU\times\cU}: P_U=P_{U'}=P, \, \Pr\inp{U\neq U'}< 2\rad}$,  
there exists an encoder $f:[1:K]\rightarrow \cU^n$, whose codewords $\vecu_{i}=f(i),\, i\in [1:K]$ are of type $P$ such that 
\begin{align}
d(\vecu_i, \vecu_j)\geq  2\rad \text{ for all }i\neq j, i, j \in [1:K], \label{eq:code1}
\end{align} and for every joint type $P_{UX}\in \cP^n\inp{\cU\times \cX}$ and $\vecx\in \cX^n$ satisfying $P_U= P$ and $\vecx \in \cT^n_{X}$, 
\begin{align}
|\{i\in [1:K]:(\vecu_i, \vecx)\in \cT^n_{UX}\}|\leq\exp{\left(n\left(\left|R-I(U;X)\right|^{+}+\epsilon\right)\right)},\label{eq:code2}
\end{align} where $\exp(a)$ denotes $2^a$.
\end{lemma}

\begin{proof}[Proof of achievability]
From Remark~\ref{rem:commonrand}, we note that consensus capacity is positive only if $\capptop(\ch_{V|X})>0$, i.e. $|\cU|\geq 2$. So, we consider a channel $W_{YZ|X}$ with  $|\cU|\geq 2$. 
Consider any rate $R_0>0$ such that 
\begin{align*}R_0<\max_{P_{U}}\min_{\stackrel{P_{X|U}:\,P_{X|U}(x|u)>0}{ \text{ only if }W_{V|X}(v|x)= W_{V|X}(v|u) \forall v}}\min\inp{I(U;Y), I(U;Z)}.
\end{align*}For sufficiently small $\zeta>0$,  there exist $\alpha>0$, $R$ and $n_1(\alpha)$ such that for every $n\geq n_1(\alpha)$, there exists $P_U\in \cP^n(\cU)$ satisfying $\min_{u\in \cU}P_{U}(u)\geq \alpha$ and 
\begin{align}R_0\leq R<\min_{\stackrel{P_{X|U}:\,P_{X|U}(x|u)>0}{ \text{ only if }W_{V|X}(v|x)= W_{V|X}(v|u) \forall v}}\min\inp{I(U;Y), I(U;Z)}-{\zeta},\label{eq:R_achievable}
\end{align} where the mutual information is evaluated under $P_{UXYZ}(u, x, y, z) = P_U(u)P_{X|U}(x|u)W_{YZ|X}(y,z|x)$.
{For such $\alpha>0, R$ and $n_1(\cdot)$, for sufficiently small $\epsilon, \delta>0$, for  all $n\geq n_1(\alpha)$ and $P_U\in \cP^n(\cU)$} satisfying \eqref{eq:R_achievable}, 
 $\min_{P_{UU'}\in \cP}I(U;U')-\epsilon\geq R\geq \epsilon$, where $\cP= \inb{P_{UU'}\in \cP^n\inp{\cU\times\cU}: P_U'=P_{U}, \, \Pr\inp{U\neq U'}< 2\rad}$.\footnote{ Note that for $\delta<1/4$, by Fano's inequality $I(U;U')\geq H(U)-H_2(2\delta)-2\delta\log{|\cU|}$. Thus, there is a sufficiently small choice of $\delta>0$ so that \eqref{eq:R_achievable} holds if $\epsilon < \zeta/2$.} We invoke Lemma~\ref{lemma:codebook} with $\alpha$ as chosen, sufficiently small $\epsilon, \delta>0$, $n\geq \max\inb{n_1(\alpha),n_0(\epsilon)}$ (where $n_0(\epsilon)$ is given by Lemma~\ref{lemma:codebook}) and $P=P_U\in \cP^n(\cU)$ satisfying \eqref{eq:R_achievable}, to get an encoder $f:[1:2^{nR}]\rightarrow \cU^n$ for which \eqref{eq:code1} and \eqref{eq:code2} hold. Let $\vecu(m) = f(m), \, m\in [1:2^{nR}]$ denote the codewords. Recall that $\phi_1$ and $\phi_2$ map the channel {outputs $y$ and $z$ to  their common channel output (i.e., connected component index of the characteristic graph). Refer to Definition~\ref{def:common-channel}.}
For $\vecy\in \cY^n$, we use $\phi_1(\vecy)$ to denote $(\phi_1(y_i))_{i\in[1:n]}$.

For $\vecx\in\cX^n$, $\vecu\in \cU^n$, $\vecv\in \cV^n$ and $\vecy\in\cY^n$, we write 
\begin{align}
&\vecx \vartriangleright \vecv,\qquad \qquad\text{ if }  D(P_{XV}||P_XW_{V|X})\leq 3\epsilon, \text{ for }P_{XV}\text{ s.t. }(\vecx, \vecv)\in \cT^n_{XV}, \text{ and }\label{eq:vartriangle}\\
&(\vecu, \vecx) \blacktriangleright \vecy, \qquad\text{ if for } P_{UXV}\text{ s.t. }(\vecu, \vecx, \vecy)\in \cT^n_{UXY},\,D(P_{UXY}||P_{UX}W_{Y|X})\leq 3\epsilon.\label{eq:blacktriangle_sub}
\end{align}

For an encoder $f$ of rate $R$, the decoder outputs $\decB(\vecy)=m$ if it is the unique $m\in[1:2^{nR}]$ such that
\begin{enumerate}
\item $f(m)\vartriangleright \phi_1(\vecy)$,
\item there is an $\bar{\vecx}\in \cX^n$ such that $d(f(m),\bar{\vecx})< \rad$ and $(f(m), \bar{\vecx})\blacktriangleright \vecy$.
\end{enumerate}
$\decB(\vecy)=\bot$ if no such unique $m$ exists.  $\decC$ is similarly defined (with $\phi_2, W_{Z|X}$ in lieu of $\phi_1, W_{Y|X}$, respectively.).
The first decoding condition requires the codeword $f(m)$ to be consistent with the common channel's output $\phi_1(\vecy)$ in the sense that $f(m)\vartriangleright \phi_1(\vecy)$; we denote this by  $m \lozenge \vecy$. The second condition, denoted by $m \blacklozenge \vecy$, requires an ``explaining'' vector $\bar{\vecx}$ which is $\rad$-close to the codeword $f(m)$ and is consistent with the output $\vecy$ in the sense that $(f(m), \bar{\vecx})\blacktriangleright \vecy$.
 For $\vecx\in\cX^n$, define the event
\[\cB_{\vecx}=(\exists m\in[1:2^{nR}]: d(f(m),\vecx)\geq \rad, m \lozenge \vecY, m \blacklozenge \vecY).\]
\begin{claim}\label{claim:achievability}
{For sufficiently small $\delta>0$, there exists a sufficiently small $\epsilon>0$ such that for large enough $n$}, for every $\vecx\in \cX^n$, 
\begin{align}
&\Pr(\vecx \vartriangleright \phi_1(\vecY) | \vecx) \geq 1-2^{-n \epsilon}, \label{eq:claim1}\\
&\Pr(\cB_{\vecx} | \vecx) \leq 2^{-n \epsilon}, \text{ and } \label{eq:claim3}\\
\intertext{if there exists $m\in [1:2^{nR}]$  such that $d(f(m),\vecx)< \rad$,}
&\Pr(m \blacklozenge \vecY | \vecx) \geq 1-2^{-n \epsilon}. \label{eq:claim2}
\end{align}
\end{claim}
Before proving the claim, let us see that it implies the theorem. To this end, consider three collectively exhaustive possibilities for the transmitted vector $\vecx$:
\begin{enumerate}[label=(\roman*)]
\item $\vecx=f(m)$ for some $m$,
\item $d(f(m),\vecx)\geq\rad$ for all $m\in[1:2^{nR}]$, and 
\item  there is an $m$ such that $d(f(m),\vecx)<\rad$, but $\vecx\neq f(m)$.
\end{enumerate} We show that the probability of error is vanishingly small in all these cases.\\	
{\em Case (i):} $\vecx=f(m)$ for some $m$. For all $m'\neq m$, the encoder  has $d(f(m),f(m'))\geq2\rad$ by \eqref{eq:code1} of Lemma~\ref{lemma:codebook}. Hence,
\begin{align*}
&\Pr(\decB(\vecY)\neq m|f(m)) = \Pr(\{m\lozenge \vecY\}^c\cup\{m \blacklozenge \vecY\}^c\cup\{\exists m'\neq m \text{ s.t } m'\lozenge \vecY,m'\blacklozenge \vecY\}|f(m))\\
& \leq \Pr(\{f(m)\vartriangleright \phi_1(\vecY)\}^c|f(m)) + \Pr(\{m \blacklozenge \vecY\}^c|f(m))+\Pr(\cB_{f(m)}|f(m))\\
& \leq 2^{-n\epsilon}+ 2^{-n\epsilon}+ 2^{-n\epsilon}\leq 2^{-n\epsilon+\log{3}},
\end{align*}
where the upper bounds on the probabilities follow from \eqref{eq:claim1}-\eqref{eq:claim2}. Similarly, $\Pr(\decC(\vecZ)\neq m|f(m)) \leq 2^{-n\epsilon+\log{3}}$. By a union bound, $$\Pr(\decB(\vecY)=\decC(\vecZ) = m|f(m))\geq1-2^{-n\epsilon+\log{6}}.$$\\
{\em Case (ii):} $d(f(m),\vecx)\geq\rad$ for all $m\in[1:2^{nR}]$. In this case, $\Pr(\decB(\vecY)\neq\bot|\vecx)\leq\Pr(\cB_{\vecx}|\vecx)\leq 2^{-n\epsilon}$ where the last step is from~\eqref{eq:claim3}. Similarly, $\Pr(\decC(\vecZ)\neq\bot|\vecx)\leq 2^{-n\epsilon}$. Hence, by a union bound $\Pr(\decB(\vecY)=\decC(\vecZ)=\bot|\vecx)\geq 1 - 2^{-n\epsilon+1}$.\\
{\em Case (iii):} There is an $m$ such that $d(f(m),\vecx)<\rad$, but $\vecx\neq f(m)$. Then, for all $m'\neq m$, by \eqref{eq:triangle}, 
\begin{align}
d(f(m'),\vecx) \geq d(f(m'),f(m)) - d(f(m),\vecx) > 2\rad - \rad = \rad.\label{eq:new_display}
\end{align} Hence, 
\begin{align}
\Pr(\decB(\vecY)\notin\{m,\bot\}|\vecx) &=\Pr(\exists m\neq m'\text{ s.t. }m' \lozenge \vecY, m\blacklozenge \vecZ|\vecx)\\
&\stackrel{(a)}{=}\Pr(\exists m\neq m'\text{ s.t. }d(f(m'),\vecx)\geq \delta, m' \lozenge \vecY, m\blacklozenge \vecZ|\vecx)\\
&\stackrel{(b)}{=}\Pr(\cB_{\vecx}|\vecx)\leq 2^{-n\epsilon},
\end{align} where $(a)$ follows from \eqref{eq:new_display} and  $(b)$ from~\eqref{eq:claim3}. By a union bound, $\Pr(\decB(\vecY),\decC(\vecZ)\in\{m,\bot\}|\vecx) \geq 1 - 2^{-n\epsilon+1}$. 
Further, $\Pr(m\blacklozenge \vecY|\vecx)\geq 1-2^{-n\epsilon}$ by \eqref{eq:claim2}. Similarly, $\Pr(m\blacklozenge \vecZ|\vecx)\geq 1-2^{-n\epsilon}$. Hence, $\Pr(m\blacklozenge \vecY, m\blacklozenge \vecZ|\vecx)\geq 1-2^{-n\epsilon+1}$ by a union bound. 
Since $\phi_1(\vecY)=\phi_2(\vecZ)$, either the first condition for message $m$ is met or not met together for both decoders. Hence, 
\begin{align*}
&\Pr(\decB(\vecY)=\decC(\vecZ)|\vecx)\\
&\geq \Pr(m\blacklozenge \vecY, m\blacklozenge \vecZ, \inp{\decB(\vecY),\decC(\vecZ)\in\{m,\bot\}}|\vecx)\\
&\geq 1 - 2^{-n\epsilon+2},
\end{align*} where the last step involves taking a union bound. 

It only remains to prove Claim~\ref{claim:achievability}.
\begin{proof}[Proof of Claim~\ref{claim:achievability}]
We will first show \eqref{eq:claim1}.
\begin{align}
&\Pr(\inb{\vecx \vartriangleright \phi_1(\vecY)}^c | \vecx)\nonumber\\
&=\Pr((\vecx, \phi_1(\vecY))\in \cT^n_{XV}, D(P_{XV}||P_XW_{V|X})> 3\epsilon|\vecx)\nonumber\\
&\leq \sum_{\substack{P_{XV}\in \cP^n(\cX\times\cV):\\D(P_{XV}||P_XW_{V|X})> 3\epsilon}}\Pr((\vecx, \phi_1(\vecY))\in \cT^n_{XV}|\vecx)\nonumber\\
&=\sum_{\substack{P_{XV}\in \cP^n(\cX\times\cV):\\D(P_{XV}||P_XW_{V|X})> 3\epsilon}}\sum_{\vecv\in \cT^n_{V|X}(\vecx)}W^n_{V|X}(\vecv|\vecx)\nonumber\\
&\stackrel{(a)}{\leq}\sum_{\substack{P_{XV}\in \cP^n(\cX\times\cV):\\D(P_{XV}||P_XW_{V|X})> 3\epsilon}}\exp\inp{-nD(P_{XV}||P_XW_{V|X})}\nonumber\\
&\leq |\cP^n(\cX\times\cV)|\exp\inp{-3n\epsilon}\nonumber\\
&\stackrel{(b)}{\leq} \exp\inp{-2n\epsilon}\label{eq:A}\\
&\leq \exp\inp{-n\epsilon},\nonumber
\end{align} where $(a)$ follows from \eqref{eq:type_property5} and $(b)$ holds {for sufficiently large $n$}. This shows \eqref{eq:claim1}. 

Next, we show \eqref{eq:claim2}.
Let $m\in [1:2^{nR}]$ and $\vecx\in \cX^n$ such that $d(f(m),\vecx)< \rad$. Note that, since $d(f(m), \vecx)<\delta$, $\Pr(m \blacklozenge \vecY | \vecx)\geq \Pr(\inp{f(m),\vecx}\blacktriangleright \vecY| \vecx)$ (as $\vecx$ may serve as the explaining vector $\bar{\vecx}$ under the event $(f(m), \bar{\vecx})\blacktriangleright \vecy$.).
We will show that 
\begin{align*}
&\Pr(\inb{\inp{f(m),\vecx}\blacktriangleright \vecY}^c| \vecx) \leq 2^{-n \epsilon}
\end{align*} which will show \eqref{eq:claim2}. 
\begin{align*}
&\Pr( \inb{\inp{f(m),\vecx}\blacktriangleright \vecY}^c| \vecx)\\
&{\leq}  \sum_{\substack{P_{UXY}\in \cP^n\inp{\cU\times\cX\times\cY}:(f(m), \vecx)\in \cT^n_{UX},\\D(P_{UXY}||P_{UX}W_{Y|X})>3\epsilon}}\Pr\inp{( f(m),\vecx, \vecY)\in \cT^n_{UXY}}\\
&= \sum_{\substack{P_{UXY}\in \cP^n\inp{\cU\times\cX\times\cY}:(f(m), \vecx)\in \cT^n_{UX},\\D(P_{UXY}||P_{UX}W_{Y|X})>3\epsilon}}\sum_{\textcolor{white}{w}\vecy\in \cT^n_{Y|UX}( f(m),\vecx)}W_{Y|X}(\vecy|\vecx)\\
&\stackrel{(a)}{\leq} \sum_{\substack{P_{UXY}\in \cP^n\inp{\cU\times\cX\times\cY}:(f(m), \vecx)\in \cT^n_{UX},\\D(P_{UXY}||P_{UX}W_{Y|X})>3\epsilon}}\exp\inp{-n\inp{D(P_{UXY}||P_{UX}W_{Y|X})}}\\
&{\leq} \sum_{\substack{P_{UXY}\in \cP^n\inp{\cU\times\cX\times\cY}:(f(m), \vecx)\in \cT^n_{UX},\\D(P_{UXY}||P_{UX}W_{Y|X})>3\epsilon}}\exp\inp{-3n\epsilon}\\
&{\leq} |\cP^n\inp{\cU\times\cX\times\cY}|\exp\inp{-3n\epsilon}\\
&\stackrel{(b)}{\leq} \exp\inp{-2n\epsilon}\leq \exp\inp{-n\epsilon},
\end{align*}where $(a)$ follows from \eqref{eq:type_property5} and $(b)$ holds for sufficiently large $n$.\\
Finally, we need to show \eqref{eq:claim3}. Recall that \[\cB_{\vecx}=(\exists m\in[1:2^{nR}]: d(f(m),\vecx)\geq \rad, m \lozenge \vecY, m \blacklozenge \vecY).\]
\begin{align}
\Pr(\cB_{\vecx} | \vecx)&\leq \Pr(\inb{\vecx \vartriangleright \phi_1(\vecY)}^c | \vecx) + \Pr(\inb{\vecx \vartriangleright \phi_1(\vecY)}\cap \cB_{\vecx} | \vecx)\nonumber\\
&\stackrel{(a)}{\leq} 2^{-2n \epsilon} + \Pr(\inb{\vecx \vartriangleright \phi_1(\vecY)}\cap \cB_{\vecx} | \vecx),\label{eq:claim3_prob}
\end{align} where $(a)$ follows from \eqref{eq:A}. 
Let $\cD$ be the set of joint types $P_{UXY}\in \cP^n\inp{\cU\times\cX\times\cY}$ satisfying 
\begin{align} 
&d(P_{UX})\geq \rad,\label{eq:D1}\\
&D(P_{X\phi_1(Y)}||P_XW_{V|X})\leq 3\epsilon,\label{eq:D2}\\
&D(P_{U\phi_1(Y)}||P_UW_{V|X})\leq 3\epsilon, \label{eq:D3}
\end{align}
and there exists $P_{\bar{X}|UXY}$  such that $P_{UXY}P_{\bar{X}|UXY}\in \cP^n\inp{\cU\times\cX\times\cY\times\cX}$ and
\begin{align}
&d(P_{U\bar{X}})\leq \rad, \label{eq:D5}\\
&D(P_{U\bar{X}Y}||P_{U\bar{X}}W_{Y|\bar{X}})\leq 3\epsilon,\label{eq:D6}
\end{align} where $W_{Y|\bar{X}}(y|x):=W_{Y|X}(y|x)$ for all $x, y$. 
Under the events $\inb{\vecx \vartriangleright \phi_1(\vecY)}$ and $\cB_{\vecx}$, there is an $m\in [1:2^{nR}]$ and $\bar{\vecx}\in \cX^n$ such that the joint type $P_{UXY\bar{X}}$ defined by $(f(m), \vecx, \vecy, \bar{\vecx})\in \cT^n_{UXY\bar{X}}$ satisfies \eqref{eq:D1}-\eqref{eq:D6}. In particular, $\inb{\vecx \vartriangleright \phi_1(\vecY)}$ implies $\eqref{eq:D2}$, $d(f(m),\vecx)\geq \rad$, $m \lozenge \vecY$ and $m \blacklozenge \vecY$ in $\cB_{\vecx}$ implies \eqref{eq:D1}, \eqref{eq:D3} and \eqref{eq:D5}-\eqref{eq:D6} respectively. Hence, 
\begin{align}\label{eq:A_x_equivalence}
\inb{\vecx \vartriangleright \phi_1(\vecY)}\cap \cB_{\vecx}\subseteq \Big\{\exists P_{UXY}\in \cD,  m\in [1:2^{nR}] \text{ such that }(f(m),\vecx,  \vecY)\in \cT^n_{UXY}\Big\}.
\end{align} Thus, 
\begin{align}
&\Pr(\inb{\vecx \vartriangleright \phi_1(\vecY)}\cap \cB_{\vecx} | \vecx)\nonumber\\
&\stackrel{(a)}{\leq} \sum_{\stackrel{P_{UXY}\in\cD:}{\vecx\in \cT^n_{X}}}\sum_{{\textcolor{white}{w}m\in [1:2^{nR}]}}\Pr\inp{(\vecx, f(m), \vecY)\in \cT^n_{UXY}| \vecx}\nonumber\\
&= \sum_{\stackrel{P_{UXY}\in\cD:}{\vecx\in \cT^n_{X}}}\sum_{\stackrel{m\in [1:2^{nR}],}{f(m)\in \cT^n_{U|X}(\vecx)}}\sum_{\vecy\in \cT^n_{Y|XU}(\vecx, f(m))}W_{Y|X}(\vecy|\vecx)\nonumber\\
&\stackrel{(b)}{\leq}\sum_{\stackrel{P_{UXY}\in\cD:}{\vecx\in \cT^n_{X}}}\sum_{\stackrel{m\in [1:2^{nR}],}{f(m)\in \cT^n_{U|X}(\vecx)}}\exp\inp{-n\inp{I(Y;U|X)-\epsilon}}\nonumber\\
&\stackrel{(c)}{\leq}\sum_{{P_{UXY}\in\cD}}\exp{\left(n\left(\left|R-I(U;X)\right|^{+}+\epsilon\right)\right)}\exp\inp{-n\inp{I(Y;U|X)-\epsilon}}\nonumber\\
&=\sum_{P_{UXY}\in\cD}\exp{\left(n\left(\left|R-I(U;X)\right|^{+}-I(Y;U|X)+2\epsilon\right)\right)},\label{eq:E1}
\end{align}
where $(a)$ is obtained by taking union bound over $P_{UXY}\in \cD, m\in [1:2^{nR}]$ (see \eqref{eq:A_x_equivalence}), $(b)$ follows by the following argument: for each $\vecy \in \cT^n_{Y|XU}(\vecx, \vecu),\, W_{Y|X}(\vecy|\vecx)$ is the same and is hence upper bounded by $1/|\cT^n_{Y|X}(\vecx)|$. By \eqref{eq:type_property3}, $1/|\cT^n_{Y|X}(\vecx)|\leq (n+1)^{|\cX||\cY|}\exp\inp{-nH(Y|X)}\leq \exp\inp{n\inp{-H(Y|X)+\epsilon}}$ {for sufficiently large $n$.} Also, $|\cT^n_{Y|XU}(\vecx, \vecu)|\leq \exp\inp{nH(Y|XU)} $ by \eqref{eq:type_property3}.  The inequality $(c)$ follows from \eqref{eq:code2}.

For each  $P_{UXY}\in\cD$, let   
\begin{align}
\zeta_{UXY}:=\exp{\left(n\left(\left|R-I(U;X)\right|^{+}-I(Y;U|X)+2\epsilon\right)\right)}.\label{eq:E}
\end{align} For each  $P_{UXY}\in\cD$, if we can show that $
\zeta_{UXY}\leq2^{-2n\epsilon}$, then, from \eqref{eq:claim3_prob} and \eqref{eq:E1}, 
\begin{align*}
\Pr(\cB_{\vecx} | \vecx)&\leq 2^{-2n \epsilon} + \Pr(\inb{\vecx \vartriangleright \phi_1(\vecY)}\cap \cB_{\vecx} | \vecx)\\
&\leq 2^{-2n \epsilon}  + |\cD|2^{-2n\epsilon}\\
&\leq 2^{-n\epsilon}
\end{align*}for sufficiently large $n$. Thus, we would have shown \eqref{eq:claim3}.
To this end, fix a distribution $P_{UXY}\in\cD$. We consider two possibilities.\\
{\em Case (i):} $R\leq I(U;X)$.\\
In this case,
\begin{align*}
\zeta_{UXY}=\exp{\left(n\left(-I(Y;U|X)+2\epsilon\right)\right)}.
\end{align*}
Note that $I(Y;U|X) \geq I(\phi_1(Y);U|X)$. Let $V:=\phi_1(Y)$. We will show that $I(V;U|X)>4\epsilon$. This would imply that $\zeta_{UXY}\leq \exp\inp{-2n\epsilon}$. For the sake of contradiction, assume $I(V;U|X)\leq 4\epsilon$. Using this and \eqref{eq:D2}, 
\begin{align*}
7\epsilon &\geq I(V;U|X) + D(P_{XV}||P_XW_{V|X})\\
& = \sum_{u,x, v}P_{UXV}( u,x, v)\inb{\log\inp{\frac{P_{UXV}( u,x, v)}{P_{XV}(x,v)P_{U|X}(u|x)}}+\log\inp{\frac{P_{XV}(x, v)}{P_X(x)W_{V|X}(v|x)}}}\\
& = D(P_{UXV}||P_{UX}W_{V|X}).
\end{align*}
By Pinsker's inequality~\cite[page 44]{CsiszarK2011}, this implies that
\begin{align}
\sqrt{14\epsilon\ln{2}}&\geq \sum_{u, x, v}\left|P_{UXV}(u, x, v)-P_{UX}(u,x)W_{V|X}(v|x)\right|\nonumber\\
&\stackrel{(a)}{=}\sum_{u, x, v}\left|P_{UXV}(u, x, v)-P_{UX}(u,x)\sum_{\tilde{u}\in \cU}{P}_{\tilde{U}|X}(\tilde{u}|x)W_{V|X}(v|\tilde{u})\right|\nonumber\\
&\geq\sum_{u, v}\left|P_{UV}(u, v)-\sum_{x, \tilde{u}}P_{UX}(u,x){P}_{\tilde{U}|X}(\tilde{u}|x)W_{V|X}(v|\tilde{u})\right|,\label{eq:KL1}
\end{align}where $(a)$ follows from \eqref{eq:Ptilde} with ${P}_{\tilde{U}|X}(\tilde{u}|x) := \tP_{{U}|X}(\tilde{u}|x)$ for all $\tilde{u}, x$. Define 
\begin{align}\label{eq:B}
P_{U\tilde{U}}(u, \tilde{u}) := \sum_{x}P_{UX}(u,x){P}_{\tilde{U}|X}(\tilde{u}|x)
\end{align} for all $u, \tilde{u}\in \cU$. Then, \eqref{eq:KL1} gives 
\begin{align}
\sqrt{14\epsilon\ln{2}}\geq\sum_{u, v}\left|P_{UV}(u, v)-\sum_{\tilde{u}}P_{U\tilde{U}}(u,\tilde{u})W_{V|X}(v|\tilde{u})\right|\label{eq:KL1'}.
\end{align}
Next, applying Pinsker's inequality to \eqref{eq:D3} gives
\begin{align}
\sum_{u, v}|P_{UV}(u, v)-P_U(u)W_{V|X}(v|{u})|\leq \sqrt{6\epsilon\ln{2}}.\label{eq:KL2}
\end{align} 
From \eqref{eq:KL1'} and \eqref{eq:KL2},
\begin{align*}
\sqrt{14\epsilon\ln{2}}+ \sqrt{6\epsilon\ln{2}}&\geq\sum_{u, v}\left|\sum_{\tilde{u}}P_{U\tilde{U}}(u,\tilde{u})W_{V|X}(v|\tilde{u})- P_U(u)W_{V|X}(v|{u})\right|\\
& = \sum_{u, v}P_{U}(u)\left|\sum_{\tilde{u}}P_{\tilde{U}|U}(\tilde{u}|u)W_{V|X}(v|\tilde{u})- W_{V|X}(v|{u})\right|\\
& \stackrel{(a)}{\geq} \alpha\sum_{u, v}\left|\sum_{\tilde{u}}P_{\tilde{U}|U}(\tilde{u}|u)W_{V|X}(v|\tilde{u})- W_{V|X}(v|{u})\right|,
\end{align*} where $(a)$ follows by recalling that $P_U(u)\geq \alpha, \, u\in \cU$. This implies that for every $u'\in \cU$,
\begin{align}
\sum_{v\in\cV}\left|\sum_{\tilde{u}\in \cU}P_{\tilde{U}|U}(\tilde{u}|u')W_{V|X}(v|\tilde{u})- W_{V|X}(v|{u'})\right|\leq \frac{1}{\alpha}\inp{\sqrt{2\epsilon\ln{2}}(\sqrt{7}+\sqrt{3})}.\label{eq:one}
\end{align}

Also, from \eqref{eq:B} and Definition~\ref{defn:dis} (equation \eqref{eq:dis_blue}), $d(P_{UX}) = d(P_{U\tilde{U}})$. Thus, using \eqref{eq:D1},
\begin{align*}
\pr(U\neq \tilde{U}) = \sum_{u, \tilde{u}:u\neq \tilde{u}}P_{U\tilde{U}}(u, \tilde{u})\geq \rad.
\end{align*} This implies that
\begin{align*}
1-\rad\geq \sum_{u}P_{U\tilde{U}}(u, u)&= \sum_uP_{U}(u)P_{\tilde{U}|U}(u|u)\\
&\geq \min_{u'\in \cU}P_{\tilde{U}|U}(u'|u').
\end{align*}
Thus, there exists $u'\in \cU$ for which 
\begin{align}
P_{\tilde{U}|U}(u'|u')\leq 1-\rad.\label{eq:UneqU}
\end{align}
From Claim~\ref{claim:defn_gamma} and \eqref{eq:UneqU}, there exists $u'\in \cU$ for which 
\begin{align}
\sum_{v\in \cV}\big|\sum_{u\in \cU}P_{U|U'}(u|u')W_{V|X}(v|u) - W_{V|X}(v|u')\big|\geq \rad\gamma.
\end{align} 
This contradicts \eqref{eq:one} for $\epsilon>0$ small enough\footnote{{Such a choice of $\epsilon$ is possible because so far we only required that $\epsilon>0, \delta>0$ be both sufficiently small independently.}} such that (recall that $\gamma>0$ in Claim~\ref{claim:defn_gamma})
\begin{align}
\rad>\frac{1}{\gamma\alpha}\inp{\sqrt{2\epsilon\ln{2}}(\sqrt{7}+\sqrt{3})}.\label{eq:bound 1}
\end{align}   Thus, $I(Y;U|X)\geq 4\epsilon$. This implies that for every  $P_{UXY}\in\cD$ with $I(U;X)\geq R$, 
\begin{align}
\zeta_{UXY}\leq \exp{\left(-2n\epsilon\right)}.\label{eq:E_bound1}
\end{align}
{\em Case (ii):}
We consider distributions $P_{UXY}\in\cD$ with $I(U;X)< R$. In this case, \eqref{eq:E} is
\begin{align}
\zeta_{UXY}&=\exp{\left(n\left(R-I(U;X)-I(Y;U|X)+2\epsilon\right)\right)}\nonumber\\
&=\exp{\left(n\left(R-I(U;XY)+2\epsilon\right)\right)}\\
&=\exp{\left(n\left(R-I(U;Y)+2\epsilon-I(U;X|Y)\right)\right)}\label{eq:EE}
\end{align}
We will show that $I(U;Y)-R>4\epsilon$ for sufficiently small $\rad>0, \epsilon>0$. This would imply that $\zeta_{UXY}\leq \exp{\inp{-2n\epsilon}}$. 
Since $P_{UXY}\in\cD$, there exists a conditional distribution $P_{\bar{X}|XUY}$ satisfying \eqref{eq:D5} and \eqref{eq:D6} (see definition of $\cD$, \eqref{eq:D1}-\eqref{eq:D6}). For this $P_{\bar{X}|XUY}$, consider $P_{XUY\bar{X}} = P_{XUY}P_{\bar{X}|XUY}$. Using \eqref{eq:D6} and Pinsker's inequality, 
\begin{align*}
\sqrt{6\epsilon\ln{2}}&\geq \sum_{u, y}\Big|P_{UY}(u, y)- \sum_{\bar{x}}P_{U}(u)P_{\bar{X}|U}(\bar{x}|u)W_{Y|X}(y|\bar{x})\Big|\\
&\stackrel{(a)}{=}\sum_{u, y}\Big|P_{UY}(u, y)- \sum_{\bar{x},v}P_{U}(u)P_{\bar{X}|U}(\bar{x}|u)W_{V|X}(v|\bar{x})W_{Y|XV}(y|\bar{x},v)\Big|\\
& \stackrel{(b)}{=} \sum_{u, y}\Big|P_{UY}(u, y)- \sum_{\bar{x},v}P_{U}(u)P_{\bar{X}|U}(\bar{x}|u)\inp{\sum_{\tilde{u}\in \cU}\tP_{U|X}(\tilde{u}|\bar{x})W_{V|X}(v|\tilde{u})}W_{Y|XV}(y|\bar{x},v)\Big|\\
& = \sum_{u, y}\Big|P_{UY}(u, y)- \sum_{\bar{x},v}\sum_{\tilde{u}}P_{U}(u)P_{\bar{X}|U}(\bar{x}|u)\tP_{U|X}(\tilde{u}|\bar{x})W_{V|X}(v|\tilde{u})W_{Y|XV}(y|\bar{x},v)\Big|\\
& = \sum_{u, y}\Big|P_{UY}(u, y)- \sum_{\bar{x},v}P_{U}(u)P_{\bar{X}|U}(\bar{x}|u)\tP_{U|X}({u}|\bar{x})W_{V|X}(v|{u})W_{Y|XV}(y|\bar{x},v)\\
&\qquad-\sum_{\bar{x},v}\sum_{\tilde{u}\neq u}P_{U}(u)P_{\bar{X}|U}(\bar{x}|u)\tP_{U|X}(\tilde{u}|\bar{x})W_{V|X}(v|\tilde{u})W_{Y|XV}(y|\bar{x},v)\Big|\\
& \geq \sum_{u, y}\Big|P_{UY}(u, y)- \sum_{\bar{x},v}P_{U}(u)P_{\bar{X}|U}(\bar{x}|u)\tP_{U|X}({u}|\bar{x})W_{V|X}(v|{u})W_{Y|XV}(y|\bar{x},v)\Big|\\
&\qquad-\sum_{u, y}\sum_{\bar{x},v}\sum_{\tilde{u}\neq u}P_{U}(u)P_{\bar{X}|U}(\bar{x}|u)\tP_{U|X}(\tilde{u}|\bar{x})W_{V|X}(v|\tilde{u})W_{Y|XV}(y|\bar{x},v)\\
& \stackrel{(c)}{=} \sum_{u, y}\Big|P_{UY}(u, y)- \sum_{\bar{x},v}P_{U}(u)P_{\bar{X}|U}(\bar{x}|u)\tP_{U|X}({u}|\bar{x})W_{V|X}(v|u)W_{Y|XV}(y|\bar{x},v)\Big|-d(P_{U\bar{X}})\\
& \stackrel{(d)}{\geq} \sum_{u, y}\Big|P_{UY}(u, y)- \sum_{\bar{x},v}P_{U}(u)P_{\bar{X}|U}(\bar{x}|u)\tP_{U|X}({u}|\bar{x})W_{V|X}(v|u)W_{Y|XV}(y|\bar{x},v)\Big|-\rad
\end{align*} where $(a)$ follows from \eqref{eq:commonch2}, $(b)$ follows from \eqref{eq:Ptilde}, $(c)$ follows from Definition~\ref{defn:dis},\newline
 i.e., $d(P_{U\bar{X}}) = \sum_{u, \bar{x}, \tilde{u}\neq u}P_{U}(u)P_{\bar{X}|U}(\bar{x}|u)\tP_{U|X}(\tilde{u}|\bar{x})$ and $(d)$ follows from \eqref{eq:D5}.
Thus, 
\begin{align}
\sum_{u, y}\Big|P_{UY}(u, y)- \sum_{\bar{x},v}P_{U}(u)P_{\bar{X}|U}(\bar{x}|u)\tP_{U|X}({u}|\bar{x})W_{V|X}(v|u)W_{Y|XV}(y|\bar{x},v)\Big|\leq \sqrt{6\epsilon\ln{2}}+\rad.\label{eq:something}
\end{align}
For $u\in \cU$, recall that (Definition~\ref{def:effective-input}) $\cX_u$ is defined as 
\begin{align}
\cX_u =\{x\in \cX:W_{V|X}(v|x)=W_{V|X}(v|u),\, \forall\,v\in \cV\}.\label{eq:defn_x_u}
\end{align}
\noindent From \eqref{eq:cX_u} of Lemma~\ref{lem:polytope}, we have
for $u\in \cU$ 
\begin{align*}
\tP_{U|X}(u|x) = 1 \text{ if and only if } x\in \cX_u .
\end{align*}
Thus, 
\begin{align}
\eta:=\min_{u\in\cU}\min_{x\notin\cX_u}\inp{1-\tP_{U|X}(u|x)}>0.\label{eq:defn_eta}
\end{align}
This implies that 
\begin{align}
d(P_{U\bar{X}})&= \sum_{u}\sum_{ \bar{x}\in \cX}\sum_{u'\neq u}P_{U\bar{X}}(u, \bar{x})\tP_{U|X}(u'|\bar{x})\nonumber\\
&\stackrel{(a)}{=} \sum_{u}\sum_{ \bar{x}\notin \cX_u}\sum_{u'\neq u}P_{U\bar{X}}(u, \bar{x})\tP_{U|X}(u'|\bar{x})\nonumber\\
&= \sum_{u}\sum_{ \bar{x}\notin \cX_u}P_{U\bar{X}}(u,\bar{x})\sum_{u'\neq u}\tP_{U|X}(u'|\bar{x})\nonumber\\
&= \sum_{u}\sum_{ \bar{x}\notin \cX_u}P_{U\bar{X}}(u,\bar{x})\inp{1-\tP_{U|X}(u|\bar{x})}\nonumber\\
&\stackrel{(b)}{\geq} \sum_{u}\sum_{ \bar{x}\notin \cX_u}P_{U\bar{X}}(u,\bar{x})\eta, \label{eq:achievability_dist}
\end{align} where $(a)$ follows by noting that $\tP_{U|X}(u'|\bar{x}) = 0$ for $\bar{x}\in \cX_u$ and $u'\neq u$ (see \eqref{eq:cX_u}) and $(b)$ follows from \eqref{eq:defn_eta}. 
Let 
\begin{align}
t:= \sum_{u}\sum_{ \bar{x}\notin \cX_u}P_{U\bar{X}}(u,\bar{x}).\label{eq:C}
\end{align} 
Now, by \eqref{eq:achievability_dist},
\begin{align*}
t&\leq d(P_{U\bar{X}})/\eta\\
&\leq \rad/\eta  \qquad\text{ by }\eqref{eq:D5}.
\end{align*} 
Hence, 
\begin{align}
\sum_{u}\sum_{ x\in \cX_u}P_{U\bar{X}}(u, x) =1- \sum_{u}\sum_{ x\notin \cX_u}P_{U\bar{X}}(u, x)=1-t\geq 1-\rad/\eta. \label{eq:defn_t}
\end{align}
Hence, $1-t>0$ for all sufficiently small $\rad>0$. For such $\delta$, we can define a joint distribution $P_{\hat{U}\hat{X}}$ as
\begin{align}
P_{\hat{U}\hat{X}}(u, x) = P_{\hat{U}}(u)P_{\hat{X}|\hat{U}}(x|u):=\begin{cases}\frac{1}{1-t}P_{U\bar{X}}(u, x),&\text{if }x\in \cX_u,\\0,&\text{otherwise.}\end{cases}\label{eq:DF0}
\end{align} Thus, proceeding from \eqref{eq:something},
\begin{align*}
\sqrt{6\epsilon\ln{2}}+\rad&\geq \sum_{u, y}\Big|P_{UY}(u, y)- \sum_{\bar{x},v}P_{U\bar{X}}(u,\bar{x})\tP_{U|X}({u}|\bar{x})W_{V|X}(v|u)W_{Y|XV}(y|\bar{x},v)\Big|\\
& = \sum_{u, y}\Big|P_{UY}(u, y)- \sum_{\bar{x}\in \cX_u,v}P_{U\bar{X}}(u,\bar{x})\tP_{U|X}({u}|\bar{x})W_{V|X}(v|u)W_{Y|XV}(y|\bar{x},v)\\
&\qquad-\sum_{\bar{x}\notin \cX_u,v}P_{U\bar{X}}(u,\bar{x})\tP_{U|X}({u}|\bar{x})W_{V|X}(v|u)W_{Y|XV}(y|\bar{x},v)\Big|\\
& \geq \sum_{u, y}\Big|P_{UY}(u, y)- \sum_{\bar{x}\in \cX_u,v}P_{U\bar{X}}(u,\bar{x})\tP_{U|X}({u}|\bar{x})W_{V|X}(v|u)W_{Y|XV}(y|\bar{x},v)\Big|\\
&\qquad-\sum_{u, y}\sum_{\bar{x}\notin \cX_u,v}P_{U\bar{X}}(u,\bar{x})\tP_{U|X}({u}|\bar{x})W_{V|X}(v|u)W_{Y|XV}(y|\bar{x},v)\\
& = \sum_{u, y}\Big|P_{UY}(u, y)- \sum_{\bar{x}\in \cX_u,v}P_{U\bar{X}}(u,\bar{x})\tP_{U|X}({u}|\bar{x})W_{V|X}(v|u)W_{Y|XV}(y|\bar{x},v)\Big|\\
&\qquad-\sum_{u}\sum_{\bar{x}\notin \cX_u}P_{U\bar{X}}(u,\bar{x})\tP_{U|X}({u}|\bar{x})\\
& \geq \sum_{u, y}\Big|P_{UY}(u, y)- \sum_{\bar{x}\in \cX_u,v}P_{U\bar{X}}(u,\bar{x})\tP_{U|X}({u}|\bar{x})W_{V|X}(v|u)W_{Y|XV}(y|\bar{x},v)\Big|\\
&\qquad-\sum_{u}\sum_{\bar{x}\notin \cX_u,}P_{U\bar{X}}(u,\bar{x})\\
& \stackrel{(a)}{=} \sum_{u, y}\Big|P_{UY}(u, y)- \sum_{\bar{x}\in \cX_u,v}P_{U\bar{X}}(u,\bar{x})\tP_{U|X}({u}|\bar{x})W_{V|X}(v|u)W_{Y|XV}(y|\bar{x},v)\Big|-t\\
& \stackrel{(b)}{=} \sum_{u, y}\Big|P_{UY}(u, y)- \sum_{\bar{x}\in \cX_u,v}(1-t)P_{\hat{U}\hat{X}}(u, \bar{x})W_{V|X}(v|u)W_{Y|XV}(y|\bar{x},v)\Big|-t\\
& \geq \sum_{u, y}\Big|P_{UY}(u, y)- \sum_{\bar{x}\in \cX_u,v}P_{\hat{U}\hat{X}}(u, \bar{x})W_{V|X}(v|u)W_{Y|XV}(y|\bar{x},v)\Big|\\
&\qquad-\inp{\sum_{u, y}\sum_{\bar{x}\in \cX_u,v}tP_{\hat{U}\hat{X}}(u, \bar{x})W_{V|X}(v|u)W_{Y|XV}(y|\bar{x},v)}-t\\
& = \sum_{u, y}\Big|P_{UY}(u, y)- \sum_{\bar{x}\in \cX_u,v}P_{\hat{U}\hat{X}}(u, \bar{x})W_{V|X}(v|\bar{x})W_{Y|XV}(y|\bar{x},v)\Big|-2t\\
&\stackrel{(c)}{=}\sum_{u, y}\Big|P_{UY}(u, y)- \sum_{\bar{x}}P_{\hat{U}}(u)P_{\hat{X}|\hat{U}}(\bar{x}|u) W_{Y|X}(y|\bar{x})\Big|-2t
\end{align*} where $(a)$ follows from \eqref{eq:C}, $(b)$ uses  \eqref{eq:cX_u} and \eqref{eq:DF0}, and $(c)$ follows from \eqref{eq:commonch2}.
Thus, 
\begin{align}
\sum_{u, y}\Big|P_{UY}(u, y)- \sum_{\bar{x}}P_{\hat{U}}(u)P_{\hat{X}|\hat{U}}(\bar{x}|u) W_{Y|X}(y|\bar{x})\Big|&\leq 2t + \sqrt{6\epsilon\ln{2}}+\rad\nonumber\\
&\leq\rad(2/\eta+1) +\sqrt{6\epsilon\ln{2}}.\nonumber
\end{align}
From \eqref{eq:bound 1}, recall that we choose $\epsilon$ to satisfy 
\begin{align}\sqrt{6\epsilon\ln{2}}<\frac{\sqrt{3}\rad\gamma\alpha}{\sqrt{7}+\sqrt{3}}. \label{eq:epsilon_choice}
\end{align} Hence,
\begin{align}
\sum_{u, y}\Big|P_{UY}(u, y)- \sum_{\bar{x}}P_{\hat{U}}(u)P_{\hat{X}|\hat{U}}(\bar{x}|u) W_{Y|X}(y|\bar{x})\Big|&\leq \rad\inp{\frac{2}{\eta}+1 +\frac{\sqrt{3}\gamma\alpha}{\sqrt{7}+\sqrt{3}}}.\label{eq:close_rate_1}
\end{align}
Further,
\begin{align}
&\sum_{u, y}\Big|\sum_{\bar{x}}P_{{U}}(u)P_{\hat{X}|\hat{U}}(\bar{x}|u) W_{Y|X}(y|\bar{x})- \sum_{\bar{x}}P_{\hat{U}}(u)P_{\hat{X}|\hat{U}}(\bar{x}|u) W_{Y|X}(y|\bar{x})\Big|\nonumber\\
&=\sum_{u, y}\Big|\sum_{\bar{x}}P_{\hat{X}|\hat{U}}(\bar{x}|u) W_{Y|X}(y|\bar{x})\inp{P_{{U}}(u)- P_{\hat{U}}(u)}\Big|\nonumber\\
&=\sum_{u}\sum_y\sum_{\bar{x}}P_{\hat{X}|\hat{U}}(\bar{x}|u) W_{Y|X}(y|\bar{x})\Big|{P_{{U}}(u)- P_{\hat{U}}(u)}\Big|\nonumber\\
&{=}\sum_{u}\Big|{P_{{U}}(u)- P_{\hat{U}}(u)}\Big|\nonumber\\
&\stackrel{(a)}{\leq} \rad\inp{\frac{2}{\eta}+1 +\frac{\sqrt{3}\gamma\alpha}{\sqrt{7}+\sqrt{3}}}.\label{eq:close_rate_2}
\end{align} where $(a)$ follows from \eqref{eq:close_rate_1}. From \eqref{eq:close_rate_1} and \eqref{eq:close_rate_2}, we obtain
\begin{align}\label{eq:close_rate}
\sum_{u, y}\Big|P_{UY}(u, y)- \sum_{\bar{x}}P_{{U}}(u)P_{\hat{X}|\hat{U}}(\bar{x}|u) W_{Y|X}(y|\bar{x})\Big|&\leq 2\rad\inp{\frac{2}{\eta}+1 +\frac{\sqrt{3}\gamma\alpha}{\sqrt{7}+\sqrt{3}}}.
\end{align}

Let $P_{{U}\hat{Y}}(u, y) = \sum_{x}P_{{U}}(u)P_{\hat{X}|\hat{U}}(x|u) W_{Y|X}(y|x), \, u\in \cU, y\in \cY$ where $P_{\hat{X}|\hat{U}}(x|u)>0$ only if $x\in \cX_u$ by definition (see \eqref{eq:DF0}) i.e., $P_{\hat{X}|\hat{U}}(x|u)>0$ only if $W_{V|X}(v|x) = W_{V|X}(v|u)$ for all $v$ (see \eqref{eq:defn_x_u}). Then, by \eqref{eq:R_achievable},
\begin{align}\label{eq:R_R}
I({U};\hat{Y})\geq \min_{\stackrel{P_{X|{U}}:\,P_{X|{U}}(x|u)>0}{ \text{ only if }W_{V|X}(v|x)= W_{V|X}(v|u) \forall v}}\min\inp{I({U};Y'), I({U};Z')}> R,
\end{align} where the mutual information in the term $\inb{\min\inp{I({U};Y'), I({U};Z')}}$ of \eqref{eq:R_R} is evaluated under $P_{UXY'Z'}(u, x, y, z) = P_U(u)P_{X|U}(x|u)W_{YZ|X}(y,z|x)$.   
Since mutual information $I(U;Y)$  is continuous in $P_{UY}$, from \eqref{eq:close_rate} and \eqref{eq:R_R}, {for sufficiently small $\rad>0$,} 
\begin{align}\label{eq:bound_2}
I(U;Y)-R>0.
\end{align} {Further, we choose $\epsilon>0$ small enough such that $I(U;Y)-R\geq 4\epsilon$ and \eqref{eq:bound 1} holds.} From \eqref{eq:EE}, this implies that
 \begin{align*}
 \zeta_{UXY}\leq \exp(-2n\epsilon).
 \end{align*}

\end{proof}
\end{proof}

\section{Discussion}\label{sec:disc}
While we considered the two-receiver broadcast channel, the results readily generalize to more than two receivers. For instance, for the three-receiver broadcast channel, the common channel may be defined analogous to Definition~\ref{def:common-channel} via a characteristic tripartite hypergraph whose hyperedges are the triples of symbols which occur together at the channel outputs of the three receivers with positive probability for some channel input symbol. The connected components of this hypergraph are the output symbols of the common channel. The definition of the effective input alphabet remains unchanged. The capacity expression in~\eqref{eq:capstr} is modified so that the inner minimum is of the three mutual information quantities corresponding to the three receivers instead of two. The proofs of converse and achievability can be verified to generalize with no significant changes required. 

As discussed in Remarks~\ref{rem:weak} and~\ref{rem:sharedrandomness}, the rate of communication with consensus that can be achieved is sensitive to how fast the error probability $\err^{(n)}$ is required to decay with the blocklength $n$. Here, we studied the most natural regime where the error probability decays exponentially, i.e., $-\log(\err^{(n)}) = \Omega(n)$, and found that the capacity remains unchanged as long as $\err^{(n)}$ is required to decay at least inverse linearly, i.e., $\err^{(n)}=o(1/n)$. Understanding the behaviour of capacity in regimes where this is further relaxed (for instance, to simply $\err^{(n)}\rightarrow 0$) would be of interest. The example from \iftoggle{long}{Appendix~\ref{app:weak}}{\cite[Appendix]{long}} shows that in these regimes, the presence of common or correlated randomness among the receivers which is unknown to the sender has an effect on the capacity.

In our model the receivers are passive. It would be interesting to study models where the receivers may also communicate with each other and/or with the sender. For instance, suppose in addition to the broadcast channel there are private noisefree links of unlimited capacity between every pair of users (i.e., between the two receivers, and between the sender and each receiver). As we mentioned in the introduction, byzantine consensus is known to be impossible in this setup in the absence of the broadcast channel~\cite{LamportSP82,Dolev82,FischerLM86}. A question of interest is to characterize the broadcast channels which permit byzantine consensus in this model. For the characterization of distributed sources which permit byzantine consensus in such a setup of private pairwise links, see~\cite{NarayananPSW23}.


\appendices
\section{The Curious Case of $\err^{(n)}$ Approaching 0 Slower Than $o(1/n)$}\label{app:weak}

The converse of Theorem~\ref{thm:Capacity} in Section~\ref{sec:proof} made use of the requirement that $\err^{(n)}=o(1/n)$. We will see that a converse cannot be shown if this is further relaxed to $\err^{(n)}=o(1)$. To this end, we show an example where $\capstr= 0$, but a positive rate is achievable with $\err^{(n)}=o({1}/{n^{\frac{1}{2}-\epsilon}})$, for any $\epsilon>0$. The example is in fact the independent binary erasure broadcast channel (i.e., the two-step binary erasure broadcast channel with $p=1$) of Section~\ref{sec:shortconverse}, but with additional common randomness shared by the decoders which is unknown to the sender. Without the additional common randomness, even two messages cannot be communicated over this channel with $\err^{(n)}\rightarrow0$ (see Remark~\ref{rem:Witsenhausen}). We also know that the presence of such randomness does not affect $\capstr$ which requires $\err^{(n)}=o(1/n)$ (see Remark~\ref{rem:commonrand}). 
However, we will show that the availability of common randomness among the receivers unknown to the sender facilitates communication with consensus at non-zero rates over the independent binary erasure channel with $\err^{(n)}=o({1}/{n^{\frac{1}{2}-\epsilon}})$ for any $\epsilon>0$.

The independent binary erasure broadcast channel with common randomness is $\ch_{(YS)(ZS)|X}=\ch_{Y|X}\ch_{Z|X}\ch_{S|X}$, where the channels $\ch_{Y|X}$ and $\ch_{Z|X}$ are identical binary erasure channels (BEC) with erasure probability $0<q<1$ and $\ch_{S|X}$ is a completely noisy channel (i.e., $\cS=\{0,1\}$ and $\ch_{S|X}(0|x)=\frac{1}{2}$ for all $x$). We denote the unerased output symbols of the independent erasure channels without the $\tilde{\phantom{a}}$ of Section~\ref{sec:shortconverse}. Specifically, $\cY=\cZ=\{0,1,e\}, \cX=\{0,1\}$ and $\ch_{Y|X}(e|x)=\ch_{Z|X}(e|x)=1-\ch_{Y|X}(x|x)=1-\ch_{Z|X}(x|x)=q, x\in\cX$.
Note that the characteristic graph has two connected components: for $v\in\{0,1\}$, $G_v$ is a connected bipartite graph on vertices $(\{0,1,e\}\times\{v\}) \cup (\{0,1,e\}\times\{v\})$.
The common channel output $V$ is $S$. Thus the common channel is the completely noisy channel $\ch_{V|X}(v|x)=\frac{1}{2}$ for all $x,v\in\{0,1\}$ and hence, by Theorem~\ref{thm:Capacity}, $\capstr=0$.

The following theorem shows that consensus is feasible over $\ch_{YZ|X}$ with positive rates (for erasure probability $q<1/4$) if the probability of error is only required to fall as $o({1}/{n^{\frac{1}{2}-\epsilon}})$ for any $\epsilon>0$ as the block length $n\rightarrow \infty$.

\begin{theorem}\label{thm:min-dist}
    Let $R<1-H(2q)$ and $\epsilon>0$. For sufficiently large $n$ there is an $(n,2^{nR})$ consensus code with $\err^{(n)}=o(n^{-\frac{1}{2}+\epsilon})$.
\end{theorem}
\begin{proof}
    The following claim follows from the Gilbert-Varshamov bound; we prove this later for completeness.
    \begin{claim}\label{clm:gv}
        For all $R<1-H(2q)$, there exists $\delta>0$ such that, for every $n$ there is an $\enc:[1:2^{nR}]\rightarrow\{0,1\}^n$ such that $d_{\text{H}}(\enc(m),\enc(m'))>n(2q+\delta)$ for all distinct $m,m'\in[1:2^{nR}]$.
    \end{claim}

    We will employ the code in Claim~\ref{clm:gv} with the decoders $\decB$ and $\decC$ described below:
    For $x^n\in\{0,1\}^n$ and $y^n\in\{0,1,e\}^n$, we write $x^n\blacktriangleright y^n$ if $y_i\in\{x_i,e\}$ for all $i\in[1:n]$.
    Let $\ell$ be the integer such that  $\frac{n\delta}{8}<2^\ell\le \frac{n\delta}{4}\le2^{\ell+1}$.
    Define $h$ to be the function which takes $n$-length bit strings as arguments, drops all but the first $\ell$ bits, and returns the integer whose binary representation is given by these $\ell$ bits. 
    Hence, $h$ maps a uniform distribution over $\{0,1\}^n$ to a uniform distribution over $[1:2^\ell]$.

    Let $y^n,z^n\in\{0,1,e\}^n$ and $s^n\in\{0,1\}^n$.
    The decoder output $\decB(y^n,s^n)=m$ if there is a unique $m\in[1:2^{nR}]$ such that
    \[ \exists \bx^n \text{ s.t. } \bx^n\blacktriangleright y^n \text{ and } d_{\text{H}}(\bx^n,\enc(m))\le h(s^n) .\]
    If no such unique $m$ exists, $\decB(y^n,s^n)=\bot$.
    Decoder output $\decC(z^n,s^n)$ is similarly defined.

    Let $x^n$ be the string sent over the channel, and $(Y^n,S^n)$ and $(Z^n,S^n)$ be the random variables corresponding to the strings received by Bob and Carol, respectively.
    Define the event
    \begin{align}
        \cB_{x^n}=\left(\exists m\in[1:2^{nR}] \text{ s.t. } d_{\text{H}}(\enc(m),x^n)>n\left(q+\frac{\delta}{2}\right) \text{ and } \exists \bx^n\blacktriangleright Y^n \text{ s.t. }  d_{\text{H}}(\bx^n,\enc(m))\le h(S^n)\right).
    \end{align}
    We will first prove the theorem assuming the following claim.
    \begin{claim}\label{clm:event-weak-ach}
        When $k=\frac{1}{48}$, $\pr(\cB_{x^n}|x^n) \le e^{-\frac{nk}{q}\delta^2}$.
    \end{claim}
    Consider three possibilities for the input $x^n$: (i) $x^n=\enc(m)$ for some $m$, (ii) $d_{\text{H}}(x^n,\enc(m))>n(q+\frac{\delta}{2})$ for all $m\in[1:2^{nR}]$, and (iii) $0<d_{\text{H}}(x^n,\enc(m))\le n(q+\frac{\delta}{2})$ for some $m$.
    \medskip

    \noindent Case (i): We have $x^n\blacktriangleright y^n$ and $d_{\text{H}}(x^n,\enc(m))=0$.
    Moreover, for all $m'\neq m$, $d_{\text{H}}(x^n,\enc(m'))=d_{\text{H}}(\enc(m),\enc(m'))>n(2q+\delta)$.
    Hence, by Claim~\ref{clm:event-weak-ach},
    \begin{align*}
        \pr(\decB(Y^n,S^n)\neq m|x^n) = \pr(\exists m'\neq m, \bx^n\blacktriangleright Y^n \text{ s.t. } d_{\text{H}}(\bx^n,\enc(m'))\le h(s^n)|x^n) \le \pr(\cB_{x^n}|x^n) \le e^{-\frac{nk}{q}\delta^2}.
    \end{align*}
    Similarly, $\pr(\decC(Z^n,S^n)\neq m|x^n)\le e^{-\frac{nk}{q}\delta^2}$.
    By a union bound, $\lambda_m\le 2e^{-\frac{nk}{q}\delta^2}$.

    \medskip

    \noindent Case (ii): Since $d_{\text{H}}(x^n,\enc(m))>n(q+\frac{\delta}{2})$ for all $m\in[1:2^{nR}]$, by Claim~\ref{clm:event-weak-ach},
    \begin{align*}
        \pr(\decB(Y^n,S^n)\neq\bot|x^n)\le\pr(\exists m\in[1:2^{nR}], \bx^n\blacktriangleright y^n \text{ s.t. } d_{\text{H}}(\bx^n,\enc(m))\le h(S^n)|x^n) \le \pr(\cB_{x^n}|x^n)
        \le e^{-\frac{nk}{q}\delta^2}.
    \end{align*}
    Similarly, $\pr(\decB(Y^n,S^n)\neq\bot|x^n)\le e^{-\frac{nk}{q}\delta^2}$.
    By a union bound, $\pr(\decB(Y^n,S^n)=\decC(Z^n,S^n)=\bot|x^n)\ge 1-2e^{-\frac{nk}{q}\delta^2}$.

    \medskip

    \noindent Case (iii): Define $\bar{X}_\mathsf{B}^n$ such that $\bar{X}_\mathsf{B}(i)$ (the coordinate $i$ of $\bar{X}_\mathsf{B}$) is $Y_i$ if $Y_i\neq e$ and $\enc_i(m)$ (the coordinate $i$ of $\enc(m)$), otherwise.
    By definition, $\bar{X}_\mathsf{B}^n\blacktriangleright Y^n$.
    Moreover, for all $\bx^n\blacktriangleright Y^n$,
    \begin{align}\label{eq:psi-sub}
        d_{\text{H}}(\bx^n,\enc(m)) = |\{i: Y_i\notin \{\enc_i(m),e\}\}| + |\{i: \bx_i\neq \enc_i(m), Y_i= e\}|
        \ge |\{i: Y_i\notin \{\enc_i(m),e\}\}| = d_{\text{H}}(\bar{X}_\mathsf{B}^n,\enc(m))
    \end{align}
    Similarly, we define $\bar{X}_\mathsf{C}^n$ with $Z^n$ in lieu of $Y^n$.

    Suppose, for all $m'\neq m$, there exists no $\bx^n$ such that $d_{\text{H}}(\bx^n,\enc(m'))\le h(S^n)$ and $\bx^n\blacktriangleright Y^n$ or $\bx^n\blacktriangleright Z^n$.
    Then, since $\bar{X}_\mathsf{B}^n\blacktriangleright Y^n$ and $\bar{X}_\mathsf{C}^n\blacktriangleright Z^n$, the decoders agree on $m$ if both $d_{\text{H}}(\bar{X}_\mathsf{B}^n,\enc(m))$ and $d_{\text{H}}(\bar{X}_\mathsf{C}^n,\enc(m))$ are at most $h(S^n)$, and, by \eqref{eq:psi-sub}, they agree on $\bot$ if both are more than $h(S^n)$.
    Defining $\mu=(1-q)d_{\text{H}}(x^n,\enc(m))$,
    \begin{align}
        &\pr(\decB(Y^n,S^n)=\decC(Z^n,S^n)|x^n)\nonumber\\
        &\quad\ge \pr(\not\exists m'\neq m \text{ and } \bx^n \text{ s.t. } d_{\text{H}}(\bx^n,\enc(m'))\le h(S^n) \text{ and } (\bx^n\blacktriangleright Y^n \text{ or } \bx^n\blacktriangleright Z^n ),
        d_{\text{H}}(\bar{X}_\mathsf{B}^n,\enc(m))\in [\mu-n^{\frac{1}{2}+\epsilon},\mu+n^{\frac{1}{2}+\epsilon}],\nonumber\\
        &\hspace{5cm} d_{\text{H}}(\bar{X}_\mathsf{C}^n,\enc(m))\in [\mu-n^{\frac{1}{2}+\epsilon},\mu+n^{\frac{1}{2}+\epsilon}]
        \text{ and } h(S^n)\notin[\mu-n^{\frac{1}{2}+\epsilon},\mu+n^{\frac{1}{2}+\epsilon}]|x^n).\label{eq:weak-ach-0}
    \end{align}
    For any $m'\neq m$, by the triangle inequality,
    \begin{align*}
        d_{\text{H}}(x^n,\enc(m')) \ge d_{\text{H}}(\enc(m),\enc(m')) - d_{\text{H}}(x^n,\enc(m)) > n(2q+\delta)-n(q+\delta/2) = n(q+\delta/2).
    \end{align*}
    Hence, $\pr(\exists m'\neq m, \bx^n\blacktriangleright Y^n \text{ s.t. } d_{\text{H}}(\bx^n,\enc(m'))\le h(S^n)|x^n)=\pr(\cB_{x^n|x^n})\le e^{-\frac{nk}{q}\delta^2}$ by Claim~\ref{clm:event-weak-ach}.
    By a union bound,
    \begin{align}\label{eq:weak-ach-1}
        \pr(\exists m'\neq m \text{ and } \bx^n \text{ s.t. } d_{\text{H}}(\bx^n,\enc(m'))\le h(S^n) \text{ and } (\bx^n\blacktriangleright Y^n \text{ or } \bx^n\blacktriangleright Z^n )|x^n) \le 2e^{-\frac{nk}{q}\delta^2}.
    \end{align}
    For all $i\in[1:n]$, $Y_i=e$ independently with probability $q$.
    Hence, $d_{\text{H}}(\bar{X}_\mathsf{B}^n,\enc(m))=|\{i:Y_i\notin \{\enc_i(m),e\}\}|=|\{i:\enc_i(m)\neq x_i \text{ s.t. } Y_i\neq e\}|$ is a binomial distribution with mean $\mu=(1-q)d_{\text{H}}(x^n,\enc(m))$ and success probability $(1-q)$.
    For any $\epsilon>0$, by the Chernoff bound,
    \begin{align*}
        \pr(d_{\text{H}}(\bar{X}_\mathsf{B}^n,\enc(m))\notin[\mu-n^{\frac{1}{2}+\epsilon},\mu+n^{\frac{1}{2}+\epsilon}]|x^n)
        = \pr(|d_{\text{H}}(\bar{X}_\mathsf{B}^n,\enc(m))-\mu|>n^{\frac{1}{2}+\epsilon}|x^n)
        \le 2e^{-\frac{\mu}{3}\left(\frac{n^{\frac{1}{2}+\epsilon}}{\mu}\right)^2} \le 2e^{-\frac{n^{2\epsilon}}{3}}.
    \end{align*}
    The final inequality used the bound $\mu\le n$.
    By a union bound,
    \begin{align}
        \pr(d_{\text{H}}(\bar{X}_\mathsf{B}^n,\enc(m))\notin[\mu-n^{\frac{1}{2}+\epsilon},\mu+n^{\frac{1}{2}+\epsilon}] \text{ or } d_{\text{H}}(\bar{X}_\mathsf{C}^n,\enc(m))\notin[\mu-n^{\frac{1}{2}+\epsilon},\mu+n^{\frac{1}{2}+\epsilon}]|x^n)
        \le 4e^{-\frac{n^{2\epsilon}}{3}}.\label{eq:weak-ach-2}
    \end{align}
    Using \eqref{eq:weak-ach-1} and \eqref{eq:weak-ach-2}, we union bound the LHS of \eqref{eq:weak-ach-0} as
    \begin{align*}
        \pr(\decB(Y^n,S^n)=\decC(Z^n,S^n)|x^n)
        &\ge 1-2e^{-\frac{nk}{q}\delta^2}-4e^{-\frac{n^{2\epsilon}}{3}}-\pr(h(S^n)\notin [\mu-n^{\frac{1}{2}+\epsilon},\mu+n^{\frac{1}{2}+\epsilon}])\\
        &\stackrel{(a)}{\ge} 1-2e^{-\frac{nk}{q}\delta^2}-4e^{-\frac{n^{2\epsilon}}{3}} + \frac{2n^{\frac{1}{2}+\epsilon}}{2^\ell}
        \stackrel{(c)}{\ge} 1-\frac{5}{\epsilon}n^{-\frac{1}{2}+\epsilon}.
    \end{align*}
    Here, (a) used the fact that $h(S^n)$ is distributed uniformly over $[1:2^\ell]$ independent of $x^n$;
    and (b) used the bounds $2^\ell>\frac{n\epsilon}{8}$ and $\frac{1}{\epsilon}n^{-\frac{1}{2}+\epsilon}\ge 2e^{-\frac{nk}{q}\delta^2}+4e^{-\frac{n^{2\epsilon}}{3}}$ for sufficiently large $n$.
    We conclude the proof by proving Claims~\ref{clm:event-weak-ach} and~\ref{clm:gv}.
    \begin{proof}[Proof of Claim~\ref{clm:event-weak-ach}]
        Size of $\{i:Y_i=e\}$ is distributed according to the binomial distribution with mean $nq$ and success probability $q$.
        By the Chernoff bound, there exists a constant $k$ such that
        \begin{align}\label{eq:chernoff-0}
            \pr\left(|\{i:Y_i=e\}|\le n\left(q+\frac{\delta}{4}\right)\right) = \pr\left(|\{i:Y_i=e\}|\le nq\left(1+\frac{\delta}{4q}\right)\right) \ge 1-e^{-nq\frac{\delta^2}{16q^2\left(2+\frac{\delta}{4q}\right)}} \ge 1-e^{-n\frac{\delta^2}{48q}}.
        \end{align}
        In the final inequality, we used $16(2+\frac{\delta}{4q})\le 48$ for sufficiently small $\delta$.
        Conditioned on this event, for all $m'$ such that $d_{\text{H}}(x^n,\enc(m'))>n(q+\delta/2)$, and $\bx^n\blacktriangleright Y^n$,
        \begin{align*}
            d_{\text{H}}(\bx^n,\enc(m'))\ge|\{i:Y_i\notin\{\enc_i(m'),e\}\}| = \{i:x_i\neq \enc_i(m')\}|-|\{i:Y_i=e\}|>n(q+\frac{\delta}{2})-n\left(q+\frac{\delta}{4}\right)>\frac{n\delta}{4}.
        \end{align*}
        Claim now follows from the fact that $h(S^n)\le2^\ell\le \frac{n\delta}{4}$.
    \end{proof}
    \begin{proof}[Proof of Claim~\ref{clm:gv}]
        By the Gilbert-Varshamov bound~\cite[Theorem 4.2.1]{GuruswamiRS19}, for every $0<\gamma<\frac{1}{2}$, there exists a linear code with rate $1-H(\gamma)$ and relative distance $\gamma$ (i.e., Hamming distance more than $n\gamma$ between any two codewords).
        Choose $\delta>0$ small enough that $2q+\delta<\frac{1}{2}$ and $1-H(2q+\delta)>R$.
        Then, for any $n$, there is a $\enc:[1:2^{nR}]\rightarrow\{0,1\}^n$ such that the hamming distance $d_{\text{Hamming}}(\enc(m),\enc(m'))>n(2q+\delta)$ for all distinct $m,m'\in[1:2^{nR}]$.
        The claim follows.
    \end{proof}
    This concludes the proof of the theorem.
\end{proof}

\section{Missing proofs from Section~\ref{sec:proof}}\label{app:polytope}
\subsection{Proof of Lemma~\ref{lem:polytope}}
\begin{proof}
We will first show \eqref{lemma5:main_statement}. Recall that $u\in\cU$ corresponds to a vertex of the convex polytope $Q$ which is the convex hull of the $|\cV|$-dimensional vectors $\{\ch_{V|X}(.|x), x\in\cX\}$, i.e., $\ch_{V|X}(.|u)$ is a vertex of the polytope. We defined $\cX_u=\{x\in\cX: \ch_{V|X}(.|x)=\ch_{V|X}(.|u)\}$, i.e., the set of all letters in $\cX$ which correspond to the same vertex as $u$. Since $\ch_{V|X}(.|u)$ is a vertex of the polytope, there cannot exist a p.m.f. $\eta_x,x\in \cX-\cX_u$ such that
\begin{align} \ch_{V|X}(.|u) = \sum_{x\in \cX-\cX_u} \eta_x \ch_{V|X}(.|x).\label{eq:polytope-lemma}\end{align}

To prove the lemma, suppose for the sake of contradiction that there is a p.m.f. $\lambda_x\geq0, x\in\cX$ such that 
\[\sum_{x\in\cX} \lambda_x\ch_{V|X}(.|x)= \ch_{V|X}(.|u),\]
and 
\[ \Lambda:=\sum_{x\in \cX-\cX_u} \lambda_x > 0.\]
Then, 
\begin{align*}
\sum_{x\in \cX-\cX_u} \lambda_x \ch_{V|X}(.|x) 
&=  \ch_{V|X}(.|u) - \sum_{x\in\cX_u} \lambda_x\ch_{V|X}(.|x)\\
&\stackrel{\text{(a)}}{=} \left(1-\sum_{x\in\cX_u}\lambda_x\right) \ch_{V|X}(.|u)\\
&= \left(\sum_{x\in\cX-\cX_u} \lambda_x\right) \ch_{V|X}(.|u)\\
&= \Lambda \ch_{V|X}(.|u),
\end{align*}
where (a) follows from the definition of $\cX_u$ above.  
Now consider the p.m.f., $\eta_x, x\in\cX-\cX_u$,
\[ \eta_x = \frac{\lambda_x}{\Lambda}.\]
Then, 
\[ \sum_{x\in\cX-\cX_u} \eta_x \ch_{V|X}(.|x) = \ch_{V|X}(.|u),\]
which contradicts the fact that there is no p.m.f. $\eta_x, x\in\cX-\cX_u$ which satisfies~\eqref{eq:polytope-lemma}.

Next, we will show \eqref{eq:cX_u}. 
For $u\in \cU$, suppose $x\in \cX_u$, then by \eqref{eq:Ptilde},
\begin{align*}
\sum_{u'}\tP_{U|X}(u'|x)W_{V|X}(v|u')&=W_{V|X}(v|x){=}W_{V|X}(v|u).  
\end{align*} Using \eqref{lemma5:main_statement}, we conclude that $\tP_{U|X}(u|x) = 1$. Now suppose, $\tP_{U|X}(u|x) = 1$ for $x\in \cX$. Then, by \eqref{eq:Ptilde}, $W_{V|X}(v|u) = W_{V|X}(v|x)$ which implies that $x\in \cX_u$.
\end{proof}

\subsection{Proof of Claim~\ref{claim:defn_gamma}}\label{sub:proof_of_claim_ref_claim_defn_gamma}
\begin{proof}
Consider any channel $W_{YZ|X}$ with   $|\cU|\geq 2$, i.e. $\capptop(\ch_{V|X})>0$. 
Lemma~\ref{lem:polytope} implies that  
for any ${u'}\in \cU$ and conditional distribution $P_{U|U'}$ mapping symbols in $\cU$ to symbols in $\cU$, if
\begin{align}
\sum_{u\in \cU}P_{U|U'}(u|u')W_{V|X}(v|u) = W_{V|X}(v|u') \text{ for every }v\in \cV,\label{eq:defn_U}
\end{align} then $P_{U|U'}(u|u') = 1_{\inb{u=u'}}$. This also implies that the $\tP_{U|X}$ in Definition~\ref{def:effective-input} is such that 
\begin{align}
\tP_{U|X}(u|u') = 1_{\inb{u = u'}}, \, \,u, u'\in \cU.\label{eq:uu'} 
\end{align}
In other words, for any channel with $\capptop(\ch_{V|X})>0$ (i.e. $|\cU|\geq 2$), there exists $\gamma>0$ such that
\begin{align}
\min_{u'\in\cU}\min_{\stackrel{P_{U} \text{ with support}}{\text{ on }\cU\setminus\{u'\}}}\sum_{v\in \cV}\big|\sum_{u\in \cU}P_{U}(u)W_{V|X}(v|u) - W_{V|X}(v|u')\big|=\gamma\label{eq:defn_U''}
\end{align}
Consider $u'\in \cU$ and $P_{U|U'}$ such that $P_{U|U'}(u'|u')\neq 1$,
\begin{align}
&\sum_{v\in \cV}\big|\sum_{u\in \cU}P_{U|U'}(u|u')W_{V|X}(v|u) - W_{V|X}(v|u')\big|\nonumber\\
&=\sum_{v\in \cV}\big|\sum_{u\in \cU\setminus\{u'\}}P_{U|U'}(u|u')W_{V|X}(v|u) - (1-P_{U|U'}(u'|u'))W_{V|X}(v|u')\big|\nonumber\\
&\stackrel{(a)}{=}(1-P_{U|U'}(u'|u'))\inp{\sum_{v\in \cV}\big|\sum_{u\in \cU\setminus\{u'\}}P'_{U|U'}(u|u')W_{V|X}(v|u) - W_{V|X}(v|u')\big|}\nonumber\\
&\stackrel{(b)}{\geq} (1-P_{U|U'}(u'|u'))\gamma\nonumber
\end{align} where in $(a)$, we defined  $P'_{U|U'}(u|u'):=P_{U|U'}(u|u')/(1-P_{U|U'}(u'|u'))$, $u\in \cU\setminus\{u'\}$, and $(b)$ follows from \eqref{eq:defn_U''}. 
Hence, for any $u'\in \cU$ and $P_{U|U'}$,
\begin{align}
&\sum_{v\in \cV}\big|\sum_{u\in \cU}P_{U|U'}(u|u')W_{V|X}(v|u) - W_{V|X}(v|u')\big|{\geq} (1-P_{U|U'}(u'|u'))\gamma
\end{align} which holds with equality if $P_{U|U'}(u'|u') = 1$.
\end{proof}

\subsection{Proof of Lemma~\ref{lemma:dis_properties}}\label{proof:lemma10_prop}
\begin{proof}
\noindent For $\vecu, \vecu'\in \cU^n$, by Definition~\ref{defn:dis}, $d(\vecu, \vecu')= d(\vecu', \vecu) = \pr(U\neq U')$ under the joint type $P_{UU'}$ given by $(\vecu, \vecu')\in \cT^n_{UU'}.$ To see this, by \eqref{eq:cX_u_particular}, $\tP_{U|X}(\tilde{u}|u') = 1_{\inb{\tilde{u}=u'}}, \, \tilde{u}, u'\in \cU$. Hence,  $P_{UU'\tilde{U}}(u, u, \tilde{u}) = P_{UU'}(u,u' )1_{\inb{\tilde{u}=u'}}$. Thus, $\tilde{U} = U'$ and 
\begin{align*}
d(\vecu, \vecu') = d(P_{UU'}) = \pr(U\neq \tilde{U})= \pr(U\neq U') = d(\vecu', \vecu).
\end{align*}  Next, we will show that for $(\vecu, \vecx, \vecu')$, $d(\vecu, \vecx) + d(\vecu', \vecx)\geq d(\vecu, \vecu')$. Let $(\vecu, \vecu',\vecx)\in \cT^n_{UU'X}$. Let $P_{UU'X\tilde{U}}(u, u', x, \tilde{u}) = P_{UU'X}(u, u', x)\tP_{U|X}(\tilde{u}|x)$ $u, u', \tilde{u}\in U$ and $x\in \cX$. Then 
\begin{align*}
d(\vecu, \vecx) + d(\vecu', \vecx) &= \pr(U\neq\tilde{U}) + \pr(U'\neq \tilde{U})\\
&\geq \pr((U\neq\tilde{U})\cup(U'\neq\tilde{U}))\\
&\geq \pr(U\neq U')\\
&=d(\vecu, \vecu').
\end{align*}
\end{proof}

\subsection{Proof of Lemma~\ref{lemma:codebook}}
\begin{proof}\label{app:proof:codebook}

We use a random coding argument to show the existence of a codebook satisfying properties \eqref{eq:code1} and \eqref{eq:code2}. Let $\cT^n$ denote the type class of $P$. We generate $\exp\inp{nR}$ (recall that $\log$ and $\exp$ are with respect to base $2$; in particular, $\exp\inp{nR}=2^{nR}$) independent random codewords ${\vecU_{1}, \vecU_{2}, \ldots, \vecU_{K}}$,  each distributed uniformly on $\cT^n$.

We will show that the probability that statement \eqref{eq:code1} or statement \eqref{eq:code2} for any fixed joint type $P_{UX}\in \cP^n\inp{\cU\times \cX}$ and $\vecx\in \cX^n$, do not hold falls doubly exponentially in $n$. Since $|\cX^n|$ grows only exponentially in $n$ and $|\cP^n\inp{\cU\times \cX}|$ polynomially in $n$, a union bound will imply  the existence of a codebook satisfying properties \eqref{eq:code1} and \eqref{eq:code2}.
We will use the concentration result \cite[Lemma A1]{CsiszarN88}, which we restate below for ready reference.
\begin{lemma}{\cite[Lemma A1]{CsiszarN88}}\label{lemma:A1}
Let $S_1, \ldots, S_{K}$ be arbitrary random variables, and let $\zeta_i(S_1, \ldots, S_{i})$ be arbitrary with $0\leq \zeta_i\leq 1, \, i = 1, 2, \ldots, K$. Then the condition 
\begin{align}
\bbE\insq{\zeta_i(S_1, \ldots, S_i)|S_1, \ldots, S_{i-1}}\leq a \text{ a.s.},\, i = 1, 2, \ldots, K\label{eq:A1}
\end{align}
implies that 
\begin{align}
\pr\inb{\frac{1}{K}\sum_{i=1}^{K}\zeta_i(S_1, \ldots, S_i)>t}\leq \exp\inb{-K\inp{t-a\log{e}}}.\label{eq:A2}
\end{align}
\end{lemma}
We will first analyze \eqref{eq:code2}. Let $\cE_1$ be the event
\begin{align*}
\cE_1 = \inb{|\{i\in [1:K]:(\vecU_i, \vecx)\in \cT^n_{UX}\}|>\exp{\left(n\left(\left|R-I(U;X)\right|^{+}+\epsilon\right)\right)}}.
\end{align*}
Suppose $(\vecU_1, \ldots, \vecU_K)$ are the random variables $(S_1, \ldots, S_K)$ in Lemma~\ref{lemma:A1}. Let
\[\zeta_i(\vecU_1, \ldots, \vecU_i) := \begin{cases} 1, &\text{if }\vecU_i\in \cT^n_{U|X}(\vecx),\\ 0, &\text{otherwise.}\end{cases}\]
Then, 
\begin{align*}
\bbE\insq{\zeta_i(\vecU_1, \ldots, \vecU_i)|\vecU_1, \ldots, \vecU_{i-1}} &= \pr(\vecU_i\in \cT^n_{U|X}(\vecx)) = \frac{|\cT^n_{U|X}(\vecx)|}{|\cT^n_{U}|}\\
&\stackrel{(a)}{\leq} \frac{\exp\inb{nH(U|X)}}{(n+1)^{-|\cU|}\exp\inb{nH(U)}}\\
& = (n+1)^{|\cU|}\exp{\inb{-nI(U;X)}}=:a, 
\end{align*} where $(a)$ uses \eqref{eq:type_property2} and \eqref{eq:type_property3}.
Let $t_1 = \frac{1}{K}\exp\inb{n\inp{|R-I(U;X)|^{+}+\epsilon}}$. Then,
\begin{align}
\pr(\cE_1) &= \pr\inb{\frac{1}{K}\sum_{i=1}^{K}\zeta_i(\vecU_1, \ldots, \vecU_i)>t_1}\nonumber\\
&\stackrel{(a)}{\leq} \exp\inb{-K\inp{t_1-a\log{e}}}\nonumber\\
& \stackrel{(b)}= \exp\inb{-\inp{\exp\inb{n\inp{|R-I(U;X)|^{+}+\epsilon}}-(n+1)^{|\cU|}\exp{\inb{n\inp{R-I(U;X)}}}\log{e}}}\nonumber\\
& \stackrel{(c)}{\leq} \exp\inb{-\inp{\exp\inb{n\inp{|R-I(U;X)|^{+}+\epsilon}}-\frac{1}{2}{\exp\inb{n\inp{|R-I(U;X)|^{+}+\epsilon}}}}}\nonumber\\
& =\exp\inb{-\frac{1}{2}{\exp\inb{n\inp{|R-I(U;X)|^{+}+\epsilon}}}}\nonumber\\
&\leq \exp\insq{-\frac{1}{2}\exp{\inp{n\epsilon}}}, \label{eq:E_1}
\end{align}
where $(a)$ follows from \eqref{eq:A2}, $(b)$ follows by noting that $K = \exp\inp{nR}$ and $(c)$ holds for large enough $n$ such that $(n+1)^{|\cU|}{\log{e}}\leq \exp\inp{n\epsilon}/2$.
To analyze \eqref{eq:code1}, let $\cP:= \inb{P_{UU'}\in \cP^n\inp{\cU\times\cU}: P_U=P_{U'}=P, \, d({P_{U U'}})< 2\rad}$. Note that 
\begin{align*}
&\pr\inp{\frac{1}{K}|\{i\in [1:K]: d(\vecU_i,\vecU_j)< 2\rad\text{ for some }j\in [1:K], \, j\neq i\}|\ge \exp\left(-\frac{n\epsilon}{2}\right)}\\
&=\pr\inp{\frac{1}{K}|\{i\in [1:K]: \exists\, P_{UU'}\in \cP\ \text{ s.t. }(\vecU_i,\vecU_j)\in \cT^n_{UU'}\text{ for some }j\in [1:K], \, j\neq i|\ge \exp\left(-\frac{n\epsilon}{2}\right)}.
\end{align*}
We will show that 
\begin{align}
&\pr\inp{\frac{1}{K}|\{i\in [1:K]: \exists\, P_{UU'}\in \cP\ \text{ s.t. }(\vecU_i,\vecU_j)\in \cT^n_{UU'}\text{ for some }j\in [1:K], \, j\neq i\}|\ge \exp\left(-\frac{n\epsilon}{2}\right)}\nonumber\\
&\leq 2\exp\inp{-\frac{1}{4}\exp\inp{\frac{n\epsilon}{2}}}\label{eq:e_22}
\end{align} and use an expurgation argument to complete the proof.
To show \eqref{eq:e_22}, we first note that
\begin{align}
&\pr\inp{\frac{1}{K}|\{i\in [1:K]: \exists\, P_{UU'}\in \cP\ \text{ s.t. }(\vecU_i,\vecU_j)\in \cT^n_{UU'}\text{ for some }j\in [1:K], \, j\neq i\}|\ge \exp\left(-\frac{n\epsilon}{2}\right)}\nonumber\\
&\leq \pr\Bigg(\frac{1}{K}\Big(|\{i\in [1:K]:\exists\, P_{UU'}\in \cP\ \text{ s.t. } (\vecU_i,\vecU_j)\in \cT^n_{UU'}\text{ for some }j\in [1:K], \, j< i\}|\nonumber\\
&\quad+|\{i\in [1:K]: \exists\, P_{UU'}\in \cP\ \text{ s.t. }(\vecU_i,\vecU_j)\in \cT^n_{UU'}\text{ for some }j\in [1:K], \, j> i\}|\Big)\ge \exp\left(-\frac{n\epsilon}{2}\right)\Bigg)\nonumber\\
&\stackrel{(a)}{\leq}\pr\inp{\frac{1}{K}|\{i\in [1:K]:\exists\, P_{UU'}\in \cP\ \text{ s.t. } (\vecU_i,\vecU_j)\in \cT^n_{UU'}\text{ for some }j\in [1:K], \, j< i\}|\ge \frac{1}{2}\exp\left(-\frac{n\epsilon}{2}\right)}\nonumber\\
&\quad+\pr\inp{\frac{1}{K}|\{i\in [1:K]: \exists\, P_{UU'}\in \cP\ \text{ s.t. }(\vecU_i,\vecU_j)\in \cT^n_{UU'}\text{ for some }j\in [1:K], \, j> i\}|\ge \frac{1}{2}\exp\left(-\frac{n\epsilon}{2}\right)}\nonumber\\
&\stackrel{(b)}{\leq} 2\pr\inp{\frac{1}{K}|\{i\in [1:K]: \exists\, P_{UU'}\in \cP\ \text{ s.t. }(\vecU_i,\vecU_j)\in \cT^n_{UU'}\text{ for some }j\in [1:K], \, j< i\}|\ge \frac{1}{2}\exp\left(-\frac{n\epsilon}{2}\right)},\label{eq:e_2}
\end{align}
where $(a)$ follows from a union bound and $(b)$ uses the symmetry of the random codebook. Thus, we only need to analyze \eqref{eq:e_2}.
For any $(\vecu_1, \ldots, \vecu_{i-1})$, let
\[\xi_i(\vecu_1, \ldots, \vecu_i) := \begin{cases} 1, &\text{if }\vecu_i\in \cup_{j<i}\cT^n_{U|U'}(\vecu_j),\text{ for some }P_{UU'}\in \cP,\\ 0, &\text{otherwise.}\end{cases}\]
We will now apply Lemma~\ref{lemma:A1}. Suppose $(\vecU_1, \ldots, \vecU_K)$ are the random variables $(S_1, \ldots, S_K)$ and $\xi_i,\, i\in [1:K]$ correspond to the functions $\zeta_i,\, i\in [1:K]$. Then, \begin{align*}
&\bbE\insq{\xi_i(\vecU_1, \ldots, \vecU_i)|(\vecU_1, \ldots, \vecU_{i-1})=(\vecu_1, \ldots, \vecu_{i-1})} \\
&\stackrel{(a)}{=}\pr\inp{\cup_{P_{UU'}\in \cP}\cup_{j<i}\inb{\vecU_i\in \cT^n_{U|U'}(\vecu_j)}}\\
&\leq\sum_{P_{UU'}\in \cP}\sum_{j<i}\pr\inp{\inb{\vecU_i\in \cT^n_{U|U'}(\vecu_j)}}\\
&\leq\sum_{P_{UU'}\in \cP}\frac{|\cT^n_{U|U'}(\vecu_j)|}{|\cT^n_{U}|}\\
&\stackrel{(b)}{\leq}\sum_{P_{UU'}\in \cP}\sum_{j<i}\frac{\exp\inb{nH(U|U')}}{(n+1)^{-|\cU|}\exp\inb{nH(U)}}\\
&\leq\sum_{P_{UU'}\in \cP}(n+1)^{|\cU|}\exp{\inb{n\inp{R-I(U;U')}}} 
\end{align*} where $(a)$ follows from $\vecU_i\indep \inp{\vecU_1, \ldots, \vecU_{i-1}}$ and $(b)$ follows from \eqref{eq:type_property2} and \eqref{eq:type_property3}.
Suppose 
\begin{align}
R-\min_{P_{UU'}\in \cP}I(U;U')\leq-3\epsilon/4,\label{eq:R_bound}
\end{align} then 
\begin{align*}
\bbE\insq{\xi_i(\vecU_1, \ldots, \vecU_i)|\vecU_1, \ldots, \vecU_{i-1}}\leq|\cP|(n+1)^{|\cX|}\exp{\inb{n\inp{-3\epsilon/4}}}=:a'.
\end{align*}
 Let $t' = \frac{1}{2}\exp\inp{-\frac{n\epsilon}{2}}$. Then from Lemma~\ref{lemma:A1},

\begin{align}
&\pr\inp{\frac{1}{K}|\{i\in [1:K]: \exists\, P_{UU'}\in \cP\ \text{ s.t. }(\vecU_i,\vecU_j)\in \cT^n_{UU'}\text{ for some }j\in [1:K], \, j< i\}|\ge \frac{1}{2}\exp\left(-\frac{n\epsilon}{2}\right)}\nonumber\\
&= \pr\inb{\frac{1}{K}\sum_{i=1}^{K}\xi_i(\vecU_1, \ldots, \vecU_i)>t'}\nonumber\\
&\stackrel{(a)}{\leq} \exp\inb{-K\inp{t'-a'\log{e}}}\nonumber\\
&=\exp\inb{-K\inp{\frac{1}{2}\exp\inp{-\frac{n\epsilon}{2}}-|\cP|(n+1)^{|\cX|}\exp{\inp{\frac{-3n\epsilon}{4}}}\log{e}}}\nonumber\\
&\stackrel{(b)}{\leq}\exp\inb{-K\inp{\frac{1}{2}\exp\inp{-\frac{n\epsilon}{2}}-\frac{1}{4}\exp{\inp{\frac{n\epsilon}{4}}}\exp{\inp{\frac{-3n\epsilon}{4}}}}} \label{eq:another_n0_cond}\\
&=\exp\inb{-K\inp{\frac{1}{4}\exp\inp{-\frac{n\epsilon}{2}}}}\nonumber\\
&\stackrel{(c)}{\leq}\exp\inb{-\inp{\frac{1}{4}\exp\inp{\frac{n\epsilon}{2}}}},\nonumber
\end{align} where $(a)$ follows from \eqref{eq:A2}, $(b)$ holds for sufficiently large $n$ and $(c)$ uses $K=\exp\inp{nR}\geq \exp\inp{n\epsilon}$. This and \eqref{eq:e_2} imply \eqref{eq:e_22}.

By using \eqref{eq:E_1}, \eqref{eq:e_22} and taking union bound over $P_{UX}\in \cP^n\inp{\cU\times \cX}$ and $\vecx\in \cX^n$ {(recall that $|\cP^n\inp{\cU\times \cX}|$ and $|\cX^n|$ grow only polynomially in $n$)}, we can conclude that there exists a codebook of rate $R$ such that 
$\epsilon\leq R\leq \min_{P_{U U'}\in \cP}I(U;U')-3\epsilon/4$,   
 whose codewords are of type $P$ such that
\begin{align}
&\frac{1}{K}|\{i\in [1:K]: d(\vecu_i,\vecu_j)\leq 2\rad\text{ for some }j\in [1:K], \, j\neq i\}|< \exp\left(-\frac{n\epsilon}{2}\right) \label{eq:code1'}
\end{align}
and for every joint type $P_{UX}\in \cP^n\inp{\cU\times \cX}$ and $\vecx\in \cX^n$ satisfying $P_U= P$ and $\vecx \in \cT^n_{X}$,
\begin{align}
&|\{i\in [1:K]:(\vecu_i, \vecx)\in \cT^n_{UX}\}|\leq\exp{\left(n\left(\left|R-I(U;X)\right|^{+}+\epsilon\right)\right)}.\label{eq:code2'}
\end{align}
In order to obtain \eqref{eq:code1} from \eqref{eq:code1'}, we expurgate $\exp\left(-\frac{n\epsilon}{2}\right)$ fraction of codewords to obtain $d(\vecu_i,\vecu_j)> 2\rad$ for every pair of distinct codewords $\vecu_i$, $\vecu_j$. The new rate is 
\begin{align}
R'&= \frac{\log\inp{K-K\exp\left(-{n\epsilon/2}\right)}}{n}\nonumber\\
&=\frac{\log\inp{\exp\inp{nR}(1-\exp\left(-{n\epsilon/2}\right))}}{n}\nonumber\\
&=R+ \frac{\log\inp{1-\exp\left(-{n\epsilon/2}\right)}}{n}\nonumber\\
&\geq R -\epsilon/4 \text{ for sufficiently large $n$.}\label{eq:n0_1}
\end{align}
Let $n_0$ be such that \eqref{eq:E_1}, \eqref{eq:another_n0_cond} and \eqref{eq:n0_1} hold. Since, rate $R$ satisfies \eqref{eq:R_bound}, we have shown the existence of a codebook of rate $R'$ such that $\epsilon\leq R'\leq \min_{P_{U U'}\in \cP}I(U;U')-\epsilon$ and it satisfies \eqref{eq:code1} and \eqref{eq:code2} .

\end{proof}

\balance
\bibliography{IEEEabrv,bib}
\end{document}